\documentclass[reqno]{amsart}

\usepackage{amsmath}
\usepackage{amsfonts}
\usepackage{graphics}
\usepackage{epsfig}
\usepackage{amssymb}
\usepackage{amsthm}
\usepackage{amscd}
\usepackage[all]{xy}
\usepackage{latexsym}
\usepackage{graphicx}
\usepackage{multirow}
\usepackage{geometry}
\usepackage{color}
\usepackage{CJKutf8}

\newtheorem{thm}{Theorem}[section]
\newtheorem*{mainthm}{Main Theorem}

\newtheorem{lem}[thm]{Lemma}
\newtheorem{prop}[thm]{Proposition}
\newtheorem{conj}[thm]{Conjecture}

\theoremstyle{definition}
\newtheorem{Def}[thm]{Definition}
\newtheorem{rem}[thm]{Remark}

\newtheorem*{ack}{Acknowledgement}

\numberwithin{equation}{section}
\numberwithin{figure}{section}
\numberwithin{table}{section}

\def\rchi{{\hbox{\raise1.5pt\hbox{$\chi$}}}}

\def\PI{(P_{\rm I})}
\def\p{\partial}

\newcommand{\bea}{\begin{eqnarray}}
\newcommand{\eea}{\end{eqnarray}}
\newcommand{\be}{\begin{equation}}
\newcommand{\ee}{\end{equation}}

\newcommand{\bC}{{\mathbb{C}}}

\newcommand{\bZ}{{\mathbb{Z}}}

\newcommand{\Res}{\mathop{\rm Res}}

\begin{document}
\large
\setcounter{section}{0}

\title[2-parameter $\tau$-function of Painlev\'e I]
{2-parameter $\tau$-function
for the first Painlev\'e equation \\[+.5em]
---Topological recursion and direct monodromy problem \\
via exact WKB analysis---}

\author[K.\ Iwaki]{Kohei Iwaki}
\address{Graduate School of Mathematics\\
Nagoya University\\
Nagoya, 464-8602, Japan}
\email{iwaki@math.nagoya-u.ac.jp}


\dedicatory{To the memory of Tatsuya Koike}

\begin{abstract}
We show that a 2-parameter family of 
$\tau$-functions for the first Painlev\'e equation
can be constructed by the discrete Fourier transform 
of the topological recursion partition function
for a family of elliptic curves. 
We also perform an exact WKB theoretic computation 
of the Stokes multipliers of associated 
isomonodromy system assuming certain conjectures. 
\end{abstract}


\subjclass[2010]{Primary: 34M55; 81T45.
Secondary: 34M60; 34M56}

\keywords{First Painlev\'e equation;
topological recursion; 
isomonodoromic deformation; 
exact WKB analysis}

\maketitle

\allowdisplaybreaks

\tableofcontents

\setlength{\parskip}{0.1ex}

\section{Introduction}

{\em Painlev\'e transcendents} are remarkable special functions 
which appear in many areas of mathematics and physics. 
These are solutions of certain nonlinear ODEs 
known as {\em Painlev\'e equations} (\cite{Painleve}). 
Originally, there were six Painlev\'e equations 
(from Painlev\'e I to Painelv\'e VI), 
but now several generalizations of the Painlev\'e equations 
(($q$-)discrete, elliptic, higher order analogues) are discovered
(see \cite{FS, GJP, KNS, Noumi-Yanada, RGH, Sakai} for instance).
More than 100 years have passed since the discovery of 
the Painlev\'e equations, 
and lots of beautiful properties have been revealed through various approaches 
(isomonodromic deformation, Riemann-Hilbert method, 
affine-Weyl symmetry, the space of initial conditions, 
and so on). However, even for the classical six Painlev\'e equations,
there are still several open problems 
to be considered (see \cite{Clarkson2} for example).

A surprising development has been achieved in 2012 
by Gamayun-Iorgov-Lisovyy (\cite{GIL}).
They studied the {\em Painlev\'e VI equation} 
\begin{align}
\frac{d^{2} q}{dt^{2}} & = 
\frac{1}{2} \biggl( 
\frac{1}{q} + \frac{1}{q - 1} 
+ \frac{1}{q - t} \biggr) \cdot
\biggl( \frac{d q}{dt} \biggr)^{2} - 
\biggl(\frac{1}{t} + \frac{1}{t-1} + \frac{1}{q - t}  
\biggr) \cdot \frac{d q}{dt} \notag \\[+.2em]
& \quad  
+ \frac{2 q (q - 1) (q - t)}{t^{2}(t - 1)^{2}} \cdot
\left( 
\Bigl(\theta_{\infty} - \frac{1}{2} \Bigr)^2
-  \frac{\theta_{0}^2 \, t}{q^{2}}  
+  \frac{\theta_{1}^2 \, (t - 1)}{(q - 1)^{2}} 
-  \frac{\bigl( \theta_{t}^2 - \frac{1}{4} \bigr) \, 
    t (t - 1)}{(q - t)^{2}}
 \right), %
 \label{eq:Painleve-VI-intro}
\end{align}
and obtained the following formula 
for the {\em $\tau$-function} (\cite{JMU, Okamoto}) 
corresponding to 2-parameter general solutions of Painlev\'e VI: 
\begin{equation} \label{eq:tau-Painleve-VI-intro}
\tau_{P_{\rm VI}}(t, \vec{\theta}, \nu,\rho) =
\sum_{k \in {\mathbb Z}} e^{2 \pi i k \rho} \cdot 
B(t,  \vec{\theta}, \nu+k).
\end{equation}
Here $B$ is a 4-point Virasoro conformal block with $c=1$ 
which has an explicit combinatorial formula 
(see \cite[(1.8)--(1.12)]{GIL}); this is a consequence of the 
Alday-Gaiotto-Tachikawa correspondence \cite{AGT} 
and the Nekrasov's formula \cite{Nekrasov}. 
The constants $\vec{\theta} = (\theta_{0}, \theta_{1}, 
\theta_{t}, \theta_{\infty})$ are related to characteristic 
exponents of the associated isomonodromy system. 
The two parameters $\nu$ and $\rho$ are the integration constants
which specify initial conditions in solving Painlev\'e VI. 
It is very interesting that 
the general solution can be expressed as  
a discrete Fourier transform, 
and it is highly non-trivial 
from a glance of \eqref{eq:Painleve-VI-intro}.  

After the work \cite{GIL}, the Painlev\'e equations attract 
much more researchers in the community of the conformal field theory, 
gauge theory, topological string theory, and so on. 
Similar formulas to \eqref{eq:tau-Painleve-VI-intro} 
are also obtained or conjectured, for other Painlev\'e equations. 
For instance, Gamayun-Iorgov-Lisovyy also find a similar 
explicit formula for Painlev\'e V and III in \cite{GIL13}. 
Nagoya gave a conjectural expression of Painlev\'e II--V as 
a discrete Fourier transform of irregular conformal blocks 
in \cite{Nagoya1, Nagoya2} (see also \cite{LNR}). 
Bonelli-Lisovyy-Maruyoshi-Sciarappa-Tanzini (\cite{BLMST}) 
also gave conjectural expressions for Painlev\'e I, II, and IV 
as a discrete Fourier transform of a partition function of 
Argyres-Douglas theories. 
Regarding on higher order and $q$-analogues, we refer 
\cite{BS, BGT, Gav, GavIL, JNS, MN}.

On the other hand, Eynard-Orantin's {\it topological recursion}
(\cite{EO07}) is a recursive algorithm to compute the 1/N-expansion of 
the correlation functions and the partition function of matrix models 
from its spectral curve (c.f., \cite{CEO}), 
and it is generalized to any algebraic
curve which may not come from a matrix model. 
The topological recursion attracts both mathematicians and physicists 
since it is expected to encode the information of 
various enumerative and quantum invariants (e.g., Gromov-Witten invariants, 
Hurwitz numbers, Jones polynomials, and so on) in a universal way.
It is also related to integrable systems 
(\cite{Borot-Eynard}), the WKB analysis through 
the theory of {\em quantum curves} (\cite{BE16, DM14, Nor} etc.), 
and several topics in mathematics and physics. 
See the review article \cite{EO09} for further information.

Relationships between the topological recursion 
and Painlev\'e equations has been discussed in 
\cite{EO07, IS} for Painlev\'e I, 
\cite{IM} for Painlev\'e II, and 
\cite{IMS} for Painlev\'e I--VI.
In these articles, the Painlev\'e $\tau$-functions 
are constructed as the partition function of 
the topological recursion applied to spectral curves 
arising in the WKB analysis of associated isomonodromy systems. 
These results are generalized to a class of 
Painlev\'e hierarchies with rank 2 isomonodromy systems 
by \cite{MO19} recently. 

However, the Painlev\'e $\tau$-functions constructed by 
topological recursion are particular solutions. 
More precisely, such solutions admit a power series expansion 
(without any exponential terms). 
Then, it is easy to see that the expansion coefficients 
are recursively determined; that is, such solutions 
cannot contain any of integration constants  
(thus they are  called ``0-parameter solutions").
This reflects the fact that the all spectral curves 
considered in those works are of genus $0$. 
Hence, it is natural to ask:
\begin{quote}
Can we generalize the relationship 
between the Painlev\'e equations and the topological recursion
in such a way that we can handle the 2-parameter solutions 
like \eqref{eq:tau-Painleve-VI-intro} ? 
\end{quote}
A candidate of such a framework is proposed by 
Borot-Eyanrd \cite{Borot-Eynard} 
in a relationship between the topological recursion 
and integrable systems, where the discrete Fourier transform 
of the partition function is introduced. 
The discrete Fourier transform has also considered 
in an equivalent but different expression by Eynard-Marin\~o \cite{EM08}, 
which is called the {\em non-perturbative partition function}.
However, in these works, a precise relationship to  
Painlev\'e equations has not been discussed.

The purpose of this paper is, based on the ideas by  
\cite{Borot-Eynard, EM08}, to answer the above question 
affirmatively for the {\em Painlev\'e I equation}
\begin{equation}
\PI~:~ \hbar^{2} \frac{d^{2}q}{dt^{2}} = 
6q^{2} + t
\end{equation} 
with a formal parameter $\hbar$ being appropriately introduced. 
That is, we will construct a 2-parameter family of 
(formal) $\tau$-functions for Painlev\'e I 
by using the framework of the topological recursion 
and the quantum curves applied for a genus $1$ spectral curve.
We will also see that a Fourier series structure, 
which was conjectured in \cite{BLMST}, appears in the expression 
of the $\tau$-function, and in the WKB analysis of isomonodromy systems.

The precise statements of our main results are formulated as follows.
Let us consider a family of smooth elliptic curves 
(with a fixed generator $A, B$ of their first homology group)
\begin{equation}
y^2 = 4x^3 + 2t x + u(t,\nu),
\end{equation}
with $u(t,\nu)$ being chosen so that the $A$-period
\begin{equation}
\nu := \frac{1}{2 \pi i} \oint_{A} ydx 
\end{equation}
is independent of $t$. 
It is regarded as the spectral curve for 
the topological recursion via 
the standard Weierstrass parametrization:
\begin{eqnarray}
\begin{cases}
x =\wp(z; g_2, g_3) \\ y = \wp'(z; g_2, g_3)
\end{cases}
\quad \text{with} \quad g_2 = -2t,\quad  g_3 = - u(t,\nu).
\end{eqnarray} 
Let $W_{g,n}(z_{1},\dots,z_{n})$ and $F_g$
be the correlator of type $(g,n)$ and the $g$-th free energy
defined by the Eyanrd-Orantin topological recursion 
applied to the spectral curve, respectively 
(see \S \ref{subsection:TR} for the definition). 
Let us also define two formal series of $\hbar$ by
\begin{equation} \label{eq:tau-intro}
Z(t,\nu; \hbar) := \exp\Biggl( \sum_{g \ge 0} 
\hbar^{2g-2} F_g(t, \nu) \Biggr)
\end{equation}
and 
\begin{multline} \label{eq:WKBsol-intro}
\psi(x,t,\nu;\hbar) := \exp\Biggl(\, 
\frac{1}{\hbar} \int^{z(x)}_{0} W_{0,1}(z) + 
\frac{1}{2} \int^{z(x)}_{0} \int^{z(x)}_{0}
\left( W_{0,2}(z_{1},z_2) - \frac{dx(z_1) \cdot dx(z_2)}{(x(z_1)-x(z_2))^2}
\right) \\
+ \sum_{\substack{g \ge 0,\, n \ge 1 \\ 2g-2+n \ge 1}} 
\frac{\hbar^{2g-2+n}}{n!} \int^{z(x)}_{0} \cdots \int^{z(x)}_{0} 
W_{g,n}(z_{1},\cdots,z_{n}) \Biggr).
\end{multline}
Here $z(x)$ is the inverse function of $x=\wp(z)$. 
The formal series $Z$ is called the 
topological recursion partition function. 
The other formal series $\psi$ (which is called the wave function 
in literature of quantum curves)
is a WKB-type formal solution of 
a linear PDE (Theorem \ref{thm:BPZ-type-equation}).

Then our main results are formulated as follows: 

\smallskip
\begin{mainthm}[Theorem \ref{thm:tau-function-theorem} 
and Theorem \ref{thm:wave-function-theorem}] 
\hspace{+.1em}
\begin{itemize}
\item[(i)] 
The formal series
\begin{equation} \label{eq:main-1-intro}
\tau_{P_{\rm I}}(t, \nu, \rho; \hbar) := 
\sum_{k \in \bZ} e^{{2 \pi i k \rho}/{\hbar}}  \cdot
Z(t, \nu+k\hbar; \hbar)
\end{equation}
is a formal $\tau$-function of $\PI$.
In other words, the formal series
\begin{equation}
q(t,\nu,\rho;\hbar) := - \hbar^2 \frac{\p^2}{\p t^2} 
\log \tau_{P_{\rm I}}(t, \nu, \rho; \hbar)
\end{equation}
satisfies $(P_{\rm I})$. 

\smallskip
\item[(ii)] The formal series 
\begin{equation} \label{eq:main-2-intro}
\Psi(x,t, \nu, \rho; \hbar) := 
\frac{\sum_{k \in \bZ} e^{{2 \pi i k \rho}/{\hbar}} \cdot
Z(t,\nu+k \hbar; \hbar) \cdot \psi(x,t,\nu + k \hbar; \hbar)}
{\sum_{k \in \bZ} e^{{2 \pi i k \rho}/{\hbar}} \cdot
Z(t,\nu+k \hbar; \hbar)}
\end{equation}
is a formal solution of the isomonodrmy system 
\begin{align}
\label{eq:JM-1-intro}
(L_{\rm I}) & : ~~~ \left[ \hbar^2 \frac{\p^2}{\p x^2} - 
\frac{\hbar}{x-q} \cdot \Bigl( \hbar \frac{\p}{\p x} - p \Bigr)
- (4x^3 + 2t x + 2 H) \right] \Psi = 0, \\
\label{eq:JM-2-intro}
(D_{\rm I}) & : ~~ \left[ \hbar \frac{\p}{\p t} - \frac{1}{2(x-q)} \cdot
\Bigl( \hbar \frac{\p}{\p x} - p \Bigr)
\right] \Psi = 0
\end{align}
associated with $\PI$. 

\end{itemize}
\end{mainthm}

We also present an approach to the direct monodromy problem
(i.e., computation of the  Stokes multipliers 
of $(L_{\rm I})$)  
of the isomonodromy system in the spirit of the 
{\em exact WKB analysis} (\cite{KT05, Voros83}). 
We are motivated by Takei's approach to 
the computation of Stokes multipliers (\cite{Takei}). 
The resulting Stokes multipliers are expressed 
by the two parameters $\nu$ and $\rho$ appearing 
in \eqref{eq:main-1-intro} in a quite explicit manner, 
and satisfy the desired relation (i.e., the cyclic relations 
of Stokes multipliers). 
However, unfortunately, our formulas of Stokes multipliers 
are not proved rigorously 
because they are based on two conjectures 
(on Borel summability and Voros connection formula)
for WKB solutions of the PDE \eqref{eq:BPZ-equation}. 
We hope to solve the conjectures and show the validity 
of our method in the future.

\smallskip
Before ending the introduction, 
let us make comments on the previous works
which are related to our main result.

\begin{enumerate}
\item 
In the context of the WKB analysis for 
Painlev\'e equations, another kind of 2-parameter formal solution
is constructed by Aoki-Kawai-Takei (\cite{AKT-P}). 
A similar solution is also constructed when $\hbar = 1$ 
by Yoshida (\cite{Yoshida}), 
Garoufalidis-Its-Kapaev-Marin\~o (\cite{GIKM}) 
and Aniceto-Schiappa-Vonk (\cite{ASV}). 
We do not understand the precise relationship 
between these two formal solutions.
(We will come back to this point in \S \ref{section:conclusion}).

\item 
Borot-Eynard (\cite{Borot-Eynard}) conjectured that 
the non-perturbative partition function 
(i.e., \eqref{eq:main-2-intro}) satisfies a certain 
Hirota-type bilinear identity in a more general situation, 
and they verified the relation at the leading order.

\item 
As is mentioned earlier, the existence of the discrete 
Fourier transformed expression of $\tau$-function is conjectured 
by \cite{BLMST} for Painlev\'e I (and also for other 
Painlev\'e equations with an irregular singular isomonodromy system). 
Our results give a full order proof of their claim 
for Painlev\'e I case.


\item 
Grassi-Gu (\cite{GG}) investigated the discrete Fourier transformed 
expression of the $\tau$-function for Painlev\'e II
from the viewpoint of matrix models. 

\item 
For Painlev\'e VI,  
an approach to the direct monodromy problem based on the spectral network 
(c.f., \cite{GMN12}) is presented in the recent article \cite{CPT} 
by Coman-Pomoni-Teschner. 
The spectral network and the Stokes graph in the exact WKB analysis 
are essentially the same objects (at least in the rank 2 cases).
The spectral network used in \cite{CPT} is Fenchel-Nielsen type, 
while the one will use in \S \ref{section:Exact-WKB} 
is the Fock-Goncharov type, in the sense of \cite{HN}. 
The author does not understand the relationship 
between these two approaches. 

\end{enumerate}

This paper is organized as follows.  
In \S \ref{section:Painleve-I} and \S \ref{section:TR} 
we briefly review some definitions and known facts 
on Painlev\'e I and the topological recursion, respectively. 
A key result (Theorem \ref{thm:BPZ-type-equation}) is also 
given in \S \ref{section:TR}.
Our main theorem will be formulated and proved in \S \ref{section:main-result}.  
The exact WKB theoretic computation of the Stokes multipliers 
is given in \S \ref{section:Exact-WKB}. 
Open questions and several (possibly) related topics are 
also discussed in \S \ref{section:conclusion}.

\begin{ack}
The author is grateful to 
Alba Grassi, 
Akishi Ikeda, 
Nikolai Iorgov,
Michio Jimbo, 
Akishi Kato, 
Taro Kimura,
Tatsuya Koike, 
Oleg Lisovyy, 
Motohico Mulase, 
Hajime Nagoya, 
Hiraku Nakajima, 
Ryo Ohkawa, 
Nicolas Orantin,
Hidetaka Sakai, 
Kanehisa Takasaki, 
Yoshitugu Takei,
Yumiko Takei
and
Yasuhiko Yamada 
for many valuable comments and discussion.
This work was supported by the JSPS KAKENHI Grand Numbers 
16K17613, 16H06337, 16K05177, 17H06127, and
JSPS and MAEDI under the Japan-France Integrated Action Program (SAKURA).

\smallskip
We dedicate the paper to the memory of Tatsuya Koike, 
who made a lot of important contributions to the theory 
of the exact WKB analysis and the Painlev\'e equations. 
He inspired us by his beautiful papers, talks, 
and private communications on various occasions.
\end{ack}

\section{Brief review of the first Painlev\'e equation $\PI$}
\label{section:Painleve-I}

The {\em Painlev\'e I equation} with a formal parameter $\hbar$ 
is the following non-linear ODE of second order:
\begin{equation}
\PI~:~ \hbar^{2} \frac{d^{2}q}{dt^{2}} = 
6q^{2} + t.
\end{equation} 
Note that the parameter $\hbar$ can be introduced to 
the $\hbar=1$ case by the rescaling 
$(q,t) \mapsto (\hbar^{-2/5}q, \hbar^{-4/5}t)$ of variables.
The parameter $\hbar$ plays no role in this section, 
however, it is crucially important in our construction 
of 2-parameter formal solution of $\PI$ presented in subsequent sections. 

This short section is devoted to a review of $\PI$. 
See \cite{Clarkson, FIKN} etc. for details and further references.

\subsection{Hamiltonian system and the $\tau$-function}

Each of the Painlev\'e equations can be written as a Hamiltonian system (\cite{Okamoto}). 
For Painlev\'e I, it is given by
\begin{equation} \label{eq:Hamiltonian-system}
\hbar \frac{dq}{dt} = \frac{\p H}{\p p}, \quad
\hbar \frac{dp}{dt} = - \frac{\p H}{\p q} 
\end{equation}
with the (time-dependent) Hamiltonian 
\begin{equation}
H := \frac{1}{2}p^2 - 2 q^3 - t q.
\end{equation}
For any solution $(q,p) = (q(t,\hbar), p(t,\hbar))$ 
of the Hamiltonian system \eqref{eq:Hamiltonian-system}, 
let us denote by
\begin{equation} 
h(t,\hbar) := H(t, q(t,\hbar), p(t,\hbar))
\end{equation}
the Hamiltonian function. It is easy to see that 
\begin{equation} \label{eq:Hamiltonian-function-and-q}
\frac{dh}{dt}(t,\hbar) = - q(t,\hbar).
\end{equation}

\begin{Def}[c.f., \cite{JM2, JMU, Okamoto}]
The {\em $\tau$-functon} 
corresponding to a solution $q(t, \hbar)$ of $\PI$ 
is defined (up to a multiplicative constant) by 
\begin{equation} \label{eq:def-of-tau-function-PI}
\tau_{P_{\rm I}}(t, \hbar) := \exp\left( 
\frac{1}{\hbar^2} \int^{t} h(t', \hbar) \, dt' \right).
\end{equation} 
\end{Def}

The equality \eqref{eq:Hamiltonian-function-and-q} shows 
that the solution $q(t)$ of $\PI$ is recovered from 
the $\tau$-function by 
\begin{equation} \label{eq:tau-and-q-in-Painleve-I}
q(t, \hbar) = - \hbar^2 \frac{\p^2}{\p t^2} \log \tau_{P_{\rm I}}(t,\hbar).
\end{equation}
Substituting the expression into $\PI$, 
we can derive an ODE satisfied by the $\tau$-function. 
It can be written in a Hirota-type bilinear equation
(c.f., \cite{Okamoto2}): 
\begin{equation} \label{eq:Hirota-Painleve-1}
\hbar^4 D^{4}_{t} \, \tau_{P_{\rm I}} \cdot \tau_{P_{\rm I}} 
+ 2 t \, \tau_{P_{\rm I}} \cdot \tau_{P_{\rm I}}  = 0.
\end{equation}
Here 
\begin{equation}
D_t^{k} f \cdot g := \sum_{\ell=0}^{k} (-1)^{\ell} \cdot 
{{k}\choose{\ell}} \cdot 
\frac{\partial^{\ell} f}{\partial t^\ell} 
\cdot 
\frac{\partial^{k-\ell} g}{\partial t^{k- \ell}} 
\end{equation}
is the Hirota derivative. 
The equality \eqref{eq:tau-and-q-in-Painleve-I} 
(or \eqref{eq:Hirota-Painleve-1}) 
is also used as the defining equation for the $\tau$-function.

\begin{rem}
For any fixed $\hbar \in {\mathbb C}^{\ast}$, 
any solution of the Hamiltonian system \eqref{eq:Hamiltonian-system}
is known to be meromorphic function of $t$ on ${\mathbb C}$ 
(the Painlev\'e property). 
Near any pole $a \in {\mathbb C}$, we can show that 
\begin{equation}
q(t, \hbar) = \frac{\hbar^2}{(t-a)^2} - \frac{a}{10 \hbar^2}(t-a)^2 
+ \cdots, \qquad
p(t, \hbar) = - \frac{2 \hbar^3}{(t-a)^3} - \frac{a}{5 \hbar}(t-a)^2 + \cdots.
\end{equation}
Therefore, the Hamiltonian function
$h(t)$ behaves as 
\begin{equation}
h(t, \hbar) = \frac{\hbar^2}{t-a} + \cdots
\end{equation}
when $t \to a$. 
This property implies that the $\tau$-function
$\tau_{P_{\rm I}}(t, \hbar)$ has a simple zero 
at the pole $a$ of $q(t, \hbar)$, and hence, 
$\tau_{P_{\rm I}}$ is an entire function of $t$ 
(although the solution of $\PI$ is meromorphic). 
This is the well-known analogy between the Painlev\'e functions 
and Weierstrass elliptic functions: 
\[
\text{
``$(q,h,\tau_{P_{\rm I}})$ $\leftrightarrow$ 
$(\wp, \zeta, \sigma)$"}.
\]
(See Appendix \ref{appendix:Weierstrass} for the Weierstrass functions.)
Hone-Zullo (\cite{Hone-Zullo}) 
discussed the analogy at the level of series expansions in more detail.
\end{rem}

\subsection{Isomonodoromy system}

Each of the Painlev\'e equations can be 
written as a compatibility condition of 
a system of linear PDEs, which is called 
the {\em isomonodromy system} 
(or the {\em Lax pair}). 
For $\PI$, such a pair is given by 
\begin{align}
\label{eq:JM-1}
(L_{\rm I}) & : ~~~ \left[ \hbar^2 \frac{\p^2}{\p x^2} - 
\frac{\hbar}{x-q} \cdot \Bigl( \hbar \frac{\p}{\p x} - p \Bigr)
- (4x^3 + 2t x + 2 H) \right] \Psi = 0, \\
\label{eq:JM-2}
(D_{\rm I}) & : ~~ \left[ \hbar \frac{\p}{\p t} - \frac{1}{2(x-q)} \cdot
\Bigl( \hbar \frac{\p}{\p x} - p \Bigr)
\right] \Psi = 0.
\end{align}
(C.f., \cite[Appendix C]{JM2} and Remark \ref{eq:matrix-form} below).

Thanks to the compatibility condition, we can find a fundamental 
system of solutions of the system $(L_{\rm I})$--$(D_{\rm I})$.
It is, in particular, a fundamental system of solutions 
of $(L_{\rm I})$, which is an ODE with an irregular singular point 
of Poincar\'e rank $5/2$ at $x= \infty$. 
(The point $x=q$ is a so-called apparent singularity of $(L_{\rm I})$; 
that is, solutions of $(L_{\rm I})$ have no singularity at the point.)
The deformation equation $(D_{\rm I})$ in $t$-direction 
guarantees that the Stokes multipliers around $x = \infty$
are independent of $t$ (and this is the reason why 
the system $(L_{\rm I})$--$(D_{\rm I})$ is called ``isomonodromic"). 

The Stokes multipliers of $(L_{\rm I})$ are the first integrals 
of the corresponding solution of $(P_{\rm I})$.
Actually, since the Stokes multipliers must satisfy a 
cyclic relation (which guarantees the single-valuedness 
of the fundamental solution of $(L_{\rm I})$; 
we will make it more precise in \S \ref{subsection:computation-of-Stokes}), 
only two Stokes multipliers are independent. 
Through the Riemann-Hilbert correspondence, 
these two Stokes multipliers are regarded as the 
integration constants in solving $\PI$. 
The Stokes multipliers have been
effectively used to study the connection problems 
of Painlev\'e transcendents
(see \cite{FIKN, Kap, Kap-Kit, LR, Takei} for example). 

Our main result, which will be presented in subsequent sections, 
is a construction of 2-parameter family of 
a formal series valued $\tau$-function for $\PI$ 
(i.e., a formal solution of \eqref{eq:Hirota-Painleve-1}) 
via the topological recursion. 
We will also construct a fundamental system 
of formal solutions of the isomonodromy system 
$(L_{\rm I})$--$(D_{\rm I})$ associated with $\PI$.

\begin{rem} \label{eq:matrix-form}
The isomonodormy system $(L_{\rm I})$--$(D_{\rm I})$
can be written in a matrix form
\begin{equation} \label{eq:Lax-matrix-form}
\hbar \frac{\p}{\p x} 
\begin{pmatrix}  \Psi \\ \Phi \end{pmatrix} 
 = 
\begin{pmatrix}
p & 4 \,(x-q) \\ x^2 + xq + q^2 + {t}/{2} & -p
\end{pmatrix}
\cdot 
\begin{pmatrix}  \Psi \\ \Phi \end{pmatrix}, 
\quad
\hbar \frac{\p}{\p t} 
\begin{pmatrix}  \Psi \\ \Phi \end{pmatrix} 
 = 
\begin{pmatrix}
0 & 2 \\ x + {q}/{2} & 0
\end{pmatrix}
\cdot 
\begin{pmatrix}  \Psi \\ \Phi \end{pmatrix} 
\end{equation}
by introducing
\begin{equation}
\Phi := \frac{1}{4 \, (x-q)} \cdot \Bigl( \hbar \frac{\p}{\p x} - p \Bigr) \Psi
= \frac{\hbar}{2}\,  \frac{\p}{\p t} \Psi.
\end{equation}
The system \eqref{eq:Lax-matrix-form} of PDEs 
is equivalent to the one given in \cite[Appendix C]{JM2}. 
This also shows that the point $x=q$ is an 
apparent singular point of $(L_{\rm I})$.
\end{rem}

\section{Topological recursion for a family of elliptic curves}
\label{section:TR}

In this section, we will discuss the properties of 
correlators and free energies of the topological recursion 
defined from a family of genus 1 spectral curves. 
Our spectral curve is modelled on the ``classical limit" 
of the differential equation $(L_{\rm I})$ in the 
isomonodromy system associated with $\PI$.
We will use the Weierstrass elliptic function 
to give a meromorphic parametrization of the curve. 
The definition and several properties of 
Weierstrass elliptic functions and $\theta$-functions 
are summarized in Appendix \ref{appendix:Weierstrass}.

\subsection{A family of elliptic curves}
\label{section:sp-curve}

Let us consider a family of elliptic curves ${\mathcal C}_{t,u}$
defined by the following algebraic relation:
\begin{equation} 
y^2 = 4x^3 + 2t x + u. 
\end{equation}
Here, $t$ is a parameter
which eventually becomes the isomonodromic time for the Painlev\'e equation, 
and $u$ is also a parameter which will be replaced by a certain function of 
$t$ and an additional parameter $\nu$ defined as follows. 

Let us take $(t_{\ast}, u_{\ast}) \in {\mathbb C}^2$ 
so that the discriminant $\Delta = -16 (8 t^3 + 27 u^2)$ 
does not vanish on a neighborhood $D_{t_\ast} \times D_{u_\ast} \subset {\mathbb C}^2$
of $(t_{\ast}, u_{\ast})$. Taking a sufficiently small neighborhood, 
we may assume that $D_{t_\ast} \times D_{u_\ast}$ is simply connected. 
Then we can identify the first homology groups $H_1({\mathcal C}_{t,u}, {\mathbb Z})$ 
of ${\mathcal C}_{t,u}$ for $(t,u) \in D_{t_\ast} \times D_{u_\ast}$
with that
$H_1({\mathcal C}_{t_\ast, u_\ast}, {\mathbb Z})$ 
of ${\mathcal C}_{t_\ast,u_\ast}$ by the parallel transport.

To apply the Eynard-Orantin's topological recursion,
let us fix a generator $A, B \in H_1({\mathcal C}_{t_\ast,u_\ast}, {\mathbb Z})$
whose intersection paring is given by $\langle A, B \rangle = +1$. 
We denote by the same symbol $A, B$ 
the generator of $H_1({\mathcal C}_{t,u}, {\mathbb Z})$ 
for $(t,u) \in D_{t_\ast} \times D_{u_\ast}$ obtained by 
the canonical isomorphism  $H_1({\mathcal C}_{t,u}, {\mathbb Z}) \simeq 
H_1({\mathcal C}_{t_\ast,u_\ast}, {\mathbb Z})$. 
Since 
\begin{equation}
\frac{\p}{\p u} \oint_A \sqrt{4x^3+2tx+u} \, dx 
\ne 0 
\quad ((t,u) \in D_{t_\ast} \times D_{u_\ast}),
\end{equation}
the implicit function theorem guarantees the existence of a neighborhood 
$D_{\nu_\ast} \subset {\mathbb C}$ of  
\begin{equation} \label{eq:nu}
\nu_\ast := 
\frac{1}{2 \pi i} \oint_{A} \sqrt{4x^3 + 2 t_\ast x + u_\ast} \, dx \in {\mathbb C}
\end{equation}
such that 
there exists a holomorphic function $u(t,\nu)$ on $D_{t_\ast} \times D_{\nu_\ast}$   
satisfying  
\begin{equation} \label{eq:nu}
u(t_\ast, \nu_\ast) = u_\ast, \quad 
\nu = \frac{1}{2 \pi i} \oint_{A} \sqrt{4x^3 + 2 t x + u(t,\nu)} \, dx
\quad ((t,\nu) \in D_{t_\ast} \times D_{\nu_\ast}).
\end{equation}

In what follows, we consider the family of elliptic curve  
parametrized by $(t,\nu) \in D_{t_\ast} \times D_{\nu_\ast}$
defined by the algebraic relation 
\begin{equation} \label{eq:spcurve}
y^2 = 4x^3 + 2t x + u(t,\nu).
\end{equation}
In other words, we will consider the family of elliptic curves 
with a fixed $A$ and $B$-cycles so that the $A$-period
$\nu$ of $ydx$ (given in \eqref{eq:nu})
is independent of $t$. 

We will also use the standard notations 
\begin{equation}
\omega_\ast = \omega_\ast(t,\nu) := \oint_\ast \frac{dx}{\sqrt{4x^3+2tx+u(t,\nu)}}, \quad
\eta_\ast  = \eta_\ast(t,\nu) := - \oint_\ast \frac{x \, dx}{\sqrt{4x^3+2tx+u(t,\nu)}}
\end{equation}
($\ast \in \{A, B\}$) for the periods of the elliptic curve.
The $t$-independence of $\nu$ requires the following system of PDEs 
\begin{eqnarray}
\label{eq:PDE-u} 
\frac{\p u}{\p t} = 2 \, \frac{\eta_A}{\omega_A}, \quad 
\frac{\p u}{\p \nu} = \frac{4 \pi i}{\omega_A}
\end{eqnarray}
for $u$, which is in fact compatible.

\subsection{Topological recursion}
\label{subsection:TR}

\subsubsection{Spectral curve}
A {\em spectral curve} is 
a triple $(\Sigma, x,y)$ of the compact Riemann surface $\Sigma$ 
with a prescribed symplectic basis of its first homology group, 
and two meromorphic functions $x,y$ on $\Sigma$ such that 
$dx$ and $dy$ never vanish simultaneously (\cite{EO07}). 
We consider the spectral curve 
\begin{equation}\label{eq:par-rep} 
(\Sigma = {\mathbb C}/\Lambda, x(z) = \wp(z), y(z) = \wp'(z))
\end{equation}
which is a parametrization of the elliptic curve \eqref{eq:spcurve}. 
Here $\wp(z) = \wp(z; g_2, g_3)$ is the Weierstrass $\wp$-function
with $g_2=-2t$ and $g_3 = - u(t,\nu)$, which is doubly-periodic 
with periods $\omega_A$ and $\omega_B$. 
$\Lambda = \bZ \cdot \omega_A + \bZ \cdot \omega_B$ is the lattice 
generated by the periods of the elliptic curve \eqref{eq:spcurve}.
See Appendix \ref{appendix:Weierstrass} for the 
definition and several properties of $\wp(z)$ 
(we omit the $t$ and $\nu$ dependence for simplicity). 
Since \eqref{eq:par-rep} is a (meromorphic) 
parametrization of the elliptic curve \eqref{eq:spcurve}, 
we also call \eqref{eq:spcurve} the spectral curve below.

Let $z_o \in {\mathbb C}$ be a generic point, 
and $\Omega$ be the quadrilateral with 
$z_o$, $z_o+\omega_A$, $z_o+\omega_B$ and $z_o+\omega_A+\omega_B$
on its vertices; that is, a fundamental domain of $\bC/\Lambda$. 
The ramification points (i.e., zeros of $dx$) on $\Omega$ 
are given by the half-periods 
$r_1 \equiv \omega_A/2$, $r_2 \equiv \omega_B/2$ and $r_3 \equiv (\omega_A+\omega_B)/2$
modulo $\Lambda$. These points correspond to the branch points 
$e_i = x(r_i)$ ($i=1,2,3$) of the elliptic curve which are defined by
$4x^3+2tx+u = 4(x-e_1)(x-e_2)(x-e_3)$.
The Galois involution $y \mapsto -y$ of the spectral curve 
\eqref{eq:par-rep} is realized by
$z \mapsto \bar{z} = -z$ 
(or $\bar{z} \equiv -z$ mod $\Lambda$)
since $\wp$ is even function of $z$.

We may assume that $z=0$ (which corresponds to $x=\infty$) 
is contained in $\Omega$ by translation. 
In the following, after fixing a branch cut on $x$-plane,
we will use the inverse function $z(x)$ of $x=\wp(z)$, 
which is given by the elliptic integral \eqref{eq:inverse-function-z}.
We fix the branch so that 
\begin{equation}
z(x) \sim -  x^{- \frac{1}{2}} \cdot 
\bigl( 1+O(x^{-\frac{1}{2}}) \bigr)
\end{equation}
holds when $x \to \infty$. 
This is equivalent to say that 
we have chosen the branch so that 
\begin{equation} \label{eq:branch-of-y}
y\bigl(z(x)\bigr) = \sqrt{4x^3+2tx+u(t,\nu)} = 
+ 2 \, x^{\frac{3}{2}} \cdot \bigl( 1+O(x^{-\frac{1}{2}}) \bigr).
\end{equation}

\subsubsection{Bergman kernel}
The Bergman kernel normalized along the A-cycle is given as follows:
\begin{equation}
B(z_1, z_2) :=  \Bigl( \wp(z_1 - z_2) 
+ \frac{\eta_A}{\omega_A} \Bigr) \cdot dz_{1}dz_2.
\end{equation}
This is characterized by the following properties: 
\begin{itemize}
\item $B(z_1, z_2)$ is a meromorphic bi-differential with double poles 
along the diagonal $z_1 \equiv z_2$ modulo $\Lambda$. 
\item $B(z_1, z_2)$ is symmetric: $B(z_1, z_2) = B(z_2, z_1)$. 
\item Integrals along $A$- and $B$-cycle are 
\begin{equation} \label{eq:normalization-Bergman-kernel}
\oint_{z_1 \in {A}} B(z_1,z_2) = 0, \quad
\oint_{z_1 \in {B}} B(z_1,z_2) = 
- \frac{\omega_A \cdot \eta_B - \omega_B \cdot \eta_A}{\omega_{A}} \cdot dz_2
= \frac{2 \pi i}{\omega_{A}} \cdot dz_2. 
\end{equation}
These properties follows from the facts that 
\begin{equation} \label{eq:Pz-function}
P(z) := - \zeta(z) + \frac{\eta_A}{\omega_A} \cdot z
\end{equation}
solves $(d_{z_1} P(z_1-z_2)) \cdot dz_2 = B(z_1, z_2)$, 
and the quasi-periodicity \eqref{eq:zeta-periods}
of $\zeta$-function. 

\end{itemize}

\subsubsection{Definition of correlators}
\label{subsection:def-of-correlators}

The Eynard-Orantin's
topological recursion for the spectral curve is 
formulated as follows.

\begin{Def}[{\cite[Definition 4.2]{EO07}}]
The {\em correlator (or the Eynard-Orantin differential)} 
$W_{g,n}(z_{1},\dots,z_{n})$ 
of type $(g,n)$ is a meromorphic multi-differential 
on the $n$-times product of $\Sigma$ 
defined by the following {\em topological recursion relation}:
\begin{equation}
W_{0,1}(z_{1}) :=  y(z_{1}) \cdot dx(z_{1}), \quad
W_{0,2}(z_{1},z_{2}) := 
B(z_1, z_2), 
\end{equation}
and for $2g-2+n \ge 1$, we define
\begin{align} 
& 
W_{g,n}(z_{1},\dots,z_{n}) := 
\frac{1}{2\pi i} \sum_{j=1}^{3}
\oint_{z \in \gamma_{j}} K(z_1,z) \nonumber \\
& \quad 
 \times  \Biggl[ \sum_{j=2}^{n} 
\left( W_{0,2}(z,z_{j}) \cdot W_{g,n-1}(\bar{z}, z_{[\hat{1}, \hat{j}]}) 
+ W_{0,2}(\bar{z},z_{j}) \cdot W_{g,n-1}(z, z_{[\hat{1}, \hat{j}]}) 
\right) \nonumber \\
& \quad 
+ W_{g-1,n+1}(z,\bar{z},z_{[\hat{1}]}) + 
\sum_{\substack{g_{1}+g_{2}=g \\ I \sqcup J = [\hat{1}]}}^{\text{stable}} 
W_{g_{1}, |I|+1}(z,z_{I}) \cdot W_{g_{2}, |J|+1}(\bar{z}, z_{J}) \Biggr]. 
\label{eq:top-rec}
\end{align}
Here $\gamma_{j}$ is a small cycle (in $z$-plane) which encircles 
the branch point $z=r_j$ ($j=1,2,3$) in counter-clockwise direction, 
and the {\em recursion kernel} $K(z_1,z)$ is given by 
\begin{equation} \label{eq:recursion-kernel}
K(z_1,z) = 
\frac{1}{2(y(z) - y(\bar{z})) \cdot dx(z)} \cdot 
\int^{w=z}_{w=\bar{z}} W_{0,2}(z_{1}, w) 
~~\left( = 
- \frac{P(z_1-z) - P(z_1 - \bar{z})}{2(y(z) - y(\bar{z})) \cdot dx(z)}
\right).
\end{equation}
Also, we use the index convention 
$[\hat{j}] = \{1, \dots, n \}\setminus\{j\}$ etc., 
and the sum in the third line in \eqref{eq:top-rec} is taken 
for indices in the stable range 
(i.e., only $W_{g,n}$'s with $2g - 2 + n \ge 1$ appear).
\end{Def}

The following properties of $W_{g,n}$ 
were established in \cite{EO07}:
\begin{itemize}
\item 
$W_{g,n}$ is invariant under any permutation of variables.

\item 
As a differential in each variable $z_i$, 
$W_{g,n}$ with $2g-2+n \ge 1$
is holomorphic at any point except for  
the ramification points $r_1, r_2, r_3$.
In particular, they are holomorphic at $z_i=0$.

\item 
$W_{g,n}$ is doubly periodic in each variable: 
\[
W_{g,n}(z_1,\dots, z_i + \omega, \dots,z_n) 
= W_{g,n}(z_1,\dots,z_i,\dots,z_n)
\quad (i=1,\dots,n,~\omega \in \Lambda).
\]

\item 
Except for $W_{0,1}$, they are normalized along the $A$-cycle:
\begin{equation} \label{eq:trivial-A-period-Wgn}
\oint_{z_i \in A} W_{g,n}(z_1, \cdots, z_i, \cdots, z_n) = 0
\qquad ((g,n) \ne (0,1),~i=1,\dots,n).
\end{equation}

\item 
$W_{g,n}$ also depends holomorphically on the parameters 
$(t,\nu)$ on $D_{t_\ast} \times D_{\nu_\ast}$.
Moreover, there are formulas which describe 
the differentiations of $W_{g,n}$ with respect to $t$ and $\nu$. 
We summarize the formulas in \S \ref{section:variation-formula}.

\end{itemize}

\subsubsection{Definition of the free energy}

Here we also recall the definition of the genus $g$ free energy 
$F_g$ introduced in \cite{EO07}.

\begin{Def}[{\cite[\S 4.2]{EO07}}] ~
\begin{itemize} 
\item
The {\it genus $0$ free energy} $F_0$ is defined in 
\cite[\S 4.2.2]{EO07}. In our case, it is given by 
\begin{equation}
F_0(t,\nu) := \frac{t \, u(t,\nu)}{5} + 
\frac{\nu}{2} \cdot \oint_B \sqrt{4x^3+2tx+u(t,\nu)} \, dx.
\end{equation}

\item
The {\it genus $1$ free energy} $F_1$ is also defined in \cite[\S 4.2.3]{EO07}
up to a multiplicative constant. We employ  
\begin{equation}
\label{eq:expression-F1}
F_1(t,\nu) := - \frac{1}{12} \log \Bigl(\omega_A(t,\nu)^6 \cdot \Delta(t,\nu) \Bigr)
\end{equation}
as the definition. 
Here $\Delta = 16(e_1 - e_2)^2(e_2-e_3)^2(e_3-e_1)^2$ is the discriminant 
of \eqref{eq:spcurve}. 

\item 
The {\it genus $g$ free energy} $F_{g}$ for $g \ge 2$ is defined by
\begin{equation}
F_{g}(t,\nu) := \frac{1}{2\pi i (2-2g)} \cdot 
\sum_{j=1}^{3} \oint_{\gamma_{j}} \Phi(z) \cdot W_{g,1}(z),
\end{equation}
where 
\begin{equation}
\Phi(z) := \int^{z}_{z_{o}} y(z) \cdot dx(z) 
\end{equation}
with an arbitrary generic point $z_{o}$.

\end{itemize}
\end{Def}

$F_g$'s are holomorphic functions of $(t,\nu)$ on $D_{t_\ast} \times D_{\nu_\ast}$.
The generating function 
\begin{equation} \label{eq:free-energy-F}
F(t,\nu;\hbar) := \sum_{g  \ge 0} \hbar^{2g-2} F_g(t,\nu)
\end{equation}
is also called the {\em free energy} of 
the spectral curve \eqref{eq:spcurve}.
We call its exponential 
\begin{equation} \label{eq:partition-function-Z}
Z(t,\nu;\hbar) := \exp\bigl(F(t,\nu;\hbar) \bigr)
= \exp\Biggl( \sum_{g \ge 0} \hbar^{2g-2} F_g(t,\nu) \Biggr) 
\end{equation}
the {\em (topological recursion) partition function}.

\bigskip
\begin{lem}[{C.f., \cite{EO07}}] ~
The genus $0$ free energy satisfies 
\begin{equation}
\label{eq:dF0-dt}
\frac{\p F_0}{\p t}(t,\nu)  = 
\frac{u(t,\nu)}{2}, \quad
\frac{\p F_0}{\p \nu}(t,\nu)  = 
\oint_B \sqrt{4x^3+2tx+u(t,\nu)} \, dx.
\end{equation}

\end{lem}
The second equality is called the Seiberg-Witten relation.

\subsection{Differentiation formula}
\label{section:variation-formula}

There are formulas which allow us to compute derivatives 
of $W_{g,n}$ and $F_g$ with respect to the parameters $t$ and $\nu$. 

\subsubsection{Differentiation with respect to $t$}
The derivatives of correlators and free energies 
with respect to $t$ are given as follows: 

\begin{prop} [C.f., {\cite[Theorem 5.1]{EO07}}] 
\label{prop:variation} ~ 
\begin{itemize}
\item[(i)] %
For $2g-2+n \ge 0$:
\begin{equation} \label{eq:variation-Wgn}
\frac{\p}{\p t} W_{g,n}(z(x_{1}), \dots, z(x_{n})) 
= \Res_{z_{n+1}=0} \frac{
W_{g,n+1}\bigl( z(x_{1}),\dots,z(x_{n}), z_{n+1} \bigr)}
{z_{n+1}}.
\end{equation}
\item[(ii)] %
For $g \ge 1$: 
\begin{equation} \label{eq:dt-closed-Fg}
\frac{\p F_{g}}{\p t}(t,\nu)
=  \Res_{z=0} \frac{W_{g,1}(z)}{z}. 
\end{equation}
\end{itemize}
\end{prop}

\begin{proof}
A direct computation shows
\[
\frac{\partial y}{\partial t}(z) \cdot dx(z) 
- \frac{\partial x}{\partial t}(z) \cdot dy(z) 
= \left( x(z) + \frac{\eta_A}{\omega_A} \right) dz 
= \Res_{w=0} \frac{B(z, w)}{w}.
\]
Therefore, the claim follows directly from \cite[Theorem 5.1]{EO07}. 
\end{proof}


\subsubsection{Differentiation with respect to $\nu$}
The derivatives with respect to $\nu$ are also given as follows: 

\begin{prop} [C.f., {\cite[Theorem 5.1]{EO07}}] 
\label{prop:variation-nu} ~
\begin{itemize}
\item[(i)] %
For $2g-2+n \ge 0$:
\begin{equation} \label{eq:variation-Wgn-nu}
\frac{\p}{\p \nu} W_{g,n}(z(x_{1}), \dots, z(x_{n})) 
= \oint_{z_{n+1} \in B}
W_{g,n+1}\bigl( z(x_{1}),\dots,z(x_{n}),z_{n+1} \bigr).
\end{equation}
\item[(ii)] %
For $g \ge 1$: 
\begin{equation} \label{eq:dt-closed-Fg}
\frac{\p F_{g}}{\p \nu}(t, \nu)
 = \oint_{z \in B} W_{g,1}(z).
\end{equation}
\end{itemize}
\end{prop}

\begin{proof}
This is also follows from \cite[Theorem 5.1]{EO07}. 
(See also \cite[\S 5.3]{EO07}.) 
\end{proof}

\begin{rem} \label{rem:validity-of-nu-derivative-for-unstables}
Although the general proof of Theorem 5.1 in \cite{EO07} 
is valid for $W_{g,n}$ with $2g-2+n \ge 0$ and $F_g$ for $g \ge 1$, 
we can verify that \eqref{eq:variation-Wgn-nu} and 
\eqref{eq:dt-closed-Fg} still holds for any $g \ge 0$ and $n \ge 1$. 
We can verify it by direct computation: 
For example, \eqref{eq:normalization-Bergman-kernel} shows  
\begin{equation}
\frac{\p}{\p \nu} W_{0,1}(z(x)) = \frac{\p y(x)}{\p \nu} dx 
= \frac{2 \pi i}{\omega_A} \cdot z(x) \, dx 
= \oint_{ x_1 \in {B}} B(z(x_1), z(x)).
\end{equation}
\end{rem}

\subsection{PDE for the WKB series}
\label{subsec:WKB-BPZ}

Let us introduce the WKB-type formal series
\begin{equation} \label{eq:WKB-series-B}
\psi_{\pm}(x,t,\nu; \hbar) 
:= \exp \bigl( S_{\pm}(x,t,\nu; \hbar) \bigr), 
\qquad
S_{\pm}(x,t,\nu;\hbar) := 
\sum_{m = -1}^{\infty} (\pm \hbar)^m S_m(x,t,\nu), 
\end{equation}
where the coefficients $\{ S_m(x,t,\nu) \}_{m \ge -1}$ 
in \eqref{eq:WKB-series-B} are defined as follows:
\begin{itemize}
\item 
The function $S_{-1}(x,t,\nu)$ is defined by 
\begin{equation} \label{eq:def-of-S-1}
S_{-1}(x,t,\nu) := 
\frac{4}{5} \, x^{\frac{5}{2}} + t \, x^{\frac{1}{2}} + 
\int^{x}_{\infty} \left( \sqrt{4x^3+2tx+u(t,\nu)} - 
\Bigl( 2 \, x^{\frac{3}{2}} + \frac{t}{2} \cdot x^{- \frac{1}{2}} \Bigr) \right) \, dx.
\end{equation}
Here the right hand-side is a regularization of 
a divergent integral 
$\int^{x}_{\infty} ydx$ 
$= \int^{z(x)}_{0} W_{0,1}(z)$.
It follows from \eqref{eq:PDE-u} that 
\begin{equation} \label{eq:t-derivative-S-1}
\frac{\p}{\p t} S_{-1}(x,t,\nu)  = P\bigl( z(x) \bigr), 
\qquad
\frac{\p}{\p \nu} S_{-1}(x,t,\nu)  = \frac{2\pi i}{\omega_A} \cdot z(x).
\end{equation}
where $P(z)$ is defined in \eqref{eq:Pz-function}.

\item 
The function $S_{0}(x,t,\nu)$ is defined as 
\begin{equation} \label{eq:def-of-S0}
S_0(x,t,\nu) := \frac{1}{2} \, F_{0,2}(z_1, z_2) |_{z_1=z_2=z(x)},
\end{equation} 
where
\begin{equation} \label{eq:def-F02}
F_{0,2}(z_1, z_2) := 
- \log \frac{\sigma(z_1+z_2)}{\sigma(z_1) \cdot \sigma(z_2)} 
+ \frac{\eta_A}{\omega_A} \cdot z_1 \, z_2.
\end{equation}
It is easy to verify that 
\begin{equation}
d_{z_1}d_{z_2} F_{0,2}(z_1, z_2) = 
W_{0,2}(z_1, z_2) - \frac{dx(z_1) \cdot dx(z_2)}{(x(z_1) - x(z_2))^2}.
\end{equation}
Using a well-known relation between Weierstrass $\sigma$-function
and the $\theta$-functions, we have
\begin{equation}
S_0(x,t,\nu) = - \frac{1}{4} \, 
\log \bigl( 4x^3+2tx+u(t,\nu) \bigr)
+ \log \left( - \frac{\theta_{11}'(0,\tau)}{\omega_A} \right)
- \log \theta_{11}\left( \frac{z(x)}{\omega_A}, \tau \right).
\end{equation}
Here $\tau := \tau(t,\nu) = \omega_B(t,\nu) / \omega_A(t,\nu)$.
It is also easy to verify that
\begin{equation} \label{eq:derivative-of-S0}
\frac{\p}{\p x} S_0(x,t,\nu) = 
- \frac{1}{4} \, \frac{\p}{\p x} \log \bigl( 4x^3+2tx+u(t,\nu) \bigr)
+ \frac{P\bigl( z(x) \bigr)}{\sqrt{ 4x^3+2tx+u(t,\nu) }}. 
\end{equation}

\item
The functions $S_{m}(x,t,\nu)$ with $m \ge 1$ are defined as
\begin{equation} \label{eq:higher-Sm}
S_{m}(x,t,\nu) := \sum_{\substack{2g-2+n=m \\ g\ge0, n\ge 1}} 
\frac{F_{g,n}(z_1,\dots,z_n)}{n!} \biggr|_{z_1=\cdots=z_n=z(x)} 
\end{equation}
with 
\begin{equation}
\label{eq:open-Fgn}
F_{g,n}(z_{1},\dots,z_{n}) := 
\int^{z_{1}}_{0}\cdots \int^{z_{n}}_{0} W_{g,n}(z_{1},\dots,z_{n}) 
\qquad
(2g-2+n \ge 1).
\end{equation}

\end{itemize}

\begin{thm} \label{thm:BPZ-type-equation}
The formal series $\psi_{\pm}(x,t,\nu; \hbar)$ 
defined in \eqref{eq:WKB-series-B}
is a WKB-type formal solution of the following PDE:
\begin{equation} \label{eq:BPZ-equation}
\left[ \hbar^2 \frac{\p^2}{\p x^2} 
- 2 \hbar^2 \frac{\p}{\p t}
- \Bigl( 4x^3 + 2 t x + 2 \hbar^2 \, \frac{\p F}{\p t}(t,\nu;\hbar) 
\Bigr) \right] \psi = 0.
\end{equation}
\end{thm}
\begin{proof}
The statement can be proved by a similar technique used in \cite{IS}. 
The proof is given in Appendix \ref{appendix:quantum-curve-theorem}.
\end{proof}

In view of \eqref{eq:dF0-dt}, the equation satisfied by the 
leading term of the WKB solution of the above PDE 
\eqref{thm:BPZ-type-equation} 
coincides with the defining equation
of the spectral curve \eqref{eq:spcurve}. 
In this sense, the PDE \eqref{thm:BPZ-type-equation} 
is a quantization of the spectral curve \eqref{eq:spcurve}.

\begin{rem}
In \cite{BCD}, Bouchard-Chidambaram-Dauphinee
obtained a similar result; namely, they also construct 
a quantization of Weierstrass elliptic curve by 
the topological recursion. 
Their quantum curve contains infinitely many 
$\hbar$-correction terms; this is because 
the Weierstrass elliptic curve is not admissible 
in the sense of \cite{BE16}. 
Theorem \ref{thm:BPZ-type-equation} shows that 
the $\hbar$-corrections can be controlled
by a derivative with respect to a deformation parameter.
\end{rem}

\subsection{Asymptotics at $x = \infty$}

It follows from the definitions of $S_m$ that 
\begin{equation}
S_{-1}(x,t,\nu) = 
\frac{4}{5} \, x^{\frac{5}{2}} + t \, x^{\frac{1}{2}} + O(x^{-\frac{1}{2}}), 
\quad
S_0(x,t,\nu) = - \frac{1}{4} \, \log x + O(x^{-\frac{1}{2}})
\end{equation}
and
\begin{equation}
S_m(x,t,\nu) = O(x^{- \frac{1}{2}}) \quad (m \ge 1)
\end{equation}
hold when $x$ tends to infinity. 
Thus, the formal series $S_{\pm}$ behaves as 
\begin{equation} 
\label{eq:x-behavior-of-S}
S_{\pm} = \pm \frac{4}{5\hbar} \cdot 
x^{\frac{5}{2}} \pm \frac{t}{\hbar} \cdot x^{\frac{1}{2}} 
- \frac{1}{4} \, \log x + A^{(\pm)}_0(\hbar) \cdot x^{-\frac{1}{2}} 
+ A^{(\pm)}_1(\hbar) \cdot x^{-1} 
+ A^{(\pm)}_2(\hbar) \cdot x^{- \frac{3}{2}} + \cdots 
\end{equation}
when $x$ tends to $\infty$. 
Here $A^{(\pm)}_i(\hbar)$'s are formal power series of $\hbar$ 
whose coefficients are independent of $x$. 
Actually, we can determine those formal series by substituting 
the expansion \eqref{eq:x-behavior-of-S} into the equation 
\eqref{eq:PDE-Riccati-quantum-curve} satisfied by $S$; 
after simple computation, we have 
\begin{equation}
A^{(\pm)}_0 = \mp \hbar \, \frac{\p F}{\p t}(t,\nu;\hbar), 
\quad
A^{(\pm)}_1 = \mp \frac{\hbar}{2} \cdot \frac{\p A^{(\pm)}_0}{\p t} = 
 \frac{\hbar^2}{2} \cdot \frac{\p^2 F}{\p t^2}(t,\nu; \hbar), 
\end{equation}
\begin{equation}
A^{(\pm)}_2 = 
\mp \left(\frac{\hbar}{3} \cdot \frac{\p A^{(\pm)}_1}{\p t} - \frac{t^2}{24 \hbar} \right)
= \mp \left( \frac{\hbar^3}{6} \cdot \frac{\p^3 F}{\p t^3}(t,\nu;\hbar) 
- \frac{t^2}{24 \hbar} \right),
\end{equation}
and so on. This implies the following (where $'$ means the $t$-derivative):
\begin{multline} \label{eq:behavior-psi-pm}
\psi_{\pm}(x,t,\nu;\hbar) 
\\ 
= \exp\left(  \pm \frac{1}{\hbar} \cdot 
\Bigl( \frac{4}{5} \, x^{\frac{5}{2}} + t \, x^{\frac{1}{2}} \Bigr)  \right)
\, \cdot \, x^{-\frac{1}{4}} \cdot 
\Biggl\{  1 \mp \hbar \, F'(t,\nu;\hbar) \cdot x^{- \frac{1}{2}} 
+ \hbar^2 \, \Bigl(  F''(t,\nu;\hbar) + F'(t,\nu;\hbar)^2 \Bigr) \cdot x^{-1} \\
\mp \frac{4\hbar^4 \, F^{(4)}(t,\nu;\hbar) 
+ 12 \hbar^4 \, F'(t,\nu;\hbar) \cdot F''(t,\nu;\hbar) 
+ 4 \hbar^2 \,  F'(t,\nu;\hbar)^3}{24} \cdot x^{- \frac{3}{2}}
+ O(x^{-2}) 
\Biggr\}.
\end{multline}

\subsubsection{Formal monodromy relation}

The WKB series $\psi_{\pm}$ is a formal series of $\hbar$ 
whose coefficients are multivalued functions on the spectral curve. 
By the {\em formal monodromy} of $\psi_{\pm}$ 
we mean the monodromy of $\psi_{\pm}$ for the term-wise analytic continuation 
along closed cycles on the spectral curve. 

The following observation is crucial for the proof of our main results. 
\begin{thm} \label{thm:formal-monodromy-psi}
The WKB series $\psi_{\pm}(x,t,\nu;\hbar)$ 
defined in \eqref{eq:WKB-series-B} has the 
following formal monodromy properties: 
\begin{itemize}
\item[{\rm (i)}]
The formal monodromy of $\psi_{\pm}(x,t,\nu;\hbar)$ 
along the $A$-cycle is given by 
\begin{equation}
\psi_{\pm}(x,t,\nu;\hbar) \mapsto 
e^{\pm {2\pi i \nu}/{\hbar}} \cdot 
\psi_{\pm}(x,t,\nu;\hbar).
\end{equation}
\item[{\rm (ii)}]
The formal monodromy of $\psi_{\pm}(x,t,\nu;\hbar)$ 
along the $B$-cycle is given by 
\begin{equation}
\psi_{\pm}(x,t,\nu;\hbar) \mapsto 
\frac{Z(t,\nu \pm \hbar; \hbar)}{Z(t,\nu;\hbar)} \cdot 
\psi_{\pm}(x,t,\nu \pm \hbar;\hbar).
\end{equation}
\end{itemize}
\end{thm}

\begin{proof}
The claim (i) is a consequence of \eqref{eq:trivial-A-period-Wgn} 
and the definition \eqref{eq:nu} of $\nu$. 
The second claim (ii) can be proved by the following computation: 
\begin{align}
& 
\text{Term-wise analytic continuation of 
$\psi_\pm(x,t,\nu;\hbar)$ along the $B$-cycle} \notag \\ 
& \quad = 
\exp\left( \sum_{\substack{g \ge 0, \, n \ge 1}} 
\frac{(\pm \hbar)^{2g-2+n}}{n!} \cdot
F_{g,n}\bigl( z(x) + \omega_B,\dots, z(x)+\omega_B \bigr) \right) \notag \\ 
& \quad = 
\exp\left( \sum_{\substack{g \ge 0, \, n \ge 1}} \frac{(\pm \hbar)^{2g-2+n}}{n!} 
\cdot 
\sum_{\ell=0}^{n} {\binom n\ell} \underbrace{\oint_{B} \cdots \oint_B}_{\ell} 
\underbrace{\int^{z(x)}_0 \dots \int^{z(x)}_0}_{n-\ell} W_{g,n}(z'_1, \dots, z'_n)
\right) \notag \\
& \quad = 
\exp\left( \sum_{\substack{g, \ell_1, \ell_2 \ge 0  \\ \ell_1+\ell_2 \ge 1}} 
\frac{(\pm \hbar)^{2g-2+\ell_1+\ell_2}}{\ell_1! \cdot \ell_2 !} \cdot
\frac{\p^{\ell_1}}{\p \nu^{\ell_1}} 
\underbrace{\int^{z(x)}_0 \dots \int^{z(x)}_0}_{\ell_2} 
W_{g,\ell_2}(z'_1, \dots, z'_{\ell_2})
\right) \notag \\
& \quad = 
\exp\left( \sum_{\ell_1 \ge 1} \frac{(\pm \hbar)^{\ell_1}}{\ell_1!} \cdot
\frac{\p^{\ell_1}}{\p \nu^{\ell_1}} 
\sum_{g \ge 0} \hbar^{2g-2} F_g(t,\nu) \right)   \notag \\
& \hspace{+2.2em}
\times \exp\left( \sum_{\ell_1 \ge 0} \frac{(\pm \hbar)^{\ell_1}}{\ell_1!} \cdot
\frac{\p^{\ell_1}}{\p \nu^{\ell_1}} 
\sum_{\substack{g \ge 0, \, \ell_2 \ge 1}}
\frac{(\pm \hbar)^{2g-2+\ell_2}}{\ell_2!} F_{g,\ell_2}(z(x), \dots, z(x))
\right) \notag \\
& \quad = 
\frac{Z(t,\nu \pm \hbar; \hbar)}{Z(t,\nu;\hbar)} \cdot 
\psi_{\pm}(x,t,\nu \pm \hbar;\hbar).
\end{align}
Here we have written $F_g = W_{g,0}$, and used Proposition \ref{prop:variation-nu} 
to reduce the integration along the $B$-cycle to the $\nu$-derivative
(c.f., Remark \ref{rem:validity-of-nu-derivative-for-unstables}).
\end{proof}

A similar computation, 
which converts a term-wise analytic continuation to a shift operator, 
was used in \cite{IKoT1, IKoT2} to establish a relationship 
between the Voros coefficients of quantum curves and the free energy
for a class of spectral curves arising from 
the hypergeometric differential equations. 
 
\begin{rem}
Let us denote by
\begin{equation}
{\rm Wr}(k_1, k_2) := Z(\nu+k_1\hbar) \cdot Z(\nu+k_2\hbar)
\cdot \left(\hbar \, \frac{\p \psi_+(\nu+k_1 \hbar)}{\p x} 
\cdot \psi_-(\nu+k_2 \hbar) 
- \psi_+(\nu+k_1 \hbar) \cdot \hbar \, \frac{\p \psi_-(\nu+k_2 \hbar)}{\p x} 
\right)
\end{equation}
the Wrosnskian of the WKB series with parameter shifts.
Theorem \ref{thm:formal-monodromy-psi} implies that 
it has formal monodoromy 
\begin{eqnarray}
{\rm Wr}(k_1, k_2) \mapsto 
\begin{cases}
{\rm Wr}(k_1, k_2) &  \text{along $A$-cycle} \\ 
{\rm Wr}(k_1+1, k_2-1) &  \text{along $B$-cycle}.
\end{cases}
\end{eqnarray}
This is because of the fact that the Wronskian has 
singularities at the branch points $x=e_1, e_2, e_3$. 
However, as we will see below, the singularities disappear 
after taking the discrete Fourier transform 
(c.f., Proposition \ref{prop:polynomiality-of-Wronskian}). 
Although the formal monodromy of $\psi_{\pm}$ is realized by 
a shift of $\nu$, the discrete Fourier transform cancels 
those shifts in the Wronskian. 
This is crucial to construct a formal solution of the isomonodromy system 
and 2-parameter $\tau$-function associated with $\PI$. 
We borrow the idea (i.e., using the discrete Fourier transform) 
from the conformal field theoretic construction 
of the $\tau$-function for Painlev\'e equations 
(c.f., \cite{GIL} etc.). 
\end{rem}

\section{Main results : 
2-parameter $\tau$-functions and the isomonodromy system}
\label{section:main-result}

\subsection{Statement of the main results}

\subsubsection{The formal $\tau$-function for Painlev\'e I}

Let us consider the formal series 
\begin{equation} \label{eq:tau-function}
\tau_{P_{\rm I}}(t,\nu,\rho;\hbar) 
:= \sum_{k \in {\mathbb Z}} e^{{2 \pi i k \rho}/{\hbar}} 
\cdot Z(t,\nu+k \hbar; \hbar),
\end{equation}
which is the discrete Fourier transform of 
the partition function $Z(t,\nu; \hbar)$ 
with respect to $\nu$. 
This type of formal series was introduced in \cite{EM08} 
in another equivalent expression, 
and called the ``non-perturbative partition function".
An important observation by \cite{Borot-Eynard, EM08} 
is that the formal series of the above form 
is expressed as a formal power series of $\hbar$ 
whose coefficients are given by a finite sum of 
the $\theta$-functions and their derivatives.

Such an expression is obtained as follows.
Using the series expansion, we have 
\begin{equation}
\frac{Z(\nu+k \hbar)}{Z(\nu)}  
= \exp\left( \frac{2\pi i k \phi}{\hbar} 
+ \pi i k^2 \tau \right) 
\cdot \,
\exp\left(
\sum_{\substack{g \ge 0, ~ n \ge 1 \\ 2g-2+n \ge 1}} 
\frac{\hbar^{2g-2+n}}{n!} \cdot k^n \cdot
\frac{\p^n F_g}{\p \nu^n}
\right)
\end{equation}
for each $k \in {\mathbb Z}$. 
Here we have set
\begin{align}
\phi = \phi(t,\nu) & 
:= \frac{1}{2 \pi i} \cdot \frac{\p F_0}{\p \nu}  
= \frac{1}{2\pi i} \oint_B W_{0,1}(z),  \label{eq:def-phi} \\
\tau = \tau(t,\nu) & 
:= \frac{1}{2 \pi i} \cdot  \frac{\p^2 F_0}{\p \nu^2} =
\frac{1}{2 \pi i}\oint_{z_1 \in B}\oint_{z_2 \in B} W_{0,2}(z_1,z_2) 
= \frac{\omega_B}{\omega_A}. 
\end{align}
Therefore, the formal series \eqref{eq:tau-function}
has the following formal power series expansion: 
\begin{equation} \label{eq:series-expansion-tau}
\tau_{P_{\rm I}}
= Z(\nu) \cdot \sum_{k \in {\mathbb Z}}
e^{{2\pi i k \rho}/{\hbar}} \cdot \frac{Z(\nu+k \hbar)}{Z(\nu)} 
= Z(\nu) \cdot \sum_{m = 0}^{\infty}
\hbar^{m} \Theta_{m}(t,\hbar),
\end{equation}
where 
\begin{align}
\Theta_{0} & = \sum_{k \in {\mathbb Z}} 
\exp\left( \frac{2\pi i k(\phi+\rho)}{\hbar} 
+ \pi i k^2 \tau \right) 
= \theta_{00}\Big( \frac{\phi+\rho}{\hbar}, \tau \Bigr),  
\\
\Theta_{1} & = 
\sum_{k \in {\mathbb Z}} 
\exp\left( \frac{2\pi i k(\phi+\rho)}{\hbar} 
+ \pi i k^2 \tau \right) \cdot 
\Bigl( 
\frac{k^3}{3!} \frac{\p^3 F_0}{\p \nu^3} + k \frac{\p F_1}{\p \nu}
\Bigr)  \notag \\
& = \frac{ \p_\nu^3 F_0}{(2\pi i )^3 \, 3!}
\cdot \frac{\p^3 \theta_{00}}{\p v^3}\Big( \frac{\phi+\rho}{\hbar}, \tau \Bigr) 
+ \frac{\p_\nu F_1}{2 \pi i} \cdot
\frac{\p \theta_{00}}{\p v}\Big( \frac{\phi+\rho}{\hbar}, \tau \Bigr) 
\end{align}
and so on. (See Appendix \ref{subsection:theta-functions} 
for the definition and several properties of $\theta$-functions.)
It is obvious that, for each $m \ge 0$,  
the coefficient $\Theta_{m}$ is expressed as 
a finite sum of $\theta$-functions and their derivatives.

\begin{rem} \label{rem:t-differential-structure}
Two formal series of the form \eqref{eq:series-expansion-tau} are obviously 
summed, multiplied, as well as the usual formal power series.  
However, we need to be careful when we differentiate the formal series by $t$. 
We define the $t$-derivative of such formal series by 
term-wise differentiation as usual:
\begin{equation}
\frac{\p}{\p t} \left( \sum_{m} \hbar^m \Theta_{m}(t,\hbar) \right) 
:= \sum_{m} \hbar^m \, \frac{\p}{\p t} \Theta_{m}(t,\hbar).
\end{equation}
Since each coefficient $\Theta_{m}(t,\hbar)$ is of the form 
\begin{equation} \label{eq:expression-of-coefficients-theta}
\Theta_{m}(t,\hbar) = 
\tilde{\Theta}_{m}(t,v) \bigl|_{v=\frac{\phi+\rho}{\hbar}}
\end{equation} 
with a certain function $\tilde{\Theta}_{m}(t,v)$, 
which is written by a finite sum of 
$\theta$-function $\theta_{00}(v,\tau)$ and their derivatives, 
the $t$-derivative does not preserve the $\hbar$-grading. 
More precisely, we have
\begin{equation}
\frac{\p}{\p t} \left( 
\sum_{m} \hbar^m \Theta_{m}(t,\hbar) \right)
= \sum_{m} \hbar^m \left[
\frac{\p \tilde{\Theta}_{m}}{\p t}(t, v) 
+ \frac{\p \phi}{\p t} \cdot
\frac{\p \tilde{\Theta}_{m+1}}{\p v}(t, v) 
\right]_{v=\frac{\phi+\rho}{\hbar}}.
\end{equation}
Note also that this type of the $t$-differential structure 
also appeared in the work of Aoki-Kawai-Takei (\cite{AKT-P, KT-PIII}),
where another type of 2-parameter family of formal solutions 
of Painlev\'e equations was constructed by the multiple-scale method.
\end{rem}

Taking the above remark into consideration, we further define 
\begin{align}
H(t,\nu,\rho;\hbar) 
& := \hbar^2 \frac{\p}{\p t} \log \tau_{P_{\rm I}}(t,\nu,\rho;\hbar)
= \sum_{m = 0}^\infty \hbar^m H_m(t,\hbar) \notag \\
& = \frac{\p F_0}{\p t} + \hbar \frac{\p \phi}{\p t} 
\cdot \left[ \frac{\p}{\p v} \log \theta_{00}(v, \tau) 
\right]_{v = \frac{\phi+\rho}{\hbar}} + O(\hbar^2),
\label{eq:hbar-expansion-of-H} \\
q(t,\nu,\rho;\hbar) 
& := -\frac{\p}{\p t} H(t,\nu,\rho;\hbar) 
= \sum_{m = 0}^\infty \hbar^m q_m(t,\hbar)
\notag \\
& = - \frac{\p^2 F_0}{\p t^2} - 
\left( \frac{\p \phi}{\p t} \right)^2 \cdot
\left[
\frac{\p^2}{\p v^2} \log \theta_{00}(v, \tau) 
\right]_{v = \frac{\phi+\rho}{\hbar}} + O(\hbar), 
\label{eq:hbar-expansion-of-q}
\\
p(t,\nu,\rho;\hbar) 
& := \hbar \frac{\p}{\p t} q(t,\nu,\rho;\hbar)
= \sum_{m = 0}^\infty \hbar^m p_m(t,\hbar) \notag \\
& = - 
\left( \frac{\p \phi}{\p t} \right)^3 \cdot 
\left[ 
\frac{\p^3}{\p v^3} \log \theta_{00}(v, \tau) 
\right]_{v = \frac{\phi+\rho}{\hbar}}  + O(\hbar).
\label{eq:hbar-expansion-of-p}
\end{align}
Note that the symbol $O(\hbar)$ does not mean the usual 
Landau's notation. We are working 
with formal power series, and the symbol just means 
the higher order terms in a formal power series of $\hbar$.
The $\hbar$-dependence of the coefficients 
only comes from the substitution $v \mapsto (\phi(t)+\rho)/\hbar$ 
into arguments of $\theta$-functions 
(c.f., \eqref{eq:expression-of-coefficients-theta}). 

As we will see below, these notations 
are consistent with those given in \S \ref{section:Painleve-I}.

\begin{lem} \label{lem:leading-terms-relation}
The leading terms 
\begin{align}
H_0 & = \frac{\p F_0}{\p t}, \\
q_0 & = - \frac{\p^2 F_0}{\p t^2} - 
\left( \frac{\p \phi}{\p t} \right)^2 \cdot
\left[\frac{\p^2}{\p v^2} \log \theta_{00}(v, \tau) 
\right]_{v = \frac{\phi+\rho}{\hbar}}, 
\label{eq:q0-leading-term}\\
p_0 & = - 
\left( \frac{\p \phi}{\p t} \right)^3 \cdot
\left[\frac{\p^3}{\p v^3} \log \theta_{00}(v, \tau) 
\right]_{v = \frac{\phi+\rho}{\hbar}} 
\end{align}
of the formal series $H$, $q$, $p$ given in 
\eqref{eq:hbar-expansion-of-H}--\eqref{eq:hbar-expansion-of-p}
satisfy the algebraic relation
\begin{equation} \label{eq:Hamiltonian-relation-leading}
H_0 = \frac{p_0^2}{2} - 2 q_0^3 - t q_0.
\end{equation}
\end{lem}
\begin{proof}
The Weierstarss functions and $\theta$-functions 
with characteristics are related as 
\begin{align} 
\wp(z) & = - \frac{d^2}{dz^2} \log \sigma(z) = 
- \frac{\eta_A}{\omega_A} - \frac{1}{\omega_A^2} \cdot
\frac{\theta_{11}'' \bigl( \frac{z}{\omega_A}, \tau \bigr) \cdot
\theta_{11} \bigl( \frac{z}{\omega_A}, \tau \bigr)
-\theta_{11}' \bigl( \frac{z}{\omega_A}, \tau \bigr)^2}
{\theta_{11} \bigl( \frac{z}{\omega_A}, \tau \bigr)^2} 
\end{align}
\begin{equation}
\theta_{11}\Bigl( v + \frac{1}{2} 
+ \frac{\tau}{2}, \tau \Bigr)  =
- e^{- \pi i v - \frac{\pi i \tau}{4}} \cdot
\theta_{00} (v,\tau).
\end{equation}
(See Appendix \ref{appendix:Weierstrass}.)
Using the relations and 
\begin{equation}
\frac{\p \phi}{\p t} = \frac{1}{2 \pi i} \cdot \frac{\p^2}{\p t \p \nu} F_0 
= \frac{1}{4 \pi i} \cdot \frac{\p u}{\p \nu} = \frac{1}{\omega_A}
\end{equation}
\begin{equation}
\frac{\p^2 F_0}{\p t^2} = \frac{1}{2} \frac{\p u}{\p t} = \frac{\eta_A}{\omega_A}
\end{equation}
(c.f., \eqref{eq:PDE-u} and \eqref{eq:def-phi}), 
we can verify that 
\begin{equation} \label{eq:leading-elliptic-behavior}
- \frac{\p^2 F_0}{\p t^2} - 
\left( \frac{\p \phi}{\p t} \right)^2 \cdot
\frac{\p^2}{\p v^2} \log \theta_{00}(v, \tau) 
= \wp\left( \omega_A \cdot v  + \frac{\omega_A + \omega_B}{2}\right)
\end{equation}
and 
\begin{equation}
\left( \frac{\p \phi}{\p t} \right)^3 \cdot
\frac{\p^3}{\p v^3} \log \theta_{00}(v, \tau) 
= \wp' \left( \omega_A \cdot v  + \frac{\omega_A + \omega_B}{2}\right).
\end{equation}
Since $\wp$ and $\wp'$ are introduced to parametrize the 
elliptic curve \eqref{eq:spcurve}, we can verify 
\begin{equation}
\frac{p_0^2}{2} - 2q_0^3 - t q_0 = \frac{1}{2} 
\left[ \wp'(z)^2 - 4 \wp^3(z) -  2 t \wp(z) 
\right]_{z=\frac{\phi+\rho}{\hbar} \,\cdot\, \omega_A 
+ \frac{\omega_A + \omega_B}{2}}
= \frac{u}{2}.
\end{equation}
Then, the desired equality \eqref{eq:Hamiltonian-relation-leading} 
follows from the formula \eqref{eq:dF0-dt}. 
\end{proof}

It also follows from the above proof and \eqref{eq:Weierstrass-equation-2} that 
\begin{equation} \label{eq:Painleve1-leading}
\left[ \hbar \frac{\p p}{\p t} \right]_{\hbar^0}  
= \left[ \wp'' \Bigl( \omega_A \cdot v  + \frac{\omega_A + \omega_B}{2}\Bigr)
\right]_{v = \frac{\phi+\rho}{\hbar}}
= \left[6\wp\Bigl( \omega_A \cdot v  + \frac{\omega_A + \omega_B}{2}\Bigr)^2 
+ t \right]_{v = \frac{\phi+\rho}{\hbar}}
= \left[ 6q^2 + t \right]_{\hbar^0},
\end{equation}
where the symbol $[\bullet]_{\hbar^m}$ means the coefficient 
of $\hbar^m$ of a given formal power series $\bullet$.

The relations \eqref{eq:Hamiltonian-relation-leading} 
and \eqref{eq:Painleve1-leading}
are the leading part of the following equalities, 
which are our first main theorem. 

\begin{thm}\label{thm:tau-function-theorem}
The formal series $H$, $q$, $p$ given in 
\eqref{eq:hbar-expansion-of-H}--\eqref{eq:hbar-expansion-of-p}
satisfy
\begin{equation} \label{eq:Hamiltonian-realtion}
H = \frac{p^2}{2} - 2q^3 - t q
\end{equation}
and 
\begin{equation} \label{eq:Paineve-1-equation}
\hbar \frac{\p p}{\p t} = 6q^2 + t.
\end{equation}
In other words, $\tau_{P_{\rm I}}$ is a 2-parameter family of 
the formal solution of \eqref{eq:Hirota-Painleve-1}, 
and $q$ is the corresponding 2-parameter formal solution of $\PI$.
\end{thm}

This statement is in fact equivalent to 
our second main theorem, which will be formulated in next subsection.
The proof of these theorems will be given later. 

\begin{rem} \label{rem:Boutroux-asymptotics}
It follows from \eqref{eq:leading-elliptic-behavior} that 
the leading term of the 2-parameter formal solution $q$ 
of Painlev\'e I is given by 
\begin{equation} \label{eq:leading-painleve-explicit}
q_0 = \wp \left( \frac{\phi + \rho}{\hbar} \cdot \omega_A + 
\frac{\omega_A + \omega_B}{2} \right) 
= \wp \left( \frac{4t}{5 \hbar} 
+ \Bigl( \frac{\rho}{\hbar} + \frac{1}{2} \Bigr)  \cdot \omega_A
+ \Bigl( \frac{\nu}{\hbar} + \frac{1}{2} \Bigr)  \cdot \omega_B
\right).
\end{equation}
Here we used an equality 
\begin{equation}
\omega_A \cdot \phi - \omega_B \cdot \nu = 
2 \cdot \Res_{x=\infty} z(x) \cdot \sqrt{4x^3+2t x + u(t,\nu)} \, dx 
= \frac{4t}{5}
\end{equation}
which can be proved by a similarly to a proof of 
the Riemann bilinear identity \eqref{eq:R-B-Id}. 

On the other hand, if we rescale the variables by 
\begin{equation}
(x,y,u) = (|t|^{\frac{1}{2}} \tilde{x}, |t|^{\frac{3}{4}} \tilde{y}, 
|t|^{\frac{3}{2}} \tilde{u}),
\end{equation}
the elliptic curve \eqref{eq:spcurve} becomes 
\begin{equation} \label{eq:rescaled-spcurve}
\tilde{y}^2 = 4 \tilde{x}^3 + 2 e^{i \varphi} \tilde{x} + \tilde{u},
\end{equation}
where $\varphi = \arg t$. 
Then, the right hand-side of \eqref{eq:leading-painleve-explicit}
can be written in terms of the Weierstrass $\wp$-function 
$\tilde{\wp}(\tilde{z})$ which parametrizes \eqref{eq:rescaled-spcurve}:
\begin{equation} \label{eq:Boutroux-asymptotics}
q_0 = |t|^{\frac{1}{2}} \cdot \tilde{\wp}\Bigl( 
\frac{4}{5 \hbar}\cdot |t|^{\frac{5}{4}} \cdot e^{i \varphi} 
+ \Bigl( \frac{\rho}{\hbar} + \frac{1}{2} \Bigr)  \cdot \tilde{\omega}_A
+ \Bigl( \frac{\nu}{\hbar} + \frac{1}{2} \Bigr)  \cdot \tilde{\omega}_B \Bigr).
\end{equation}
Here, 
\begin{equation}
\tilde{\omega}_\ast = |t|^{\frac{1}{4}} \cdot \omega_\ast \qquad
(\ast \in \{A, B \})
\end{equation}
is the $A$ and $B$-periods of $d\tilde{x}/\tilde{y}$ 
of the rescaled elliptic curve \eqref{eq:rescaled-spcurve}.
At the level of formal series computation, the expression 
\eqref{eq:Boutroux-asymptotics} is consistent with the
``elliptic aymptotic", which was first discovered 
by Boutroux \cite{Boutroux}, and studied by Its, Kitaev, Kapaev et.al. 
(see \cite{FIKN, Kap, Kap-Kit} etc.).
\end{rem}

\subsubsection{The formal solution of the isomonodromy system}

Let us introduce another formal series
\begin{equation} \label{eq:wave-function}
\Psi_{\pm}(x,t,\nu, \rho; \hbar) := 
\frac{\sum_{k \in \bZ} e^{{2 \pi i k \rho}/{\hbar}} \cdot
Z(t,\nu+k \hbar; \hbar) \cdot \psi_{\pm}(x,t,\nu + k \hbar; \hbar)}
{\sum_{k \in \bZ} e^{{2 \pi i k \rho}/{\hbar}} \cdot
Z(t,\nu+k \hbar; \hbar)}.
\end{equation}
By a similar argument given in the previous subsection, 
we can verify that $\Psi_{\pm}$ has a formal power series expansion of $\hbar$ 
in the following form:
\begin{equation} \label{eq:series-expansion-wave-function}
\Psi_{\pm} = 
\psi_{\pm}(\nu) \cdot \frac{
\sum_{m = 0}^{\infty} \hbar^{m} \, \Xi_{\pm, m}(x,t,\hbar)}
{\sum_{m = 0}^{\infty} \hbar^{m} \, \Theta_{m}(t,\hbar)},
\end{equation}
with certain functions $\Xi_{\pm, m}(x,t,\hbar)$ which contains 
$\theta$-functions and their derivatives. 
First few terms are given by 
\begin{align}
\Xi_{\pm, 0} & = \sum_{k \in {\mathbb Z}} 
\exp\left( \frac{2\pi i k(\phi+\rho)}{\hbar} \pm 2 \pi i k \frac{z(x)}{\omega_A}
+ \pi i k^2 \tau \right) 
= \theta_{00}\Big( \frac{\phi+\rho}{\hbar} \pm \frac{z(x)}{\omega_A}, \tau \Bigr),  
\\
{\Xi}_{\pm, 1} & = \sum_{k \in {\mathbb Z}} 
\exp\left( \frac{2\pi i k(\phi+\rho)}{\hbar} \pm 2 \pi i k \frac{z(x)}{\omega_A}
+ \pi i k^2 \tau \right) 
\cdot \left(  
\frac{k^3}{3!} \frac{\p^3 F_0}{\p \nu^3} 
+ k \frac{\p F_1}{\p \nu}
\pm  \frac{k^2}{2!} \frac{\p^2 S_{-1}}{\p \nu^2} 
+ k \frac{\p S_0}{\p \nu}  \right) 
\notag \\
& = 
\frac{\p_\nu^3 F_0}{(2\pi i)^3 3!} \cdot 
\frac{\p^3 \theta_{00}}{\p v^3}\Big(\frac{\phi+\rho}{\hbar} \pm \frac{z(x)}{\omega_A}, \tau \Bigr)
+ \frac{\p_\nu F_1}{2\pi i} \cdot 
\frac{\p \theta_{00}}{\p v}\Big(\frac{\phi+\rho}{\hbar} \pm \frac{z(x)}{\omega_A}, \tau \Bigr)
\notag \\
& \qquad
\pm \frac{\p^2_\nu S_{-1}}{(2\pi i)^2 2!} \cdot 
\frac{\p^2 \theta_{00}}{\p v^2}
\Big(\frac{\phi+\rho}{\hbar} \pm \frac{z(x)}{\omega_A}, \tau \Bigr)
+ \frac{\p_\nu S_0}{2\pi i} \cdot 
\frac{\p \theta_{00}}{\p v}\Big(\frac{\phi+\rho}{\hbar} 
\pm \frac{z(x)}{\omega_A}, \tau \Bigr).
\end{align}
It is obvious that the coefficients ${\Xi}_{\pm, m}$ 
for higher order powers of $\hbar$ are expressed as 
\begin{equation} \label{eq:hbar-dependence-of-Xi}
{\Xi}_{\pm, m}(x,t,\hbar) = \tilde{\Xi}_{\pm, m}(x,t,w) 
\bigl|_{w=\frac{\phi+\rho}{\hbar} \pm \frac{z(x)}{\omega_A}}
\end{equation}
with a certain function $\tilde{\Xi}_{\pm, m}(x,t,w)$ which is 
also written by a finite sum of 
$\theta$-function $\theta_{00}(w,\tau)$ and their derivatives. 

Here we remark that, unlike the $t$-derivative, 
$x$-derivative preserves the $\hbar$-grading since 
the coefficient of $1/\hbar$ in the argument of the $\theta$-functions
does not depend on $x$.

Our second main result is formulated as follows. 
\begin{thm}\label{thm:wave-function-theorem}
The formal series $\Psi_{\pm}$ satisfies 
the isomonodrmy system $(L_{\rm I})$--$(D_{\rm I})$ 
associated with $\PI$ given in \eqref{eq:JM-1}--\eqref{eq:JM-2}. 
Here $H$, $q$ and $p$ in \eqref{eq:JM-1}--\eqref{eq:JM-2}  
are given by \eqref{eq:hbar-expansion-of-H}, 
\eqref{eq:hbar-expansion-of-q}
and \eqref{eq:hbar-expansion-of-p}, respectively.
\end{thm}

The rest of the section is devoted to a proof of the statement. 

\subsection{Proof of main results}

Theorem \ref{thm:tau-function-theorem} 
and Theorem \ref{thm:wave-function-theorem} 
will be proved after several steps.

\subsubsection{Formal monodromy relation}

\begin{thm} \label{thm:formal-monodromy-Psi}
The formal series $\Psi_{\pm}(x,t,\nu,\rho;\hbar)$ 
defined in \eqref{eq:wave-function}
has the following formal monodromy properties: 
\begin{itemize}
\item[{\rm (i)}]
The formal monodromy of $\Psi_{\pm}(x,t,\nu,\rho;\hbar)$ 
along the $A$-cycle is given by 
\begin{equation}
\Psi_{\pm}(x,t,\nu,\rho;\hbar) \mapsto 
e^{\pm {2\pi i \nu}/{\hbar}} \cdot 
\Psi_{\pm}(x,t,\nu,\rho;\hbar).
\end{equation}
\item[{\rm (ii)}]
The formal monodromy of $\Psi_{\pm}(x,t,\nu,\rho;\hbar)$ 
along the $B$-cycle is given by 
\begin{equation}
\Psi_{\pm}(x,t,\nu,\rho;\hbar) \mapsto 
e^{\mp {2\pi i \rho}/{\hbar}} \cdot 
\Psi_{\pm}(x,t,\nu,\rho;\hbar).
\end{equation}
\end{itemize}
\end{thm}

\begin{proof}
The claim (i) is a consequence of the periodicity 
$\exp\bigl(\pm \frac{2\pi i \, (\nu+ k \hbar)}{\hbar}\bigr) = 
\exp\bigl(\pm \frac{2\pi i \nu}{\hbar} \bigr)$ 
of the exponential function. The second statement (ii) 
follows directly from (ii) in Theorem \ref{thm:formal-monodromy-psi}.
\end{proof}

\subsubsection{The formal series $Q$ and $R$}

Let us conisder the formal series $Q$ and $R$ defined by
\begin{equation} \label{eq:def-Q-and-R}
\begin{pmatrix} 
\Psi_+ & \Psi_- \\ 
\hbar \p_t \Psi_+ & \hbar \p_t \Psi_-
\end{pmatrix} = 
\begin{pmatrix} 
1 & 0 \\ 
Q & R
\end{pmatrix}
\cdot
\begin{pmatrix} 
\Psi_+ & \Psi_- \\ 
\hbar \p_x \Psi_+ & \hbar \p_x \Psi_-
\end{pmatrix}.
\end{equation}
More specifically, these formal series are defined by 
\begin{equation}
Q := \frac{\hbar \frac{\p {\Psi}_{+}}{\p x} \cdot \hbar \frac{\p {\Psi}_{-}}{\p t} - 
\hbar \frac{\p {\Psi}_{-}}{\p x} \cdot \hbar \frac{\p {\Psi}_{+}}{\p t}}
{{\rm Wr}^{x}[\Psi_+, \Psi_-]}, \quad
R := \frac{{\rm Wr}^t[{\Psi}_{+}, {\Psi}_{-}]}{{\rm Wr}^x[{\Psi}_{+}, {\Psi}_{-}]}, 
\end{equation}
where 
\begin{equation} \label{eq:def-of-two-Wronskians}
{\rm Wr}^x[{\Psi}_{+}, {\Psi}_{-}] := 
\hbar \frac{\p {\Psi}_{+}}{\p x} \cdot {\Psi}_- - 
{\Psi}_+ \cdot \hbar \frac{\p {\Psi}_{-}}{\p x}, 
\quad  
{\rm Wr}^t[{\Psi}_{+}, {\Psi}_{-}] := 
\hbar \frac{\p {\Psi}_{+}}{\p t} \cdot {\Psi}_- - 
{\Psi}_+ \cdot \hbar \frac{\p {\Psi}_{-}}{\p t}  
\end{equation}
are the Wronskians. Although $\Psi_{\pm}$ contains terms of 
the form \eqref{eq:hbar-dependence-of-Xi}, the series 
expansion of the Wronskians are expressed in the following form:
\begin{equation}
{\rm Wr}^x[{\Psi}_{+}, {\Psi}_{-}] = 
\sum_{m=0}^{\infty} \hbar^m {\rm Wr}^x_{m}(x,t,v)\bigl|_{v=\frac{\phi}{\hbar}}, 
\quad
{\rm Wr}^t[{\Psi}_{+}, {\Psi}_{-}] = 
\sum_{m=0}^{\infty} \hbar^m {\rm Wr}^t_{m}(x,t,v)\bigl|_{v=\frac{\phi}{\hbar}}.
\end{equation}
This is because of the fact that 
$\p_v^k\theta_{00}(X+Y, \tau) \cdot
\p_v^\ell\theta_{00}(X-Y, \tau)$
of derivatives of the $\theta$-functions can be written as 
a differential polynomial of $\theta_{00}(X,\tau)$, $\theta_{00}(Y,\tau)$ 
$\theta_{11}(X,\tau)$ and $\theta_{11}(Y,\tau)$
thanks to 
\begin{equation}
\theta_{00}(X+Y, \tau) \cdot 
\theta_{00}(X-Y, \tau) \cdot 
\theta_{00}(0,\tau)^2 =
\theta_{00}(X,\tau)^2 \cdot 
\theta_{00}(Y, \tau)^2 + 
\theta_{11}(X,\tau)^2 \cdot
\theta_{11}(Y,\tau)^2
\end{equation}
among the $\theta$-functions.
Therefore, $Q$ and $R$ are also written as a formal power series
\begin{equation}
Q = \sum_{m = 0}^{\infty} \hbar^{m} Q_m(x,t,\hbar), \quad
R = \sum_{m = 0}^{\infty} \hbar^{m} R_m(x,t,\hbar),
\end{equation}
where $Q_m$ and $R_m$ are of the form 
\begin{equation}
Q_m(x,t,\hbar) = \tilde{Q}_m(x,t,v) \bigl|_{v=\frac{\phi}{\hbar}}, 
\quad
R_m(x,t,\hbar) = \tilde{R}_m(x,t,v) \bigl|_{v=\frac{\phi}{\hbar}}.
\end{equation}

An important consequence of Theorem \ref{thm:formal-monodromy-Psi} 
is that the coefficients of $Q$ and $R$ 
are rational in $x$ since the two Wronski matrices 
appeared in both sides of \eqref{eq:def-Q-and-R} 
have the same formal monodromy along any cycles on the spectral curve. 
(Essential singularities never appear since the exponential behaviours 
in the WKB series cancel after taking the Wronskians.)
Candidates of the poles of $Q_m$ and $R_m$ are $x=\infty$, 
the branch points $e_1, e_2, e_3$,  
and the zeros of the leading terms of the Wronskian. 
In the next subsection, we will see that $Q_m$ and $R_m$ 
are in fact holomorphic at the branch points. 

\subsubsection{Holomorphicity of $Q$ and $R$ at the branch points}

The main claim in this subsection is 

\begin{prop} \label{prop:no-pole-theorem}
The coefficients $Q_m$ and $R_m$ of $Q$ and $R$ have no poles at branch points.
\end{prop}

\begin{proof}
First, Theorem \ref{thm:BPZ-type-equation} implies that 
\begin{equation}
\left[ 
\hbar^2 \frac{\partial^2}{\partial x^2} - 2 \hbar^2 \frac{\partial}{\partial t} 
- (4x^3+2tx) \right] 
\Bigr(Z(\nu+k\hbar ) \cdot \psi_\pm(\nu+k \hbar) \Bigl) = 0
\end{equation}
holds for any $k \in {\mathbb Z}$. 
Therefore, $\Psi_{\pm}$ are formal solutions of 
\begin{equation} \label{eq:partial-Lax-1}
\left[ \hbar^2 \frac{\partial^2}{\partial x^2} 
- 2 \hbar^2 \frac{\partial}{\partial t} 
- (4x^3+2tx+2H) \right] \Psi = 0.
\end{equation}
Then, it follows from the definition of $Q$ and $R$ that 
$\Psi_{\pm}$ satisfies the following system of PDEs: 
\begin{align}
\label{eq:JM-1-pre}
& \left[ \hbar^2 \frac{\p^2}{\p x^2} - 
2\hbar \Bigl( R \hbar \frac{\p}{\p x}  + Q \Bigr)
- (4x^3 + 2t x + 2 H) \right] \Psi = 0, \\
\label{eq:JM-2-pre}
& \left[ \hbar \frac{\p}{\p t} - \Bigl( R \hbar \frac{\p}{\p x}  + Q \Bigr)
\right] \Psi = 0.
\end{align}
Thus the above system of PDE is compatible; 
that is, $Q$ and $R$ satisfy
\begin{align}
\label{eq:compatibility-AB-1}
0 & = 2 \, \frac{\p Q}{\p x} - 2\hbar \, \frac{\p R}{\p t} 
+ 4 \hbar \, R \cdot \frac{\p R}{\p x} + \hbar \, \frac{\p^2 R}{\p x^2},   \\
\label{eq:compatibility-AB-2}
0 & = 2\, (4x^3+2tx+2H)  \cdot \frac{\p R}{\p x} 
+ 2 \, (6x^2+t) \cdot R - 2 \, (x-q) 
- 2 \hbar \, \frac{\p Q}{\p t} 
+ 4 \hbar \, Q \cdot \frac{\p R}{\p x} 
+ \hbar \, \frac{\p^2 Q}{\p x^2}.  
\end{align}
In particular, the leading terms $Q_0 = \tilde{Q}_0(x,t,v)|_{v=\frac{\phi}{\hbar}}$ 
and $R_0=\tilde{R}_0(x,t,v)|_{v=\frac{\phi}{\hbar}}$ satisfy 
\begin{align}
\label{eq:compatibility-AB-1-leading}
0 & 
= \left[ 2\frac{\p \tilde{Q}_0(x,t,v)}{\p x} 
- 2 \, \frac{\p \phi}{\p t} \cdot \frac{\p \tilde{R}_0(x,t,v)}{\p v} 
\right]_{v= \frac{\phi}{\hbar}},  \\
\label{eq:compatibility-AB-2-leading}
0 & 
= \left[ 2 \, (4x^3+2tx+2H_0) \cdot \frac{\p \tilde{R}_0(x,t,v)}{\p x} 
+ 2 \, (6x^2+t) \cdot \tilde{R}_0(x,t,v) - 2 \, (x-q_0) 
- 2 \, \frac{\p \phi}{\p t} \cdot \frac{\p \tilde{Q}_0(x,t,v)}{\p v}
\right]_{v= \frac{\phi}{\hbar}}.
\end{align}

Suppose for contradiction that the leading term $Q_0$ 
and $R_0$ have a pole at a branch point $x = e_1$; 
that is, 
suppose that the Laurent series expansion of 
$Q_0$ and $R_0$ at $x=e_1$ are given by 
\begin{equation}
\tilde{Q}_0(x,t,v) = \frac{a_0(t,v)}{(x-e_1(t))^{d}} +\cdots, 
\quad 
\tilde{R}_0(x,t,v) = \frac{b_0(t,v)}{(x-e_1(t))^{d'}} +\cdots 
\end{equation} 
with $d, d' \ge 1$ and $a_0(t,v)$, $b_0(t,v)$ being $x$-independent. 

It follows from \eqref{eq:compatibility-AB-1-leading}
that the pole orders must satisfy $d' = d+1$ since $e_1$ is independent of $v$. 
On the other hand, the pole order of 
\begin{equation}
2 \, (4x^3+2tx+2H_0) \cdot \frac{\p \tilde{R}_0(x,t,v)}{\p x} 
+ 2 \, (6x^2+t) \cdot \tilde{R}_0(x,t,v) 
= -\frac{(8d+4)(e_1-e_2)(e_1-e_3) b_0(t,v)}{(x-e_1(t))^{d+1}} + \cdots 
\end{equation}
is $d+1$, which is greater than that of ${\p \tilde{Q}_0(x,t,v)}/{\p v}$. 
This contradicts to \eqref{eq:compatibility-AB-2-leading}, 
and hence we can conclude that 
$Q_0$ and $R_0$ must be holomorphic at $x=e_1$.

Fix $m \ge 1$, and suppose that $Q_{k}$ and $R_{k}$ for $0 \le k \le m-1$ 
are holomorphic at $x=e_1$. Then, \eqref{eq:compatibility-AB-1} and 
\eqref{eq:compatibility-AB-2} imply that the terms
\begin{align}
& 
\left[ 2 \, \frac{\p \tilde{Q}_m(x,t,v)}{\p x}
- 2 \, \frac{\p \phi}{\p t} \cdot \frac{\p \tilde{R}_m(x,t,v)}{\p v} 
\right]_{v= \frac{\phi}{\hbar}} 
\label{eq:compatibility-AB-1-subleading}
\\
& \left[ 2 \, (4x^3+2tx+2H_0) \cdot \frac{\p \tilde{R}_m(x,t,v)}{\p x} 
+ 2\, (6x^2+t) \cdot \tilde{R}_m(x,t,v) 
- 2 \, \frac{\p \tilde{Q}_m(x,t,v)}{\p v}
\right]_{v= \frac{\phi}{\hbar}}
\label{eq:compatibility-AB-2-subleading}
\end{align}
are expressed in terms of polynomials of $Q_0, \dots, Q_{m-1}$, 
$R_0, \dots, R_{m-1}$ and their derivatives. 
Thus the both of \eqref{eq:compatibility-AB-1-subleading} and
\eqref{eq:compatibility-AB-2-subleading} are holomorphic at $x=e_1$ 
under the hypothesis. Then, by the same argument presented above, 
we can conclude that $Q_m$ and $R_m$ must be holomorphic  
at $x=e_1$ (and at other branch points $e_2$, $e_3$ by the same reason). 
Thus we have proved that the coefficients of $Q$ and $R$ do not have
poles at the branch points.
\end{proof}

\begin{rem}
The above proof suggests that 
a pole of $Q$ and $R$ must be 
$x=\infty$, or a point depending on $\hbar$.
Actually, they have poles at $x=q_0(t, \hbar)$ 
which is a leading term of $q$ given in \eqref{eq:q0-leading-term}.
The pole arises as the zero of the leading term of the 
Wronskian (see \eqref{eq:expressions-of-Wronskians} below). 
\end{rem}

\subsubsection{Polynomiality of the Wronskians}

It follows from \eqref{eq:partial-Lax-1} that 
the Wronskians \eqref{eq:def-of-two-Wronskians}
satisfy 
\begin{equation}
\frac{\p}{\p x} {\rm Wr}^x[{\Psi}_{+}, {\Psi}_{-}] 
= 2 \, {\rm Wr}^t[{\Psi}_{+}, {\Psi}_{-}],
\end{equation} 
and hence, 
\begin{equation}
\label{eq:log-der-of-x-Wronskian}
R = \frac{1}{2} \frac{\p}{\p x} \log {\rm Wr}^x[{\Psi}_{+}, {\Psi}_{-}].
\end{equation}

Theorem \ref{thm:formal-monodromy-Psi} and 
Proposition \ref{prop:no-pole-theorem} implies 
\begin{prop} \label{prop:polynomiality-of-Wronskian}
The coefficients of the Wronskians ${\rm Wr}^x$ and ${\rm Wr}^t$ 
are polynomial of $x$. 
\end{prop}

\begin{proof}
Theorem \ref{thm:formal-monodromy-Psi}
imply that the coefficients of the Wronskians have no monodromy along any 
cycle on the spectral curve. Therefore, they are rational functions of $x$.
Since the coefficients of $\Psi_{\pm}$ has a singularity at 
$x=\infty$ and the branch points $e_1, e_2, e_3$, 
the coefficients of Wronskians may have poles at these points. 
However, \eqref{eq:log-der-of-x-Wronskian} and 
Proposition \ref{prop:no-pole-theorem} show that 
the coefficients of Wronskians do not have poles at the branch points. 
Therefore, the coefficients of Wronskian is holomorphic except for $x=\infty$, 
and hence they are polynomials of $x$. 
This completes the proof.
\end{proof}

\subsubsection{Completion of the proof of main theorems}

It follows from the asymptotics \eqref{eq:behavior-psi-pm} that 
\begin{multline}
{\rm Wr}^{x}[\Psi_+, \Psi_-] = 4x - 4 q 
- \frac{1}{3} \left( \hbar \frac{dp}{dt} - (6q^2+t) \right)  x^{-1} \\ 
+ \left( \frac{2 \, (2H-p^2+4q^3+2tq)}{15} 
+ \frac{q}{5} \Bigl( \hbar \frac{dp}{dt} - (6q^2+t) \Bigr) 
- \frac{\hbar^2}{90} \cdot \frac{d^2}{dt^2} 
\Bigl(\hbar \frac{dp}{dt} - (6q^2+t) \Bigr) \right) x^{-2}
+ O(x^{-3})
\end{multline}
holds as $x \to \infty$. 
Since the coefficients of the Wronskian is polynomial in $x$, 
the coefficients of negative powers of $x$ must vanish; that is, we have
\begin{equation}
\hbar \frac{dp}{dt} - (6q^2+t) = 0, 
\quad
2H-p^2 +4q^3+2tq = 0.
\end{equation}
These relations show that the formal series $q$ 
given in \eqref{eq:hbar-expansion-of-q}
satisfies the Painlev\'e I equation $(P_{\rm I})$, 
and in particular, $\tau_{P_{\rm I}}$ is the corresponding $\tau$-function. 
This completes the proof of Theorem \ref{thm:tau-function-theorem}.

Furthermore, the above asymptotic behavior and the polynomiality 
of the Wronskians imply 
\begin{equation} \label{eq:expressions-of-Wronskians}
{\rm Wr}^{x}[\Psi_+, \Psi_-] = 4 \, (x - q), \quad
{\rm Wr}^{t}[\Psi_+, \Psi_-] = 2.
\end{equation}
Thus we have 
\begin{equation} \label{eq:formula-R}
R = \frac{1}{2} \cdot 
\frac{ \p_x {\rm Wr}^x[{\Psi}_{+}, {\Psi}_{-}] }
{ {\rm Wr}^x[{\Psi}_{+}, {\Psi}_{-}] } =
 \frac{1}{2 \, (x-q)}
\end{equation}
and 
\begin{equation} \label{eq:formula-Q}
Q 
= \frac{\hbar \, \p_{t} {\rm Wr}^{x}[\Psi_+, \Psi_-] 
- \hbar \, \p_{x} {\rm Wr}^{t}[\Psi_+, \Psi_-]}
{2 \, {\rm Wr}^{x}[\Psi_+, \Psi_-]} = - \frac{p}{2 \, (x-q)}.
\end{equation}
Thus the system of PDEs \eqref{eq:JM-1-pre}-\eqref{eq:JM-2-pre} 
satisfied by $\Psi_{\pm}$ coincides with the isomonodromy system 
\eqref{eq:JM-1}-\eqref{eq:JM-2} associated with $(P_{\rm I})$. 
This completes the proof of Theorem \ref{thm:wave-function-theorem}.

\section{Exact WKB theoretic approach to the direct monodoromy problem}
\label{section:Exact-WKB}

In this section, we give a heuristic discussion on 
the {\it direct monodromy problem} (i.e., computation of the Stokes multipliers
of the linear ODE $(L_{\rm I})$) associated with $\PI$. 
Those Stokes multipliers should be independent of $t$ 
due to the isomonodromic property. 
Unfortunately, our computation is not mathematically rigorous
at this moment because it is based on two conjectures 
on Borel summability and the connection formula on Stokes curves 
(Conjecture \ref{conj:Borel-summability} 
and Conjecture \ref{conj:Voros-formula}). 
We hope that those conjectures are solved in the future.

Our method is a fusion of the exact WKB theoretic 
computation of the monodromy/Stokes matrices 
developed by \cite{SAKT} (see \cite[\S 3]{KT05}) 
and the discrete Fourier transform. 
We are motivated by a technique used in 
the conformal field theoretic construction 
by \cite{ILT} etc. 

\subsection{Stokes graph}

As well as the cases of Schr\"odinger-type ODEs, 
we expect that the Stokes graph controls the 
global behavior of the (Borel resummed) WKB solution 
\eqref{eq:WKB-series-B} of the PDE \eqref{eq:BPZ-equation}.
For a fixed $(t,\nu) \in D_{t_\ast} \times D_{\nu_\ast}$, 
we introduce {\em Stokes curves} by the following equality:
\begin{equation} \label{eq:Stokes-curve-def}
{\rm Im} \int^{x}_{e_i} \sqrt{4x^3+2tx+u(t,\nu)} \, dx  = 0 
\quad (i=1,2,3).
\end{equation}
Stokes curves are the trajectories of the quadratic differential 
$(4x^3+2tx+u(t,\nu))dx^{\otimes2}$ (see \cite{Strebel} for properties 
of trajectories of quadratic differentials). 
Here, recall that $e_1, e_2, e_3$ are branch points 
of the spectral curve \eqref{eq:spcurve}. 
Since the branch points are simple zeros of the quadratic differential
for any $(t,\nu) \in D_{t_\ast} \times D_{\nu_\ast}$, 
three Stokes curves emanate from each $e_i$. 
The branch points, the infinity, and the Stokes curves 
form a graph on ${\mathbb P}^1$, which is called the {\em Stokes graph} 
(for a fixed $(t,\nu) \in D_{t_\ast} \times D_{\nu_\ast}$). 
Interior of a face of the Stokes graph is called a Stokes region. 
A Stokes segment (or a saddle connection) 
is a Stokes curve connecting two branch points. 

\begin{rem} \label{rem:approximated-Stokes}
Below we will draw ``approximated" Stokes curves 
which are defined as the Stokes curves with $u(t, \nu)$ 
in \eqref{eq:Stokes-curve-def} being replaced by 
\begin{equation}
u_{\rm app}(t, \nu) := 
\frac{s^2}{9t} + \frac{2 i \nu  s}{t} - \frac{5\nu^2}{4t}, 
\quad
s = 24^{1/4} \cdot (-t)^{5/4}.
\end{equation}
We choose the function $u_{\rm app}$ referencing the 
first few terms of the following asymptotic behavior 
of $F_0$ near $t=\infty$ computed in \cite[\S 3.1]{BLMST}:
\begin{equation} \label{eq:asymptotics-F0-infinity}
F_0(t,\nu) = \frac{s^2}{45} + \frac{4}{5} i \nu s 
+ \frac{\nu^2}{2} \cdot \log \frac{\nu}{48 i s} 
- \frac{3}{4} \nu^2 - \frac{47i\nu^3}{48s} 
- \frac{7717\nu^4}{4608s^2} + O(s^{-3})
\end{equation}
when $|t| \to +\infty$.
Thus, we may expect that the approximated Stokes 
graph is close to the actual Stokes graph 
if we choose sufficiently large $t$. 
\end{rem}

\subsection{Two conjectures}

Here we state two conjectures on the analytic properties 
of the WKB solution $\psi_{\pm}(x,t,\nu;\hbar)$ of the PDE \eqref{eq:BPZ-equation}.
These conjectures are true in the case of the usual Schr\"odinger equations 
(i.e., second order ODEs with a small parameter $\hbar$) with 
a polynomial or rational potential function (c.f., \cite{DLS, Koike-Schafke}).

\subsubsection{Borel summability}

In the theory of the exact WKB analysis, 
the Stokes graph for a Schr\"odinger-type ODE 
is used to describe the regions where the WKB solutions 
are Borel summable as formal power series of $\hbar$ 
(c.f., \cite{KT05, Voros83}). 
We expect that the same results for 
the Borel summability established by 
\cite{DLS, Koike-Schafke} 
also hold for the WKB solution of the PDE \eqref{eq:BPZ-equation}.

\begin{conj} \label{conj:Borel-summability} \hspace{+.0em}
\begin{itemize}
\item[(i)]
Fix $(t,\nu) \in D_{t_\ast} \times D_{\nu_\ast}$, and suppose 
that there is no Stokes segment in the Stokes graph for $(t,\nu)$.
Then, for any fixed $x$ on each Stokes region $U$,
the WKB solution $\psi_{\pm}(x,t,\nu;\hbar)$ given in \eqref{eq:WKB-series-B}
is Borel summable as a formal power series of $\hbar$. 
In this situation, the Borel sum of $\psi_{\pm}(x,t,\nu;\hbar)$ is 
a holomorphic solution of the PDE \eqref{eq:BPZ-equation} 
defined on $U \times D_{t_\ast}$.

\smallskip
\item[(ii)]
For a fixed $(t,\nu) \in D_{t_\ast} \times D_{\nu_\ast}$,
the formal series $Z(t,\nu;\hbar)$ given in \eqref{eq:partition-function-Z} 
is Borel summable as a formal power series of $\hbar$ 
if the Stokes graph for $(t,\nu)$ has no Stokes segment.

\end{itemize}
\end{conj}

See \cite{Costin08, Sauzin} for the Borel summation. 
In what follows we assume that Conjecture \ref{conj:Borel-summability} is true. 

\begin{rem}
We can verify that the both of the claims (i) and (ii)
in Conjecture \ref{conj:Borel-summability} are true 
for the WKB series and partition function for the 
spectral curves arising from the Gauss hypergeometric 
equations and its confluent equations. 
Those spectral curves were studied in \cite{IKoT1, IKoT2} in detail, 
where explicit formulas (in terms of Bernoulli numbers)
of the free energy were obtained (some of the results were known from before). 
In a standard argument in the exact WKB analysis, 
Stokes segments yield certain singularities on the Borel plane 
(c.f., \cite{DDP93, DP99, Voros83} etc.), and in particular, 
the WKB solutions and the Voros coefficients are not 
Borel summable when a Stokes segment exists. 
As is shown in \cite{IKoT1, IKoT2}, the Voros coefficients 
for those examples are written in terms of the free energy, 
and hence, we cannot expect the Borel summability of
the partition function $Z$ in such a situation. 

It is also worth mentioning that, as is discussed by 
\cite{DDP93, Voros83}, 
a Stokes segment yields a kind of Stokes phenomenon for the WKB solutions 
and the Voros coefficients. 
It is called the 
parametric Stokes phenomenon and studied 
by \cite{AT, I13, Koike-Takei, Takei-Sato-conj} etc.\ in detail. 
In some cases the parametric Stokes phenomenon has 
a close relationship to the cluster exchange relation 
(\cite{IN14, IN15}). 

\end{rem}

\subsubsection{Voros connection formula}

We also expect that the Voros' connection formula 
(\cite[\S 6]{Voros83}, \cite[\S 2]{KT05})
holds for the WKB solution $\psi_{\pm}(x,t,\nu;\hbar)$ 
of the PDE \eqref{eq:BPZ-equation}.

\begin{conj} \label{conj:Voros-formula}
Fix $(t,\nu) \in D_{t_\ast} \times D_{\nu_\ast}$, so that 
there is no Stokes segment in the Stokes graph.  
Suppose that a Stokes curve $C$, connecting a branch point $e$ and 
$\infty$, is a common boundary of two Stokes regions $I$ and $II$. 
We label the Stokes regions so that the region $II$ comes 
next to the region $I$ when we turn around the branch point $e$
in the counterclockwise direction (see Figure \ref{fig:Voros-formula}). 
If we denote by $\psi_{\pm}^{J}$ the Borel sum of 
$\psi_{\pm}$ defined on the region $J \in \{I, II\}$, 
then one of the following relations holds:

\begin{itemize}
\item[(i)] 
If ${\rm Re} \int^{x}_{e} \sqrt{4x^3+2tx+u(t,\nu)} \, dx > 0$ on $C$, 
then  
\begin{eqnarray} \label{eq:Voros-formula-plus}
\begin{cases}
\psi^{I}_{+}(x,t,\nu;\hbar) & = \quad 
\psi^{II}_{+}(x,t,\nu;\hbar) + \tilde{\psi}^{II}_{-}(x,t,\nu;\hbar)  \\[+.3em]
\psi^{I}_{-}(x,t,\nu;\hbar) & = \quad \psi^{II}_{-}(x,t,\nu;\hbar).
\end{cases}
\end{eqnarray}
\item[(ii)]
If ${\rm Re} \int^{x}_{e} \sqrt{4x^3+2tx+u(t,\nu)} \, dx < 0$ on $C$, 
then  
\begin{eqnarray}
\label{eq:Voros-formula-minus}
\begin{cases}
\psi^{I}_{+}(x,t,\nu;\hbar) & = \quad \psi^{II}_{+}(x,t,\nu;\hbar) \\[+.3em]
\psi^{I}_{-}(x,t,\nu;\hbar) & = \quad \psi^{II}_{-}(x,t,\nu;\hbar)
+ \tilde{\psi}^{II}_{+}(x,t,\nu;\hbar).  
\end{cases}
\end{eqnarray}
\end{itemize}
Here $\tilde{\psi}^{II}_{\mp}(x,t,\nu;\hbar)$ in the above formula 
are the Borel sum of a formal series 
$\tilde{\psi}_{\mp}(x,t,\nu;\hbar)$ in the region $II$, 
which is obtained by the term-wise analytic continuation 
of the original WKB solution $\psi_{\pm}(x,t,\nu;\hbar)$ 
along a ``detoured path encircling the branch point $e$" 
shown in Figure \ref{fig:Voros-formula}. 
More precisely, $\tilde{\psi}_{\mp}(x,t,\nu;\hbar)$ 
is the term-wise analytic continuation of $\psi_{\pm}(x,t,\nu;\hbar)$
along a path starting from a point in the region $I$ 
to a point in the region $II$ which turns around the 
branch point $e$ in the clockwise direction.
\end{conj}

\begin{figure}[t]
  \begin{center}
   \includegraphics[width=40mm]{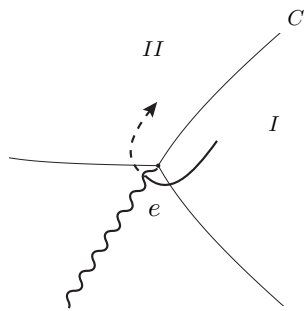}
  \end{center}
   \caption{The detoured path which encircles the branch point $e$.}
  \label{fig:Voros-formula}
\end{figure}

Since $e$ is a branch point of the spectral curve, 
the term-wise analytic continuation of $\psi_{\pm}$ 
along the detoured path has the phase function 
with the different sign from the original one. 
To visualize the two sheets of the spectral curve, 
we draw the branch cut (the wiggly line) and 
use the solid (resp., dashed) lines to represent 
a part of the path on the first (resp., second) sheet. 
We will use the same rule when we draw paths on the spectral curve.

The connection formula guarantees the single-valuedness 
of the Borel resummed WKB solutions around branch points
(see \cite[\S 6]{Voros83}). 
We also note that the above formula also has a close relationship 
to the ``path-lifting" in the work of Gaiotto-Moore-Neitzke 
\cite{GMN12}, where Stokes graphs are called spectral networks.

Regarding on the term-wise analytic continuation along the detoured path, 
we can prove the following (without assuming any conjecture): 

\begin{prop} \label{prop:detoured-continuation}
Under the same situation as in Conjecture \ref{conj:Voros-formula}, 
take a point $x_o$ on the Stokes curve $C$ and an open 
neighborhood $U_o$ of $x_o$ contained in $I \cup C \cup II$. 
Let $\psi_{\pm}(x,t,\nu; \hbar)$ be the WKB solution defined by 
\eqref{eq:WKB-series-B} for $x \in U_o$
with the integration path from $0$ to $z(x)$ 
in \eqref{eq:WKB-series-B} being chosen as the 
image by $x \mapsto z(x)$ of the composition of the following two paths on $x$-plane: 
The part of the Stokes curve $C$ connecting $\infty$ and $x_o$, and 
any path from $x_o$ to $x$ contained in $U_o$. 
Then, the term-wise analytic continuation $\tilde{\psi}_{\mp}(x,t,\nu; \hbar)$ 
of $\psi_{\pm}(x,t,\nu; \hbar)$ along the detoured path encircling $e$ is given by 
\begin{equation} \label{eq:formal-continuation-detoured}
\tilde{\psi}_{\mp}(x,t,\nu; \hbar) = 
i \cdot (-1)^{m_C n_C} \cdot \exp\left(  
\pm \frac{2 \pi i {m_C} \nu}{\hbar} \right) \cdot 
\frac{Z(t,\nu \pm n_C \hbar; \hbar)}{Z(t,\nu; \hbar)} \cdot 
\psi_{\mp} (x,t,\nu \pm n_C \hbar; \hbar).
\end{equation}
Here $m_C$ and $n_C$ are integers defined by the condition 
\begin{equation} \label{eq:half-period-along-C}
\int^{e}_{\infty} \frac{dx}{\sqrt{4x^3+2tx+u(t,\nu)}} = 
\frac{m_C \, \omega_A + n_C \, \omega_B}{2},
\end{equation} 
where the left hand-side of \eqref{eq:half-period-along-C} 
is defined as the integration along the Stokes curve $C$.  
(Thus, either $m_C$ or $n_C$ is an odd integer.)
\end{prop}

\begin{proof}
When $x$ turns around $e$ along the detoured path, 
the function $z(x)$, whose defining path of integration in \eqref{eq:WKB-series-B}
is chosen as the same path for the normalization of $\psi_{\pm}(x,t,\nu; \hbar)$, 
admits the monodromy 
\begin{equation}
z(x) \mapsto - z(x) + m_C \, \omega_A + n_C \, \omega_B.
\end{equation}
Then, the result \eqref{eq:formal-continuation-detoured} can be proved 
by a similar computation done in the proof of Theorem \ref{thm:formal-monodromy-psi}. 
\end{proof}

\subsection{Computation of Stokes multipliers of $(L_{\rm I})$}
\label{subsection:computation-of-Stokes}

Assuming Conjectures \ref{conj:Borel-summability} and 
\ref{conj:Voros-formula} (and the convergence of the 
Fourier series \eqref{eq:wave-function} after taking the Borel summation), 
we demonstrate how the exact WKB method computes 
the Stokes multipliers of $(L_{\rm I})$. 
We will apply the Voros connection formula for $\psi_{\pm}$, 
and take the discrete Fourier transform to compute the 
Stokes multipliers for $\Psi_{\pm}$.

\begin{figure}[t]
 \begin{minipage}{0.4\hsize}
  \begin{center}
   \includegraphics[width=50mm]{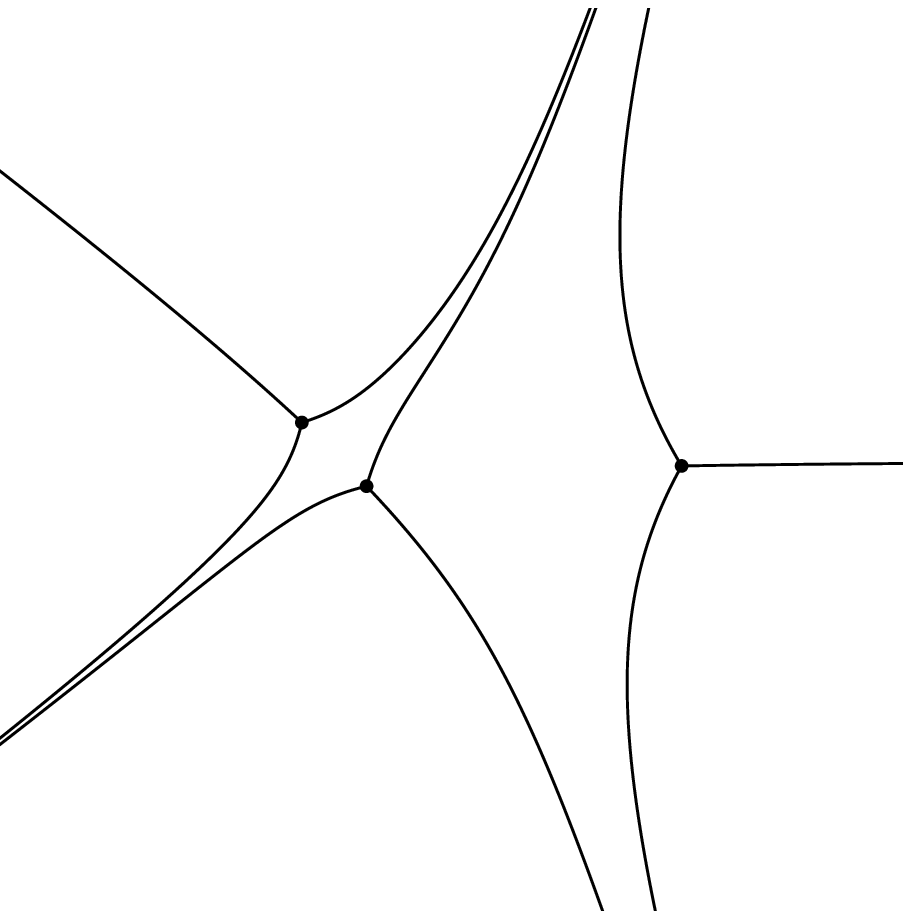} \\
   {(a)}
  \end{center}
 \end{minipage}
 \begin{minipage}{0.4\hsize}
  \begin{center}
   \includegraphics[width=50mm]{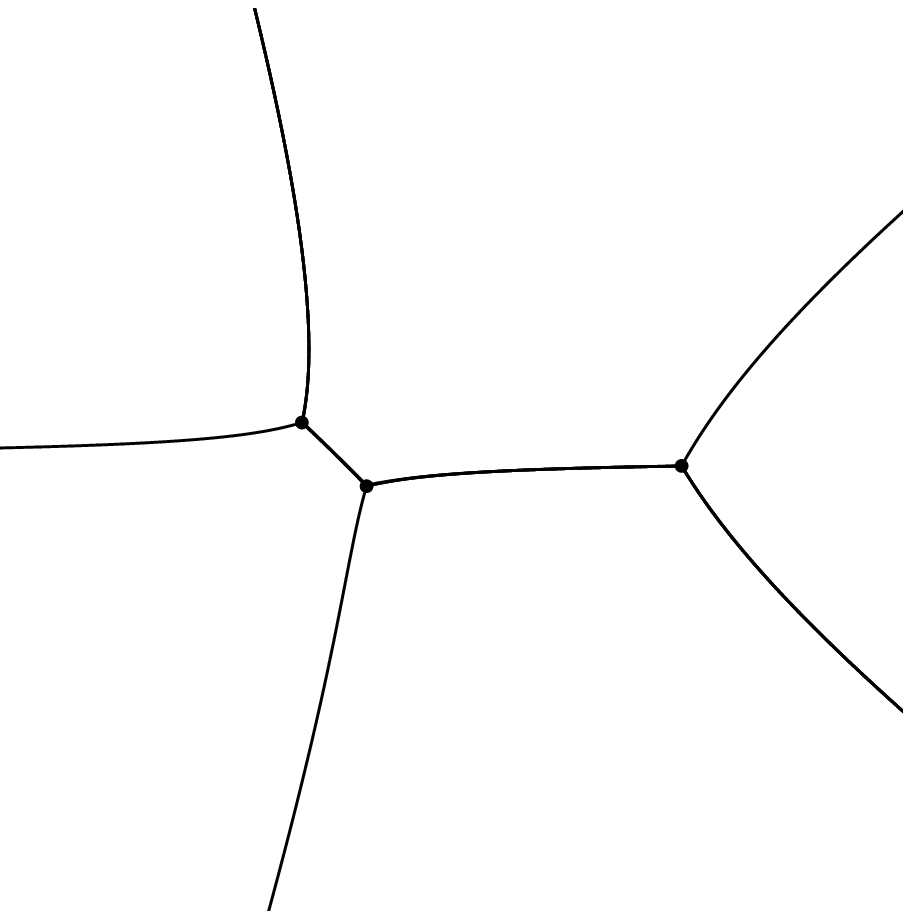} \\
   {(b)}
  \end{center}
 \end{minipage}
   \caption{(a) The Stokes curves  for $t = -9.9313 + 1.17017 i$ and $\nu=1/2$. 
   (b) The anti-Stokes curves for the same $t$ and $\nu$. }
  \label{fig:Stokes-Boutroux-1}
\end{figure}

\begin{figure}[t]
  \begin{center}
   \includegraphics[width=90mm]{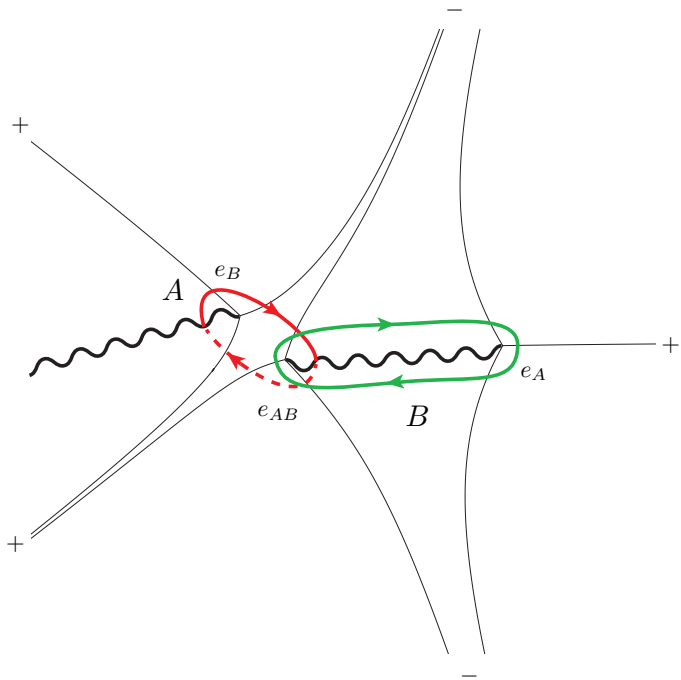}
  \end{center}
   \caption{The $A$-cycle (red) and the $B$-cycle (green).}
  \label{fig:A-and-B-cycles}
\end{figure}

We choose $t = -9.9313 + 1.17017 i$ and $\nu = 1/2$ for our computation. 
The (approximated) Stokes graph is shown in Figure \ref{fig:Stokes-Boutroux-1} (a). 
We also draw the (approximated) anti-Stokes curves, which are 
defined by the condition 
${\rm Re} \int^{x}_{e_i} \sqrt{4x^3+2tx+u(t,\nu)} \, dx = 0$, 
in Figure \ref{fig:Stokes-Boutroux-1} (b). 
We can observe that the branch points are connected by 
the anti-Stokes curves simultaneously; 
this means that the Boutroux condition
\begin{equation}
{\rm Re} \oint_\gamma \sqrt{4x^3+2tx+u(t,\nu)} \, dx  = 0
\quad 
\text{for any closed path $\gamma$}
\end{equation}
is satisfied. 
We expect that the condition will play a role 
when we discuss the convergence of 2-parameter solutions.  
This is the reason why we choose the above parameters $t$ and $\nu$.


After drawing branch cuts (wiggly lines) 
as in Figure \ref{fig:A-and-B-cycles}, 
we take the branch so that \eqref{eq:branch-of-y} 
holds on the first sheet when $x \to \infty$ with ${\rm Re} \, x > 0$.
The symbols $\pm$ on the Stokes curves in Figure \ref{fig:A-and-B-cycles}
represent the sign of ${\rm Re} \int^{x}_{e_i} \sqrt{4x^3+2t x + u(t,\nu)} \, dx$
on the Stokes curves. 
The $A$-cycle and $B$-cycle, which satisfy
\eqref{eq:nu} and \eqref{eq:dF0-dt}
are also indicated in the same figure. 
We choose the labeling $e_A, e_B, e_{AB}$ of the branch points so that 
the following relations hold modulo $\Lambda$:
\begin{equation} \label{eq:half-periods-normalization}
\int^{e_A}_{\infty} \frac{dx}{y} \equiv \frac{\omega_A}{2},\quad
\int^{e_B}_{\infty} \frac{dx}{y} \equiv \frac{\omega_B}{2},\quad
\int^{e_{AB}}_{\infty} \frac{dx}{y} \equiv \frac{\omega_A + \omega_B}{2}.
\end{equation}

There are five asymptotic directions $\arg x = 2\ell \pi/5$  ($\ell=0, \pm 1, \pm 2$) of 
Stokes curves near the infinity.  Denote by $s_\ell$ ($\ell=0, \pm 1, \pm 2$) 
the Stokes multipliers of $(L_{\rm I})$ corresponding to these directions, 
as indicated in Figure \ref{fig:Stokes-multipliers}. 
More precisely, these constants are determined by 
\begin{equation}
(\Psi^{(\ell)}_{+}, \Psi^{(\ell)}_{-}) = 
\begin{cases}
\displaystyle (\Psi^{(\ell+1)}_{+}, \Psi^{(\ell+1)}_{-}) \cdot 
\begin{pmatrix} 1 & 0 \\ s_\ell & 1 \end{pmatrix} 
& \quad \text{for $\ell=0, \pm 2$,} \\[+1.3em]
\displaystyle (\Psi^{(\ell+1)}_{+}, \Psi^{(\ell+1)}_{-}) \cdot 
\begin{pmatrix} 1 & s_\ell \\ 0 & 1 \end{pmatrix} 
& \quad \text{for $\ell=\pm 1$,} 
\end{cases}
\end{equation}
where
\begin{equation} \label{eq:Borel-resummed-Psi}
\Psi_{\pm}^{(\ell)} = \frac{
\sum_{k \in {\mathbb Z}} e^{2\pi i k \rho/\hbar} 
\cdot Z(t,\nu+k \hbar; \hbar) \cdot \psi_{\pm}^{(\ell)}(x,t,\nu+k \hbar; \hbar)
}{
\sum_{k \in {\mathbb Z}} e^{2\pi i k \rho/\hbar} \cdot Z(t,\nu+k \hbar; \hbar) 
}
\end{equation}
with 
$\psi_{\pm}^{(\ell)}$ being the Borel sum of $\psi_\pm$ on the Stokes region $\ell$ 
specified in Figure \ref{fig:Stokes-multipliers}.
($Z$ in the right hand-side of \eqref{eq:Borel-resummed-Psi} is also understood 
to be the Borel sum. We use the same symbol for simplicity.)

To compute $s_\ell$, we employ the connection formulas 
\eqref{eq:Voros-formula-plus}-\eqref{eq:Voros-formula-minus} 
and Proposition \ref{prop:detoured-continuation}. 
For the convenience, we assign labels to the Stokes curves 
in Figure \ref{fig:Stokes-multipliers} by the following rule: 
The Stokes curve $C_{\ast, \ell}$ for $\ast \in \{A, B, AB \}$ 
and $\ell \in \{0, \pm1, \pm2\}$ is the one emanating 
from the branch point $e_{\ast}$ and flowing to $\infty$ 
with the asymptotic direction $\arg x = 2\ell \pi/5$.
The pairs of integers $(m_C, n_C)$ which are assigned 
for the Stokes curves by \eqref{eq:half-period-along-C} 
are summarized in Table \ref{table:pair-of-integers-for-C}.
These integers can be computed from the intersection numbers 
of the detoured path (regarded as a relative homology class 
on the spectral curve) and the $A$-cycle and $B$-cycle.

\begin{figure}[t]
  \begin{center}
   \includegraphics[width=100mm]{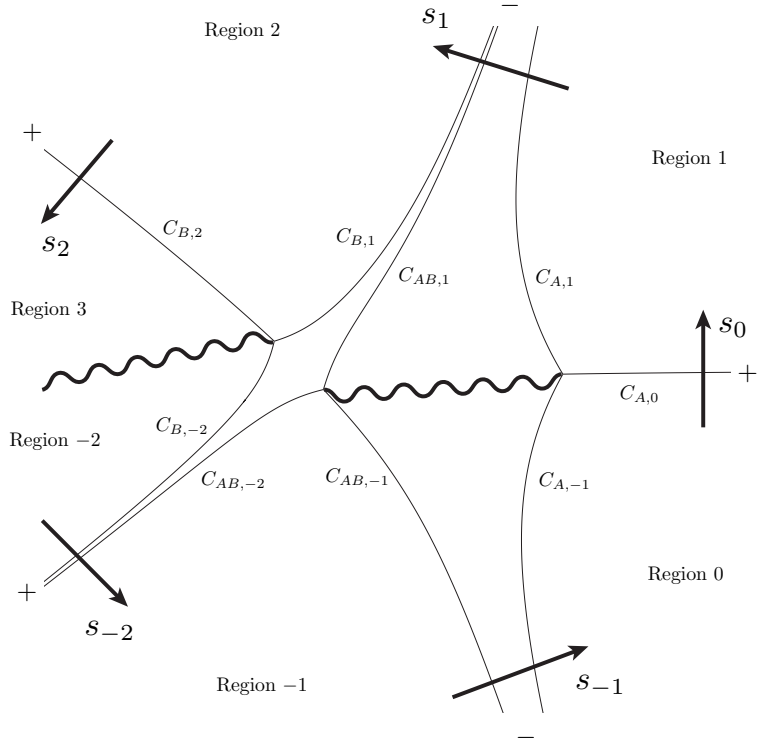} 
  \end{center}
   \caption{The Stokes multipliers of $(L_{\rm I})$ around $x=\infty$. }
  \label{fig:Stokes-multipliers}
\end{figure}

\begin{table}[t]
  \begin{tabular}{|c||c|c|c|c|c|c|c|c|c|c|} \hline
    Stokes curve $C$
    & $C_{A,0}$ & $C_{A,1}$ & $C_{A,-1}$ & $C_{B,1}$ & $C_{B,2}$ & $C_{B,-2}$ 
    & $C_{AB,1}$ & $C_{AB,-1}$ & $C_{AB,-2}$ \\ \hline \hline
    $(m_C, n_C)$
    & $(1,0)$ & $(1,0)$ & $(1,0)$ & $(0,-1)$ & $(0,-1)$ & $(0,1)$
    & $(1,-1)$ & $(1,1)$ & $(1,1)$ \\ \hline
  \end{tabular} 
   \vspace{+1.em}
     \caption{The pair $(m_C, n_C)$ for Stokes curves in Figure \ref{fig:Stokes-multipliers}.}
     \label{table:pair-of-integers-for-C}
  \end{table}

\smallskip
Under the above preparation, 
let us compute the Stokes multipliers.

\vspace{+.5em} \noindent
$\bullet$ \underline{\bf Computation of $s_0$.} \, 
Since the sign of the Stokes curve $C_{A,0}$ is plus, the WKB solution 
$\psi_+$ is dominant. Hence, the Voros formula implies
\begin{equation}
\begin{cases}
\psi^{(0)}_{+}(x,t,\nu;\hbar) & = \quad 
\psi^{(1)}_{+}(x,t,\nu;\hbar) + \tilde{\psi}^{(1)}_{-}(x,t,\nu;\hbar) \\[+.5em]
\psi^{(0)}_{-}(x,t,\nu;\hbar) & = \quad \psi^{(1)}_{-}(x,t,\nu;\hbar).
\end{cases}
\end{equation}
Since $(m_{C_{A,0}}, n_{C_{A,0}}) = (1,0)$ as shown 
in Table \ref{table:pair-of-integers-for-C}, 
Proposition \ref{prop:detoured-continuation} implies 
that the analytic continuation $\tilde{\psi}_{-}$ 
along the detoured path encircling $e_A$ is given by 
\begin{equation}
\tilde{\psi}^{(1)}_{-}(x,t,\nu;\hbar) = 
i \cdot e^{2\pi i \nu/\hbar} \cdot \psi^{(1)}_{-}(x,t,\nu; \hbar). 
\end{equation} 
By shifting $\nu \mapsto \nu + k \hbar$ and taking 
the discrete Fourier transform, we have
\begin{align}
\Psi_{+}^{(0)}(x,t,\nu,\rho; \hbar) & = 
\frac{
\sum_{k \in {\mathbb Z}} e^{2\pi i k \rho/\hbar} 
\cdot Z(t,\nu+k \hbar; \hbar) \cdot \psi_{+}^{(0)}(x,t,\nu+k \hbar; \hbar)
}{
\sum_{k \in {\mathbb Z}} e^{2\pi i k \rho/\hbar} \cdot Z(t,\nu+k \hbar; \hbar) 
} \notag \\
& = 
\frac{
\sum_{k \in {\mathbb Z}} e^{2\pi i k \rho/\hbar} 
\cdot Z(t,\nu+k \hbar; \hbar) \cdot \psi_{+}^{(1)}(x,t,\nu+k \hbar; \hbar)
}{
\sum_{k \in {\mathbb Z}} e^{2\pi i k \rho/\hbar} \cdot Z(t,\nu+k \hbar; \hbar) 
} \notag \\ 
& \quad +
\frac{
\sum_{k \in {\mathbb Z}} e^{2\pi i k \rho/\hbar} 
\cdot Z(t,\nu+k \hbar; \hbar) \cdot i \cdot 
e^{2 \pi i (\nu+k \hbar)/\hbar } \cdot  \psi_{-}^{(1)}(x,t,\nu+k \hbar; \hbar)
}{
\sum_{k \in {\mathbb Z}} e^{2\pi i k \rho/\hbar} \cdot Z(t,\nu+k \hbar; \hbar) 
} \notag \\ 
& = \Psi_{+}^{(1)}(x,t,\nu,\rho; \hbar) 
+ i \cdot e^{2 \pi i \nu/\hbar} \cdot \Psi_{-}^{(1)}(x,t,\nu,\rho; \hbar),
\end{align}
and 
\begin{equation}
\Psi_{-}^{(0)}(x,t,\nu,\rho; \hbar)  = \Psi_{-}^{(1)}(x,t,\nu,\rho; \hbar).
\end{equation}
We have used the periodicity 
$e^{2 \pi i (\nu+k \hbar)/\hbar} = e^{2\pi i \nu /\hbar }$ of the exponential function.
From this computation, we conclude 
\begin{equation}
s_0 = i \cdot e^{2\pi i \cdot \nu/\hbar}.
\end{equation}

\vspace{+.5em}\noindent
$\bullet$ \underline{Computation of $s_{-1}$.}  
In this case, we need to cross two Stokes curves 
$C_{AB, -1}$ and $C_{A, -1}$ passing an intermediate 
Stokes region between them. 
Both Stokes curves have the minus sign, so the WKB solution 
$\psi_{-}$ is dominant. 

Let us denote by $\psi_{\pm}^{(-1/2)}$ the Borel sum 
of the WKB solution $\psi_{\pm}$ on the intermediate 
Stokes region. Then, the Voros formula and 
Proposition \ref{prop:detoured-continuation} imply 
\begin{equation}
\begin{cases}
\psi^{(-1)}_{+}(x,t,\nu;\hbar) &=\quad \psi^{(-1/2)}_{+}(x,t,\nu;\hbar)
\\[+.5em]
\psi^{(-1)}_{-}(x,t,\nu;\hbar) &=\quad \psi^{(-1/2)}_{-}(x,t,\nu;\hbar) 
+ \tilde{\psi}^{(-1/2)}_{+}(x,t,\nu;\hbar)
\end{cases}
\end{equation} 
holds on the Stokes curve $C_{AB,-1}$ with 
\begin{equation}
\tilde{\psi}^{(-1/2)}_{+}(x,t,\nu;\hbar) =  
- i \cdot e^{- 2 \pi i \nu/\hbar} \cdot 
\frac{Z(t,\nu-\hbar;\hbar)}{Z(t,\nu; \hbar)} \cdot 
\psi^{(-1/2)}_{+}(x,t,\nu- \hbar; \hbar).
\end{equation}
Similarly, on the Stokes curve $C_{A,-1}$, we have
\begin{equation}
\begin{cases}
\psi^{(-1/2)}_{+}(x,t,\nu;\hbar) &=\quad \psi^{(0)}_{+}(x,t,\nu;\hbar) 
\\[+.5em]
\psi^{(-1/2)}_{-}(x,t,\nu;\hbar) &=\quad \psi^{(0)}_{-}(x,t,\nu;\hbar) 
+ \tilde{\psi}^{(0)}_{+}(x,t,\nu;\hbar)
\end{cases}
\end{equation} 
with 
\begin{equation}
\tilde{\psi}^{(0)}_{+}(x,t,\nu;\hbar) =  
i \cdot e^{- 2 \pi i \nu/\hbar} \cdot 
\psi^{(0)}_{+}(x,t,\nu; \hbar).
\end{equation}
Summarizing these formulas, we have
\begin{align}
\psi^{(-1)}_{+}(x,t,\nu;\hbar) & = \psi^{(0)}_{+}(x,t,\nu;\hbar) \\
\psi^{(-1)}_{-}(x,t,\nu;\hbar) & = \psi^{(0)}_{-}(x,t,\nu;\hbar)
- i \cdot e^{- 2 \pi i \nu/\hbar} \cdot 
\frac{Z(t,\nu-\hbar;\hbar)}{Z(t,\nu; \hbar)} \cdot 
\psi^{(0)}_{+}(x,t,\nu - \hbar; \hbar)
\notag  \\ 
& \qquad 
+ i \cdot e^{- 2 \pi i \nu/\hbar} \cdot 
\psi^{(0)}_{+}(x,t,\nu; \hbar).
\end{align}
Again, by taking the discrete Fourier transform, we have
\begin{equation}
\Psi^{(-1)}_{+}(x,t,\nu,\rho;\hbar) = \Psi^{(0)}_{+}(x,t,\nu,\rho;\hbar)
\end{equation}
and 
\begin{align}
\Psi^{(-1)}_{-}(x,t,\nu,\rho;\hbar) & = 
\frac{
\sum_{k \in {\mathbb Z}} e^{2\pi i k \rho/\hbar} 
\cdot Z(t,\nu+k \hbar; \hbar) \cdot \psi_{-}^{(0)}(x,t,\nu+k \hbar; \hbar)
}{
\sum_{k \in {\mathbb Z}} e^{2\pi i k \rho/\hbar} \cdot Z(t,\nu+k \hbar; \hbar) 
}  
\notag \\
& \quad 
- \frac{
\sum_{k \in {\mathbb Z}} e^{2\pi i k \rho/\hbar}  
\cdot Z(t,\nu+k \hbar-\hbar; \hbar) 
\cdot  i \cdot e^{-2\pi i (\nu+k \hbar)/ \hbar} 
\cdot \psi_{+}^{(0)}(x,t,\nu + k \hbar -\hbar; \hbar)
}{
\sum_{k \in {\mathbb Z}} e^{2\pi i k \rho/\hbar} \cdot Z(t,\nu+k \hbar; \hbar) 
}  \notag \\
& \quad 
+ 
\frac{
\sum_{k \in {\mathbb Z}} e^{2\pi i k \rho/\hbar} 
\cdot Z(t,\nu+k \hbar; \hbar) \cdot i \cdot 
e^{-2\pi i (\nu+k \hbar)/ \hbar} \cdot  \psi_{+}^{(0)}(x,t,\nu+k \hbar; \hbar)
}{
\sum_{k \in {\mathbb Z}} e^{2\pi i k \rho/\hbar} \cdot Z(t,\nu+k \hbar; \hbar) 
}  \notag \\
& = \Psi^{(0)}_{-}(x,t,\nu,\rho;\hbar) + 
i \cdot ( - e^{- 2 \pi i (\nu - \rho)/\hbar} + e^{- 2 \pi i \nu / \hbar} )
\cdot \Psi^{(0)}_{+}(x,t,\nu,\rho;\hbar). 
\end{align}
Thus we have
\begin{equation}
s_{-1} = 
i \cdot ( - e^{- 2 \pi i (\nu - \rho)/\hbar} + e^{- 2 \pi i \nu / \hbar} ).
\end{equation}

\vspace{+.5em}\noindent
$\bullet$ \underline{\bf Computation of other Stokes multipliers.} \, 
The above computations imply that the Voros formulas and 
Proposition \ref{prop:detoured-continuation} are enough to calculate 
the Stokes multipliers. The general formula 
is given as follows:
\begin{equation} \label{eq:general-formula-for-Stokes-multipliers}
s_\ell = i \cdot \sum_{C} \, (-1)^{m_C n_C} \cdot
e^{ {\rm sign}(C) \cdot 2 \pi i (m_C \nu - n_C \rho)/ \hbar}. 
\end{equation}
Here the summation is taken over all Stokes curves which 
have $\arg x = 2 \ell \pi/5$ as the asymptotic direction, 
and ${\rm sign}(C) \in \{ \pm 1 \}$  is the sign for the Stokes curve $C$
which specifies the dominant WKB solution on $C$.
In summary, we have the following table of the Stokes multipliers 
in Figure \ref{fig:Stokes-multipliers}:
\begin{align} \label{eq:Stokes-multipliers-1}
\begin{cases}
s_{-2}  =  i \cdot ( e^{- 2 \pi i \rho/\hbar} - e^{2\pi i (\nu - \rho)/\hbar}), \\
s_{-1}  = i \cdot ( - e^{- 2 \pi i (\nu - \rho)/\hbar} + e^{- 2 \pi i \nu / \hbar} ), \\
s_0 \hspace{+.6em} = i \cdot e^{2 \pi i \nu/\hbar}, \\
s_1 \hspace{+.6em} = i \cdot 
( e^{-2 \pi i \nu/\hbar} - e^{-2 \pi i (\nu+\rho)/\hbar} 
+ e^{-2 \pi i \rho/\hbar}), \\
s_2 \hspace{+.6em} = i \cdot e^{2\pi i \rho/\hbar}.
\end{cases}
\end{align}

Let us give two remarks. First, the Stokes multipliers we have computed 
are independent of $t$. This is consistent with the fact that 
$\Psi_{\pm}$ satisfies the isomonodromy system $(L_{\rm I})$-$(D_{\rm I})$
given in \eqref{eq:JM-1}-\eqref{eq:JM-2}. 
Second, $s_{\ell}$'s  satisfy the cyclic relation: 
\begin{equation} \label{eq:cyclic-relation}
\begin{pmatrix} 1 & 0 \\ s_{2} & 1 \end{pmatrix} \cdot 
\begin{pmatrix} 1 & s_{1} \\ 0 & 1 \end{pmatrix} \cdot
\begin{pmatrix} 1 & 0 \\ s_{0} & 1 \end{pmatrix} \cdot
\begin{pmatrix} 1 & s_{-1} \\ 0 & 1 \end{pmatrix} \cdot
\begin{pmatrix} 1 & 0 \\ s_{-2} & 1 \end{pmatrix}  
= 
\begin{pmatrix} 0 & i \\ i & 0 \end{pmatrix}. 
\end{equation}
In other words, they satisfy the relations 
\begin{equation}
1 + s_{\ell-1} \cdot s_{\ell} + i \cdot s_{\ell+2} = 0, \quad s_{\ell+5} = s_{\ell},
\end{equation}
which are essentially the cluster exchange relations 
in the $A_2$ cluster algebra, 
or the defining equation of the monodromy space, or the wild character variety
(c.f., \cite{FIKN, SvdP}). 
These are supporting evidence 
for the validity of our computation.

\section{Conclusion and open problems}
\label{section:conclusion}

In this article, we have constructed a 2-parameter 
family of (formal) $\tau$-function 
for the first Painlev\'e equation $(P_{\rm I})$
as the discrete Fourier transform of the 
topological recursion partition function 
applied to a family of genus $1$ spectral curves. 
We also obtained a solution of isomonodromy system 
associated with $(P_{\rm I})$ with the aid of 
techniques used in literature of quantum curves.

We summarize several questions and   
open problems in the following list. 

\begin{enumerate}
\renewcommand{\theenumiii}{\Roman{enumiii}}
\item 
We expect that the results obtained in this paper  
for Painlev\'e I can be generalized to 
other Painlev\'e equations including a class of 
higher order analogues.  
The Fourier series structure also appears 
in the higher order Painlev\'e equations 
(see \cite{Gav, GavIL} for example).
In a recent article \cite{MO19} by Marchal-Orantin, 
a systematic approach to rank 2 isomonodromy systems 
via the topological recursion is established. 
These works seem to be good references for 
the direction.

\item 
Aoki-Kawai-Takei also constructed a 
2-parameter family of WKB-theoretic formal solutions 
of Painlev\'e equations (\cite{AKT-P}); see also 
\cite{ASV, GIKM, Yoshida} for $\hbar=1$ case.
More generally, full-parameter formal solutions are also
constructed for a class of higher order 
Painlev\'e equations (Painlev\'e hierarchies) in \cite{AHU} etc. 
Their solution contains infinitely many exponential terms, 
but it is not directly related to $\theta$-functions or elliptic functions.
It is interesting to find a relationship between 
our 2-parameter solution and theirs. 

Kawai-Takei also introduced the notion of 
(non-linear version of)
turning points and Stokes curves for the 
Painlev\'e equations in \cite{KT-PI}, 
and proved that their 2-parameter formal solution 
of Painlev\'e II -- VI can be reduced to that 
of Painlev\'e I near each simple turning point 
by a certain change of variables
(\cite{KT-PI, KT-PIII}; see also \cite{I15}). 
These results should have a counterpart 
in the framework developed in this paper.

\item 
The computational method of the Stokes multipliers 
presented in \S \ref{section:Exact-WKB} contains several 
heuristic arguments: 
Borel summability of the WKB series, 
Voros connection formulas, convergence of the Fourier series, 
are the open problems to be solved. 

Borel summability is proved for some special cases:
Kamimoto-Koike (\cite{Kamimoto-Koike}) proved the Borel summability
of $0$-parameter solutions of classical six Painlev\'e equations 
on compliments of Stokes curves in the sense of Kawai-Takei (\cite{KT-PI}).
See also the work by Costin (\cite{Costin-P}) which includes 
the Borel summability of $1$-parameter (trans-series solution)
of Painlev\'e equations with $\hbar=1$.

Moreover, in articles \cite{ASV, David, FIKN, Kap, Kap-Kit, LR, Takei} etc., 
connection problem (non-linear Stokes phenomenon) of Painlev\'e equations 
are discussed. Our approach to the direct monodromy problem should be 
compared to the previous results.

\item 
The Stokes graphs and trajectories of quadratic differentials
are also used in the study of the space of Bridgeland stability 
conditions (\cite{Bri}) for a class of Calabi-Yau 3 categories 
of quiver representations (\cite{Br-Sm}). 
In \cite{Sutherland1} and \cite{Sutherland2},  
Sutherland studied the stability conditions 
for a class of quivers, called ``Painlev\'e quivers". 
In the case of Painlev\'e I, the relevant quadratic differential 
is equivalent to the one used in \S \ref{section:Exact-WKB} 
(whose associated quiver is of type $A_2$). 
In this picture, the variation of $t$
(i.e., deformation of the spectral curve) 
is related to the variation of stability conditions.
It seems to be interesting to investigate a relationship 
between the analytic continuation of Painlev\'e transcendents 
and the variation of stability conditions (wall-crossings) 
for the corresponding Painlev\'e quivers. 
The work \cite{LR} by Lisovyy-Roussillon pointed that 
the dilogarithm identity appears as a consistency condition
in the connection problem for Painlev\'e I.


\item
In the ($c=1$ Virasoro) conformal field theoretic 
construction of $\tau$-function for Painlev\'e VI 
by Gamayun-Iorgov-Lisovyy (\cite{GIL}),
the expansion coefficients of 2-parameter
$\tau$-function are described by an 
explicit and combinatorial formula, thanks to 
the AGT correspondence \cite{AGT} 
and Nekrasov' formula \cite{Nekrasov}. 
In the works \cite{GL, GL2} of Gavrylenko-Lisovyy, 
the same explicit formulas 
for Painlev\'e VI and III are obtained through 
an explicit calculation of the Fredholm determinant 
for the associated Riemann-Hilbert problem 
(see also \cite{CGL}). 
It would be great if we obtain such 
an explicit formula for Painlev\'e I.

\item
Conjectural expressions of 2-parameter 
$\tau$-functions for Painlev\'e II -- Painlev\'e V 
as the discrete Fourier transform of irregular 
conformal blocks are discussed by the works \cite{Nagoya1, Nagoya2} 
of Nagoya. However, finding such an expression is 
an open problem for Painlev\'e I case.


\item 
Theorem \ref{thm:tau-function-theorem} shows 
in particular that our formal $\tau$-function \eqref{eq:tau-function} 
satisfies the Hirota-type bilinear identity \eqref{eq:Hirota-Painleve-1}. 
In \cite{Ber-Shch, BS}, Bershtein-Shchechkin gave an interesting observation; 
the Hirota-type identities arising from a class of 
($q$-)Painlev\'e equations can be interpreted as the
Nakajima-Yoshioka's blow-up equation (\cite{NY03, NY04}).
It seems to be interesting to give such an interpretation 
(from the instanton counting) for the Hirota-type equation 
\eqref{eq:Hirota-Painleve-1}.


\end{enumerate}


\appendix

\section{Weierstrass functions and $\theta$-functions}
\label{appendix:Weierstrass}

Here we summarize properties of 
Weierstrass elliptic functions and 
$\theta$-functions which are relevant in this paper.
See \cite{WW} for more details.

\subsection{Weierstrass functions}

\subsubsection{Weierstrass $\wp$-function}

Let 
\begin{equation}
\omega_A := \oint_{A} \frac{dx}{\sqrt{4x^3 - g_2 x - g_3}},\quad  
\omega_B := \oint_{B} \frac{dx}{\sqrt{4x^3 - g_2 x - g_3}}
\end{equation}
be the periods of smooth elliptic curve $\Sigma$ defined by $y^2 = 4x^3 - g_2 x - g_3$.
Here ${A}, {B}$ are generators of 
the first homology group $H_1(\Sigma, \bZ)$. 
We assume that $\tau := \omega_B/\omega_A$ has a positive imaginary part.
The coefficients $g_2$, $g_3$ are related to these periods by 
\begin{equation}
g_2 = 60 \sum_{\omega \in \Lambda \setminus\{0\}} \frac{1}{\omega^4}, 
\qquad 
g_3 = 140 \sum_{\omega \in \Lambda \setminus\{0\}} \frac{1}{\omega^6}.
\end{equation}
Here $\Lambda = {\mathbb Z} \cdot \omega_A + {\mathbb Z} \cdot \omega_B$ 
be the lattice generated by the two complex numbers $\omega_A$ and $\omega_B$. 

Under the notations, the Weierstrass elliptic function
$\wp(z) (= \wp(z;g_2,g_3))$ is defined by 
\begin{equation}
\wp(z) := \frac{1}{z^2} + 
\sum_{\omega \in \Lambda \setminus \{0\}} 
\left( \frac{1}{(z-\omega)^2} - \frac{1}{\omega^2} \right).
\end{equation}
It is also constructed as the inverse function of the elliptic integral
\begin{equation} \label{eq:inverse-function-z}
z(x) := \int^{x}_{\infty} \frac{dx}{\sqrt{4x^3-g_2 x -g_3}},
\end{equation}
and hence $\wp(z)$ is doubly-periodic function 
with periods $\omega_A$ and $\omega_B$. 
It has double pole at 
$z = m \, \omega_A + n \, \omega_B$ for any $(m,n) \in \bZ^2$.
We also note that $\wp(z)$ is an even function; 
that is $\wp(-z) = \wp(z)$.

The Weierstarss $\wp$-function satisfies 
the following non-linear ODE:
\begin{equation} \label{eq:Weierstrass-equation}
\left( \frac{d\wp}{dz}(z) \right)^2 = 
4 \wp(z)^3 - g_2 \, \wp(z) - g_3.
\end{equation}
Thus the $\wp$-function is used to parametrize 
elliptic curves. Moreover, differentiating the relation, 
we also have
\begin{equation} \label{eq:Weierstrass-equation-2}
\frac{d^2\wp}{dz^2}(z) = 6 \wp(z)^2 - \frac{g_2}{2}.
\end{equation}

\subsubsection{Weierstrass $\zeta$-function}
We also introduce the Weierstrass $\zeta$-function
\begin{equation} \label{eq:zeta-def}
\zeta(z) := \frac{1}{z} + \sum_{\omega \in \Lambda \setminus \{0\}} 
\left( \frac{1}{z-\omega} + \frac{1}{\omega} + \frac{z}{\omega^2} \right),
\end{equation}
which satisfies $- \zeta'(z) = \wp(z)$. 
The function $\zeta(z)$ is not doubly-periodic, 
but it satisfies  
\begin{equation} \label{eq:zeta-periods}
\zeta(z + m \, \omega_A + n \, \omega_B) = 
\zeta(z) + m \, \eta_A + n  \, \eta_B
\quad (m,n \in {\mathbb Z}).
\end{equation}
The constants $\eta_{A}$ and $\eta_{B}$ are 
also expressed as elliptic integral (of the second kind):
\begin{equation}
\eta_{\ast} := - \oint_{\ast} 
\frac{x dx}{\sqrt{4x^3-g_2 x -g_3}}
= \frac{1}{2} \, \zeta\left( \frac{\omega_\ast}{2} \right)  \quad (\ast \in \{A, B \}).
\end{equation}
The Riemann bilinear identity shows
\begin{equation} \label{eq:R-B-Id}
\eta_A \cdot \omega_B - \eta_B \cdot \omega_A = 2\pi i.
\end{equation}

\subsubsection{Weierstrass $\sigma$-function}
The integral of the Weierstrass zeta function 
is expressed as the logarithm of the Weierstrass $\sigma$-function defined by 
\begin{equation} \label{eq:sigma-def}
\sigma(z) := z \cdot \prod_{\omega \in \Lambda \setminus \{ 0\}} 
\left(1 - \frac{z}{\omega} \right) \cdot
\exp\left( \frac{z}{\omega} + \frac{z^2}{2 \omega^2} \right).
\end{equation}
This satisfies $d \log \sigma(z) /dz = \zeta(z)$. 
$\sigma$-function possesses the following quasi-periodicity:
\begin{equation} \label{eq:quasi-periodicitiy-sigma-general}
\sigma(z+m\,\omega_A + n\, \omega_B) = 
(-1)^{m+n+mn} \cdot \exp\left( (m\,\eta_A + n\, \eta_B) \cdot 
\Bigl(z + \frac{m\,\omega_A + n\, \omega_B}{2} \Bigr) \right)
\cdot \sigma(z) \quad
(m,n \in {\mathbb Z}).
\end{equation}
It is also known that the $\sigma$-function also satisfies the addition formula:
\begin{equation}
\label{eq:addition-sigma}
\frac{\sigma(z+w) \cdot \sigma(z-w)}{\sigma(z)^2 \cdot \sigma(w)^2} = \wp(w) - \wp(z).
\end{equation}

\subsection{$\theta$-functions}
\label{subsection:theta-functions}

The Riemann $\theta$-function is defined by
\begin{equation}
\theta(v,\tau) := \sum_{k \in {\mathbb Z}} e^{2\pi i k v + \pi i k^2 \tau}.
\end{equation}
This is known to convergent uniformly on 
$(v,\tau) \in {\mathbb C} \times {\mathbb H}$. 
We also use the $\theta$-functions with characteristics:
\begin{align}
\theta_{00}(v,\tau) := \theta(v,\tau) 
& = \sum_{k \in {\mathbb Z}} e^{2\pi i k v + \pi k^2 \tau} \\
\theta_{01}(v,\tau) & := 
\sum_{k \in {\mathbb Z}} e^{2\pi i k \bigl( v+\frac{1}{2} \bigr) + \pi k^2 \tau} \\
\theta_{10}(v,\tau) & := 
\sum_{k \in {\mathbb Z}} e^{2\pi i \bigl(k+\frac{1}{2}\bigr) v + 
 \pi \bigl(n+\frac{1}{2}\bigr)^2 \tau} \\
\theta_{11}(v,\tau) & := 
\sum_{k \in {\mathbb Z}} e^{2\pi i \bigl(k+\frac{1}{2}\bigr)
 \bigl(v+\frac{1}{2}\bigr) + \pi \bigl(k+\frac{1}{2}\bigr)^2 \tau}. 
\end{align}
The parity of these functions are
\begin{equation}
\theta_{00}(-v,\tau) = \theta_{00}(v,\tau),\quad
\theta_{01}(-v,\tau) = \theta_{01}(v,\tau),\quad
\theta_{10}(-v,\tau) = \theta_{10}(v,\tau),\quad
\theta_{11}(-v,\tau) = -\theta_{11}(v,\tau).
\end{equation}
The relation to the Weierstrass $\sigma$-function is given as 
\begin{equation}
\sigma(z) = \exp\left( \frac{\eta_A}{2 \omega_A} \cdot z^2 \right) 
\cdot \frac{\omega_A}{\theta_{11}'(0,\tau)} \cdot 
\theta_{11}\left( \frac{z}{\omega_A},\tau \right).
\end{equation}
By taking the logarithm derivative, we have
\begin{equation} \label{eq:wp-and-theta-00}
\wp(z) = 
- \frac{\eta_A}{\omega_A} - \frac{1}{\omega_A^2} \cdot
\left[ \frac{\p^2}{\p v^2} \log  \theta_{11}(v,\tau) 
\right]_{v = \frac{z}{\omega_A}}
= 
- \frac{\eta_A}{\omega_A} - \frac{1}{\omega_A^2} \cdot
\left[\frac{\p^2}{\p v^2} \log  \theta_{00}(v,\tau) 
\right]_{v = \frac{z}{\omega_A} - \frac{1}{2} - \frac{\tau}{2}}.
\end{equation}
The last equality follows from the relation
\begin{equation}
\theta_{00}(v,\tau) = - e^{\pi i v + \frac{i \pi \tau}{4}} \cdot 
\theta_{11}\Bigl( v+\frac{1}{2} + \frac{\tau}{2},\tau \Bigr). 
\end{equation}

\section{Proof of Theorem \ref{thm:BPZ-type-equation}}
\label{appendix:quantum-curve-theorem}

Here we give a proof of Theorem \ref{thm:BPZ-type-equation} 
which plays an important role in the proof of our main result.
We use the same notation used in \S \ref{subsection:def-of-correlators}. 
(For example, the symbol $z_{[\hat{j}]}$ for $j = 1,\dots, n$ 
means the $(n-1)$-tuple of variables 
$(z_1, \dots, \hat{z}_{j}, \dots, z_n)$ without $j$-th entry.) 

\begin{lem} \label{lem:diff-rec}
The function $F_{g,n}(z_1, \dots, z_n)$ defined 
in \eqref{eq:open-Fgn}
satisfies the following equality for $2g-2+n \ge 1$:
\begin{equation}  \label{eq:diff-rec} %
\frac{\p F_{g,n}}{\p z_{1}}(z_{1},\dots,z_{n}) 
= 
G_{g,n}(z_1, \dots, z_n)
- \frac{1}{\omega_A} \oint_{z \in A} 
\frac{R_{g,n}(z,z_2,\dots,z_n)}{(y(z) - y(\bar{z})) \cdot dx(z)},
\end{equation}  
with $G_{g,n}$ and $R_{g,n}$ being given as follows:
\begin{itemize}
\item 
For $2g-2+n = 1$, we set 
\begin{align}
G_{0,3}(z_1,z_2,z_3) & := 
- \sum_{j=2}^{3} \bigl( P(z_1+z_j) - P(z_1-z_j) \bigr) \notag \\
& \qquad \times
\biggl( \frac{1}{2y(z_{1}) \cdot \frac{dx}{dz}(z_{1})} \cdot 
\frac{\p F_{0,2}}{\p z_{1}}(z_{[\hat{j}]}) 
 - \frac{1}{2y(z_{j}) \cdot \frac{dx}{dz}(z_{j})} \cdot
\frac{\p F_{0,2}}{\p z_{j}}(z_{[\hat{1}]}) \biggr) \notag   \\
& \quad + \frac{1}{y(z_1) \cdot \frac{dx}{dz}(z_1)} 
\cdot \frac{\p F_{0,2}}{\p z_1}(z_1, z_2) \cdot \frac{\p F_{0,2}}{\p z_1}(z_1, z_3), 
\\[+.5em]
G_{1,1}(z_1) & := - \frac{1}{2y(z_{1}) \cdot \frac{dx}{dz}(z_{1})} 
\cdot \frac{\p^2}{\p u_1 \p u_2} F_{0,2}(u_1, u_2) 
\biggl|_{u_{1}=u_{2}=z_{1}},
\end{align}
with $P(z)$ being given in \eqref{eq:Pz-function}, 
and for $2g-2+n \ge 2$, we set 
\begin{align}
& G_{g,n}(z_1, \dots, z_n) \notag \\ 
& \quad := - \sum_{j=2}^{n} \bigl( P(z_1+z_j) - P(z_1-z_j) \bigr) \notag \\
& \qquad \times
\biggl( \frac{1}{2y(z_{1}) \cdot \frac{dx}{dz}(z_{1})} \cdot 
\frac{\p F_{g,n-1}}{\p z_{1}}(z_{[\hat{j}]}) 
 - \frac{1}{2y(z_{j}) \cdot \frac{dx}{dz}(z_{j})} \cdot
\frac{\p F_{g,n-1}}{\p z_{j}}(z_{[\hat{1}]}) \biggr)  
\notag \\
& \qquad - \frac{1}{2y(z_{1}) \cdot \frac{dx}{dz}(z_{1})} \cdot 
\frac{\p^{2}}{\p u_{1} \p u_{2}}
\biggl( F_{g-1,n+1}(u_{1},u_{2},z_{[\hat{1}]}) 
\notag \\
& \hspace{+10.em}  + 
\sum_{\substack{g_{1}+g_{2}=g \\ I \sqcup J = [\hat{1}]}}^{\rm stable} 
F_{g_{1}, |I|+1}(u_{1},z_{I}) \cdot F_{g_{2}, |J|+1}(u_{2}, z_{J}) 
\biggr)\Biggr|_{u_{1}=u_{2}=z_{1}}. 
\end{align}

\item
$R_{g,n}(z,z_{2},\dots,z_{n})$ is a quadratic differential 
in the variable $z$ (and functions of other variables $z_2, \dots, z_n$) 
defined by
\begin{align}
R_{0,3}(z, z_2, z_3) & := 
\biggl(\int^{z_{2}}_{0} W_{0,2}(z,z_{2}) \biggr) \cdot 
\biggl(\int^{z_{3}}_{0} W_{0,2}(\bar{z},z_{3}) \biggr) \notag \\
& \qquad + 
\biggl(\int^{z_{3}}_{0} W_{0,2}(z,z_{3}) \biggr) \cdot 
\biggl(\int^{z_{2}}_{0} W_{0,2}(\bar{z},z_{2}) \biggr), 
\\
R_{1,1}(z) & := W_{0,2}(z, \bar{z}), 
\end{align}
for $2g-2+n = 1$, and 
\begin{align}  
& R_{g,n}(z,z_{2},\dots,z_{n})  \notag \\ 
& := \sum_{j=2}^{n} 
\Biggl[ \biggl(\int^{z_{j}}_{0} W_{0,2}(z,z_{j}) \biggr) \cdot
\biggl(\int^{z_{[\hat{1}, \hat{j}]}}_{0}
W_{g,n-1}(\bar{z}, z_{[\hat{1},\hat{j}]}) \biggr) 
\notag \\ 
& \hspace{+7.em} + \biggl(\int^{z_{j}}_{0}W_{0,2}(\bar{z},z_{j}) \biggr)  \cdot
\biggl(\int^{z_{[\hat{1}, \hat{j}]}}_{0}
W_{g,n-1}(z, z_{[\hat{1},\hat{j}]}) 
\biggr) \Biggr] 
\notag \\
& \quad +  \int^{z_{[\hat{1}]}}_{0}
W_{g-1,n+1}(z,\bar{z},z_{[\hat{1}]})  
 +  \sum_{\substack{g_{1}+g_{2}=g \\ 
 I \sqcup J = [\hat{1}]}}^{\text{\rm stable}} 
\biggl( \int^{z_{I}}_{0} 
W_{g_{1}, |I|+1}(z,z_{I}) \biggr) \cdot
\biggl( \int^{z_{J}}_{0} 
W_{g_{2}, |J|+1}(\bar{z}, z_{J})
\biggr) \hspace{-3.em} \notag \\[-1.em]
\end{align}
for $2g-2+n \ge 2$. 
Here, for a set $L =\{\ell_{1}, \dots, \ell_{k} \} 
\subset \{1,\dots,n \}$ of indices, we have used the notation
\begin{equation}
\int^{z_{L}}_{0} W_{g,n}(z_{1},\dots,z_{n}) := 
\int^{z_{\ell_{1}}}_{0} \cdots 
\int^{z_{\ell_{k}}}_{0} W_{g,n}(z_{1},\dots,z_{n}).
\end{equation}
\end{itemize}
\end{lem}

\begin{proof}
First we show the claim in the case $2g-2+n \ge 2$. 
We employ a similar technique used in the proof of \cite[Theorem 3.11]{IS}. 

Integrating the topological recursion relation \eqref{eq:top-rec} 
with respect to $z_{2}, \dots, z_{n}$, we have 
\begin{align} 
\frac{\p}{\p z_{1}} F_{g,n}(z_{1},\dots,z_{n}) \, dz_1 & = 
\int_{0}^{z_{2}} \cdots 
\int_{0}^{z_{n}} W_{g,n}(z_{1},z_{2},\dots,z_{n}) \notag \\
& =  
\frac{1}{2\pi i} \sum_{j=1}^{3}\oint_{\gamma_{j}} 
K(z_1,z) \cdot R_{g,n}(z,z_{1},\dots,z_{n}).
\label{eq:pre-diff-recursion}
\end{align}
Note that, as a differential of $z$, $R_{g,n}$ has a simple pole at 
$z \equiv z_{1}, \bar{z}_{1}, \dots, z_{n}, \bar{z}_n$ modulo $\Lambda$, 
and no other poles except for the ramification points.  
Hence, the residue theorem implies
\begin{align} 
& \frac{\p}{\p z_{1}} F_{g,n}(z_{1},\dots,z_{n}) \, dz_1 \notag \\ 
& = - \sum_{i=1}^n 
\sum_{p=z_{i}, \bar{z}_{i}}
\Res_{z=p} 
K(z_1,z) \cdot R_{g,n}(z,z_{2},\dots,z_{n}) 
+ \frac{1}{2\pi i} \oint_{z \in \partial D}
K(z_1,z) \cdot R_{g,n}(z,z_{2},\dots,z_{n}) \notag \\
& = %
\sum_{i=1}^n 
\sum_{p=z_{i}, \bar{z}_{i}}
\Res_{z=p} 
\frac{P(z_1+z) - P(z_1-z)}{4y(z) \cdot dx(z)} \cdot 
\Biggl( 
\sum_{j=2}^{n} 
\Biggl[ 
\Bigl( P_1(z+z_j) - P_1(z-z_j) \Bigr)
\cdot
\frac{\partial}{\partial u} F_{g,n-1}(u, z_{[\hat{1},\hat{j}]})
\Biggr] \Biggl|_{u=z}
\notag \\
& \quad +
\frac{\partial^2}{\partial u_1 \partial u_2} 
\biggl[
F_{g-1, n+1}(u_1, u_2, z_{[\hat{1}]}) 
 + \sum_{\substack{g_{1}+g_{2}=g \\ 
 I \sqcup J = [\hat{1}]}}^{\text{stable}} 
 F_{g_{1}, |I|+1}(u_1, z_{I}) \cdot
 F_{g_{2}, |J|+1}(u_2, z_{J})
\biggr]\Biggl|_{u_1=u_2=z}
\Biggr) \, dz_1 
\notag \\
& \quad 
+ \frac{1}{2\pi i}\oint_{z \in \partial \Omega} 
K(z_1,z) \cdot R_{g,n}(z,z_{2},\dots,z_{n}).
\label{eq:integrating-TR}
\end{align}
The last term is the integration along 
the boundary of the fundamental domain $\Omega$
of the elliptic curve. 

The first two lines of the right hand-side of \eqref{eq:integrating-TR}
coincides with $G_{g,n}(z_1, \dots, z_n)$ 
in the desired equality \eqref{eq:diff-rec}
(c.f., \cite[Theorem 4.7]{DM14}).
On the other hand, since
\begin{equation}
K(z_1, z+\omega_\ast) - K(z_1,z) = 
\begin{cases}
0 & \text{for $\ast = A$} \\
\displaystyle 
\frac{1}{(y(z) - y(\bar{z})) \cdot dx(z)} 
\cdot \frac{2 \pi i}{\omega_A} \cdot dz_1
& \text{for $\ast = B$},
\end{cases}
\end{equation}
the integration along $\partial \Omega$ is computed as follows:
\begin{align}
& \frac{1}{2\pi i}\int_{z \in \partial D} 
K(z_1,z) \cdot R_{g,n}(z,z_{2},\dots,z_{n}) \\
& \quad =  
\frac{1}{2 \pi i} 
\int_{z \in A} (K(z_1,z) - K(z_1, z+\omega_B)) \cdot
R_{g,n}(z,z_{2},\dots,z_{n}) \\
& \qquad -
\frac{1}{2 \pi i} 
\int_{z \in B} (K(z_1,z) - K(z_1, z+\omega_A)) \cdot
R_{g,n}(z,z_{2},\dots,z_{n})
\\
& \quad = 
-\frac{dz_1}{\omega_A} \oint_{z \in A}
\frac{R_{g,n}(z,z_{2},\dots,z_{n})}
{(y(z) - y(\bar{z})) \cdot dx(z)}.
\end{align}
Thus we have proved \eqref{eq:diff-rec} for $2g-2+n \ge 2$.

The exceptional two cases $(g,n)=(0,3)$ and $(1,1)$ can be checked similarly
by using the identity 
\begin{equation}
\frac{\p}{\p z_1} F_{0,2}(z_1, z_2) \cdot dz_1 = \bigl( P(z_1+z_2) - P(z_1) \bigr) \cdot dz_1 
= - \int^{z_2}_{0} W_{0,2}(\bar{z}_1, z_2)
\end{equation}
which immediately follows from the definition \eqref{eq:def-F02} of $F_{0,2}$.
\end{proof}

The following formula will be used to 
relate the $A$-cycle integral in the
right hand-side of \eqref{eq:diff-rec} 
with the $t$-derivatives.

\begin{lem} \label{lem:diff-rec-and-t-dervative}
For $2g-2+n \ge 1$, we have
\begin{equation}
\frac{\p}{\p t} F_{g,n-1}(z(x_{1}),\dots,z(x_{n-1})) = 
E_{g,n-1}(z(x_1), \dots, z(x_{n-1})), 
\end{equation}
where
\begin{align}
& 
E_{g,n-1}(z_1, \dots, z_{n-1}) \notag \\
& \quad := 
- \frac{1}{\omega_A} \oint_{z \in A} 
\frac{R_{g,n}(z,z_1,\dots,z_{n-1})}{(y(z) - y(\bar{z})) \cdot dx(z)} 
+ \sum_{j=1}^{n-1}   
\frac{P(z_j)}{y(z_{j}) \cdot \frac{dx}{dz}(z_{j})} \cdot
\frac{\p F_{g,n-1}}{\p z_{j}}(z_1,\dots,z_{n-1}).
\end{align}
\end{lem}

\begin{proof}
Replacing the label $z_1 \leftrightarrow z_n$ in \eqref{eq:diff-rec} 
and taking the following residue around $z_n = 0$, we have 
\begin{multline}
\Res_{z_{n} = 0} 
\left( 
z_n^{-1} \cdot \frac{\p F_{g,n}}{\p z_{n}}(z_{1},\dots,z_{n}) \cdot dz_n
\right)
\\
= - \frac{1}{\omega_A} \oint_{z \in A} 
\frac{R_{g,n}(z,z_1,\dots,z_{n-1})}{(y(z) - y(\bar{z}))dx(z)} 
+ \sum_{j=1}^{n-1}   
\frac{P(z_j)}{y(z_{j}) \cdot \frac{dx}{dz}(z_{j})} \cdot 
\frac{\p F_{g,n-1}}{\p z_{j}}(z_1,\dots,z_{n-1}). 
\end{multline}
Proposition \ref{prop:variation} shows that 
the left hand-side is nothing but the $t$-derivation 
of $F_{g,n-1}$, if we use the $x$-coordinates. 
This completes the proof.  
\end{proof}

Let us set 
$\tilde{G}_{g,n}(z_1,\dots,z_n) := 
\p_{z_1} F_{g,n}(z_1,\dots, z_n) - G_{g,n}(z_1,\dots,z_n)$.
By a similar computation in \cite[Theorem 6.5]{DM14}, we have
\begin{align}
& \sum_{\substack{g \ge 0,~ n \ge 1 \\ 2g-2+n = m}}
\left[ 
\frac{2\tilde{G}_{g,n}(z_1,\dots,z_n)}{{(n-1)!}}
 \right]_{z_1 = \cdots = z_n = z(x)} \notag \\[+.5em]
&\quad = 
\begin{cases} \displaystyle
 2 y\bigl( z(x) \bigr) \cdot \frac{\p S_m(x)}{\p x}
 + \left( - \frac{\p_x y\bigl(z(x)\bigr)}{y\bigl(z(x)\bigr)} \right)
 \cdot \frac{\p S_{m-1}(x)}{\p x}  &  \\[+1.5em]
 \hspace{+5.em} 
 \displaystyle + 
 \sum_{\substack{m_1, m_2 \ge 1 \\ m_1 + m_2 = m-1}} 
 \frac{\p S_{m_1}(x)}{\p x} \cdot \frac{\p S_{m_2}(x)}{\p x} 
 +  \frac{\p^2 S_{m-1}(x)}{\p x^2} & \quad \text{for $m \ge 2$}
 \\[+2.5em]
 \displaystyle
 2 y\bigl( z(x) \bigr) \cdot \frac{\p S_1(x)}{\p x}
 + \left( - \frac{\p_x y\bigl(z(x)\bigr)}{y\bigl(z(x)\bigr)} \right)
 \cdot \frac{\p S_{0}(x)}{\p x} 
 - \left( \frac{\p S_{0}(x)}{\p x} \right)^2 
 + \frac{\p^2 S_{0}(x)}{\p x^2} & \quad \text{for $m = 1$}.
 \end{cases}
 \label{eq:pre-Riccati-relation}
\end{align}
(Recall that $S_m(x)$'s are defined in \S \ref{subsec:WKB-BPZ}.)
Here we have used the identity 
$y\bigl( z(x) \bigr) = \frac{dx}{dz}\bigl( z(x) \bigr)$.

Next, using Lemma \ref{lem:diff-rec} and Lemma \ref{lem:diff-rec-and-t-dervative}, 
let us find another expression of the left hand-side of \eqref{eq:pre-Riccati-relation}.
First, we note 
\begin{align}
& 
\sum_{\substack{g \ge 0,~ n \ge 1 \\ 2g-2+n = m}}
\left[\frac{2\tilde{G}_{g,n}(z_1,\dots,z_n)}{{(n-1)!}}
\right]_{z_1 = \cdots = z_n = z(x)} \notag 
\\ 
& = 
\sum_{\substack{g \ge 0,~ n \ge 1 \\ 2g-2+n = m}}
\left[\frac{2\tilde{G}_{g,n}(z_1,\dots,z_n)}{{(n-1)!}}
\right]_{z_1 = \cdots = z_n = z(x)} 
- 
\sum_{\substack{g \ge 0,~ n \ge 2 \\ 2g-2+n = m}}
\left[\frac{2\tilde{G}_{g,n}(z_1,\dots,z_n)}{{(n-1)!}}
\right]_{z_1 = \cdots = z_n = z(x)} \notag  \\
& \qquad  + 
\sum_{\substack{g \ge 0,~ n \ge 2 \\ 2g-2+n = m}}
\left[ 
\frac{2\tilde{G}_{g,n}(z_1,\dots,z_n)}{{(n-1)!}}
- \frac{2 E_{g,{n-1}}(z_1,\dots,z_{n-1})}{{(n-1)!}}
\right]_{z_1 = \cdots = z_n = z(x)} \notag \\
& \qquad 
+ \sum_{\substack{g \ge 0,~ n \ge 2 \\ 2g-2+n = m}}
\left[\frac{2 E_{g,{n-1}}(z_1,\dots,z_{n-1})}{{(n-1)!}}
\right]_{z_1 = \cdots = z_n = z(x)}. 
\label{eq:sub-Riccati-0}
\end{align}
The first line of the right hand-side is 
\begin{align}
& \sum_{\substack{g \ge 0,~ n \ge 1 \\ 2g-2+n = m}}
\left[\frac{2\tilde{G}_{g,n}(z_1,\dots,z_n)}{{(n-1)!}}
\right]_{z_1 = \cdots = z_n = z(x)} 
- 
\sum_{\substack{g \ge 0,~ n \ge 2 \\ 2g-2+n = m}}
\left[\frac{2\tilde{G}_{g,n}(z_1,\dots,z_n)}{{(n-1)!}}
\right]_{z_1 = \cdots = z_n = z(x)} \notag  \\
& \qquad = 
\begin{cases}
0 & \text{if $m$ is even} \\[+.3em]
\displaystyle 
2 \frac{\p}{\p t} F_{\frac{m+1}{2}} & \text{if $m$ is odd}
\end{cases} 
\qquad = \quad
\left[ 2 \frac{\p F}{\p t} (t,\nu; \hbar) \right]_{\hbar^{m-1}}.
\label{eq:sub-Riccati-1}
\end{align}
(C.f., \cite[Lemma 4.5]{IS}.) 
The notation $[\bullet]_{\hbar^{k}}$ means 
the coefficient of $\hbar^k$ in a formal power series 
$\bullet$ of $\hbar$. 
The second and third lines are also expressed as
\begin{align}
& \sum_{\substack{g \ge 0,~ n \ge 2 \\ 2g-2+n = m}}
\left[ 
\frac{2\tilde{G}_{g,n}(z_1,\dots,z_n)}{{(n-1)!}}
- \frac{2 E_{g,{n-1}}(z_1,\dots,z_{n-1})}{{(n-1)!}}
\right]_{z_1 = \cdots = z_n = z(x)} \notag \\
& \quad = 
- \sum_{\substack{g \ge 0,~ n \ge 2 \\ 2g-2+n = m}} 
\left[\sum_{j=1}^{n-1} 
\frac{2P(z_j)}{y(z_{j}) \cdot \frac{dx}{dz}(z_{j})} \cdot 
\frac{\p F_{g,n-1}}{\p z_{j}}(z_1,\dots,z_{n-1}) 
\right]_{z_1 = \cdots = z_n = z(x)}
\notag \\
& \quad = 
- \frac{2P\bigl(z(x) \bigr)}{y\bigl( z(x) \bigr)} 
\cdot \frac{\p S_{m-1}(x)}{\p x}
\label{eq:sub-Riccati-2}
\end{align} 
and 
\begin{align} 
& \sum_{\substack{g \ge 0,~ n \ge 2 \\ 2g-2+n = m}}
\left[\frac{2 E_{g,{n-1}}(z_1,\dots,z_{n-1})}{{(n-1)!}}
\right]_{z_1 = \cdots = z_n = z(x)} 
\notag \\
& \quad  = 2 \sum_{\substack{g \ge 0,~ n \ge 2 \\ 2g-2+n = m}}
\frac{\p}{\p t} F_{g,n-1}(z(x),\dots, z(x))  
~=~ 
2 \frac{\p S_{m-1}(x)}{\p t},
\label{eq:sub-Riccati-3}
\end{align}
respectively (c.f., Lemma \ref{lem:diff-rec-and-t-dervative}).
Combining \eqref{eq:pre-Riccati-relation}--\eqref{eq:sub-Riccati-3}, 
we have the following recursion relation satisfied by $S_m$'s: 
\begin{equation}
\label{eq:the-Riccati-recursion-Sm}
\sum_{\substack{m_1, m_2 \ge -1 \\ m_1 + m_2 = m-1}} 
 \frac{\p S_{m_1}(x)}{\p x} \cdot \frac{\p S_{m_2}(x)}{\p x} 
 +  \frac{\p^2 S_{m-1}(x)}{\p x^2} 
 = 2  \frac{\p S_{m-1}(x)}{\p t} + 
 \left[ 2 \frac{\p F}{\p t} (t,\nu; \hbar) \right]_{\hbar^{m-1}}.
\end{equation}
Here we used \eqref{eq:def-of-S-1} and \eqref{eq:derivative-of-S0} 
to obtain the above expression. 
Although \eqref{eq:the-Riccati-recursion-Sm} 
is valid for $m \ge 1$ a-priori (because it is derived from 
\eqref{eq:pre-Riccati-relation} etc.), 
we can verify that it is also valid for $m=0$ 
thanks to the property \eqref{eq:t-derivative-S-1}. 
Together with the equation 
\begin{equation}
\left( \frac{\p S_{-1}(x)}{\p x} \right)^2 
= y\bigl( z(x) \bigr)^2 = 4x^3 + 2t x + u(t,\nu)
\end{equation}
for the leading term, the recursion relations are 
summarized into a single PDE 
\begin{equation} \label{eq:PDE-Riccati-quantum-curve}
\hbar^{2}\left( \left(\frac{\p S}{\p x}\right)^{2} + 
\frac{\p^{2} S}{\p x^{2}} \right) 
 =  
2\hbar^{2} \frac{\p S}{\p t} + 
\left( 4x^{3}+2t x +  2\hbar^2 \frac{\p F}{\p t}(t,\nu;\hbar) \right)
\end{equation}
for $S$ given in \eqref{eq:WKB-series-B}.
The last equation is equivalent to the PDE \eqref{eq:BPZ-equation}, 
and hence, we have proved Theorem \ref{thm:BPZ-type-equation}.





\begin{thebibliography}{99}



\bibitem{AGT}
L. Alday, D. Gaiotto and Y. Tachikawa, 
Liouville correlation functions from four-dimensional gauge theories, 
{\it Lett. Math. Phys.}, {\bf 91} (2010), 167--197; 
arXiv:0906.3219 [hep-th].

\bibitem{ASV}
I. Aniceto, R. Schiappa and M. Vonk, 
The Resurgence of Instantons in String Theory, 
{\it Comm. Number Theor. Phys.}, 
{\bf 6} (2012),  339--496; arXiv:1106.5922 [hep-th].




\bibitem{AHU}
T. Aoki, N. Honda and Y. Umeta, 
On a construction of general formal solutions for equations 
of the first Painlev\'e hierarchy I. 
{\it Adv. Math}., {\bf 235} (2013), 496--524.

\bibitem{AT}
T. Aoki and M. Tanda, 
Borel sums of Voros coefficients of hypergeometric 
differential equations with a large parameter, 
RIMS K\^oky\^uroku, {\bf 1861} (2013), 17--24. 


\bibitem{AKT-P}
T. Aoki, T. Kawai and Y. Takei, 
WKB analysis of Painlev\'e transcendents with a large parameter II, 
in {\it Structure of Solutions of Differential Equations}, 
World Scientific, 1996, pp.1--49.


\bibitem{Ber-Shch}
M. Bershtein and A. Shchechkin, 
Bilinear equations on Painleve tau functions from CFT,
{\it Commun. Math. Phys.}, {\bf 339} (2015), 1021--1061;
arXiv:1406.3008 [math-ph].

\bibitem{BS}
M. Bershtein and A. Shchechkin, 
Painlev\'e equations from Nakajima-Yoshioka blow-up relations, 
preprint; arXiv:1811.04050 [math-ph].


\bibitem{BGT}
G. Bonelli, A. Grassi and A. Tanzini, 
Quantum curves and $q$-deformed Painlev\'e equations, preprint; 
arXiv:1710.11603 [hep-th].

\bibitem{BLMST}
G. Bonelli, O. Lisovyy, K. Maruyoshi, A. Sciarappa and A. Tanzini, 
On Painlev\'e/gauge theory correspondence, 
{\it  Lett. Math. Phys.}, {\bf 107} (2017), 2359--2413;
arXiv:1612.06235 [hep-th]. 
		
\bibitem{Borot-Eynard}
G. Borot and B. Eynard, 
Geometry of Spectral Curves and All Order Dispersive Integrable System,
{\it SIGMA}, {\bf 8} (2012), 53 pages;
arXiv:1110.4936 [math-ph].


\bibitem{BCD}
V, Bouchard, N. K. Chidambaram and T. Dauphinee, 
Quantizing Weierstrass, 
{\it Commun. Num. Theor. Phys.}, {\bf 12} (2018), 253--303;
arXiv:1610.00225 [math-ph].

\bibitem{BE16}
V. Bouchard and B. Eynard, 
Reconstructing WKB from topological recursion, 
{\it Journal de l'Ecole polytechnique -- Mathematiques}, 
{\bf 4} (2017), 845--908; 
arXiv:1606.04498 [math-ph]




\bibitem{Boutroux}
P. Boutroux, 
Recherches sur les transcendentes de M. Painlev\'e et l'\'etude 
asymptotique des \'equations diff\'erentielles du seconde ordre. 
{\it Ann. \'Ecole Norm. Sup\'er}. {\bf 30} (1913), 255--375.

\bibitem{Bri}
T. Bridgeland, 
Stability conditions on triangulated categories, 
{\it Ann. of Math.}, {\bf 166} (2007), 317--345; 
math.AG/0212237.


\bibitem{Br-Sm}
T. Bridgeland and I. Smith,
Quadratic differentials as stability conditions,
{\it I. Publ. math. IHES}, {\bf 121} (2015), 155--278;
arXiv:1302.7030 [math.AG]. 

\bibitem{CGL}
M. Cafasso, P. Gavrylenko and O. Lisovyy, 
Tau functions as Widom constants,
{\it Commun. Math. Phys.}, {\bf 365} (2019), 741--772;
arXiv:1712.08546 [math-ph].
 
\bibitem{CEO}
L. Chekhov, B. Eynard and N. Orantin, 
Free energy topological expansion for the 2-matrix model. 
{\it JHEP}, {\bf 12} (2006), 053; 
arXiv:math-ph/0603003.
 
\bibitem{Clarkson}
P. A. Clarkson, 
Painlev\'e transcendents, 
in {\it Digital Library of Special Functions}, 
Chapter 32,
\verb+https://dlmf.nist.gov/32+ 
 
\bibitem{Clarkson2}
P. A. Clarkson, 
Open Problems for Painlev\'e Equations, 
{\it SIGMA}, {\bf 15} (2019), 20 pages;
arXiv:1901.10122 [math.CA] 
 
\bibitem{CPT}
I. Coman, E. Pomoni and J. Teschner, 
From quantum curves to topological string partition functions,
preprint; arXiv:1811.01978 [hep-th].

\bibitem{Costin-P}
O. Costin, 
On Borel summation and Stokes phenomena of 
nonlinear differential systems, 
{\it Duke Math. J.}, {\bf 93} (1998), 289-344; 
arXiv:math/0608408 [math.CA].


\bibitem{Costin08}
O. Costin, 
Asymptotics and Borel Summability, 
Monographs and surveys in pure and applied mathematics, 
vol. {\bf 141}, Chapmann and Hall/CRC, 2008.

\bibitem{David}
F. David, 
Non-perturbative effects in matrix models and 
vacua of two dimensional gravity, 
{\it Phys. Lett. B.}, 
{\bf 302} (1993), 403--410; 
arXiv:hep-th/9212106.

\bibitem{DDP93}
E. Delabaere, H. Dillinger and F. Pham,
R\'esurgence de Voros et p\'eriodes des courves hyperelliptique,
{\it Annales de l'Institut Fourier}, {\bf 43} (1993), 163--199.

\bibitem{DP99}
E. Delabaere and F. Pham,
Resurgent methods in semi-classical asymptotics,
{\it Annales de l'I.H.P. Physique th\'eorique}, 
{\bf 71} (1999), 1--94.






\bibitem{DM14}
O. Dumitrescu and M. Mulase, 
Quantum curves for Hitchin fibrations 
and the Eynard-Orantin theory, 
{\it Lett. Math. Phys.} {\bf 104} (2014), 635--671;
arXiv:1310.6022 [math.AG].



\bibitem{DLS}
T. M. Dunster, D. A. Lutz and R. Sch\"afke, 
Convergent Liouville-Green expansions for second-order 
linear differential equations, with an application to Bessel functions, 
{\it Proc. Roy. Soc. London}, {\bf A 440} (1993), 37--54.



\bibitem{EM08}
B. Eynard and M. Mari{\~n}o, 
A holomorphic and background independent partition function 
for matrix models and topological strings, 
{\it J. Geom. Phys}. {\bf 61} (2011), 1181--1202; 
arXiv:0810.4273 [hep-th].

\bibitem{EO07}
B. Eynard and N. Orantin, 
Invariants of algebraic curves and topological expansion,
{\it Comm. Number Theory Phys}. {\bf 1} (2007), 347--452; 
arXiv:math-ph/0702045.

\bibitem{EO09}
B. Eynard and N. Orantin,  
Topological recursion in enumerative geometry and random matrices, 
{\it J. Phys. A: Math. Theor.}, {\bf 42} (2009), 293001 (117pp).


\bibitem{FIKN}
A. S. Fokas, A. R. Its, A. A Kapaev and V. Y. Novokshenov, 
Painlev\'e Transcendents: The Riemann-Hilbert Approach, 
Mathematical Surveys and Monographs, {\bf 128}, 
AMS, Providence, RI, 2006.

\bibitem{FS}
K. Fuji and T. Suzuki, 
Drinfeld-Sokolov hierarchies of type A and 
fourth order Painlev\'e systems, 
{\it Funkcial. Ekvac.} {\bf 53} (2010), 143--167; 
arXiv:0904.3434 [math-ph].



\bibitem{GMN12}
D. Gaiotto, G. W. Moore and A. Neitzke,
Spectral networks,  {\it Ann. Henri Poincar\'e}, 
{\bf 14} (2012), 1643--1731; 
arXiv:1204.4824.

\bibitem{GIL}
O. Gamayun, N. Iorgov and O. Lisovyy,
Conformal field theory of Painlev\'e VI,
{\it JHEP}, {\bf 10} (2012), 038;
arXiv:1207.0787 [hep-th].


\bibitem{GIL13}
O. Gamayun, N. Iorgov and O. Lisovyy, 
How instanton combinatorics solves Painlev\'e VI, V and III's, 
{\it J. Phys. A: Math. Theor.}, {\bf 46} (2013), 335203;
arXiv:1302.1832 [hep-th].

\bibitem{Gav}
P. Gavrylenko, 
Isomonodromic $\tau$-functions and $W_N$ conformal blocks, 
{\it JHEP}, {\bf 2015} (2015), 167; 
arXiv:1505.00259 [hep-th].


\bibitem{GavIL}
P. Gavrylenko, N. Iorgov and O. Lisovyy, 
On solutions of the Fuji-Suzuki-Tsuda system,
{\it SIGMA}, {\bf 14} (2018), 123, 27 pages;
arXiv:1806.08650 [math-ph].

\bibitem{GL}
P. Gavrylenko and O. Lisovyy, 
Fredholm determinant and Nekrasov sum representations 
of isomonodromic tau functions, 
{\it Commun. Math. Phys.}, {\bf 363} (2018), 1--58;
arXiv:1608.00958 [math-ph].

\bibitem{GL2}
P. Gavrylenko and O. Lisovyy, 
Pure $SU(2)$ gauge theory partition function and generalized Bessel kernel, 
preprint; arXiv:1705.01869 [math-ph].


\bibitem{GIKM}
S. Garoufalidis, A. Its, A. Kapaev and M. Marino, 
Asymptotics of the instantons of Painleve I, 
{\it Int. Math. Res. Not.}, 
{\bf 2012} (2012), 561--606; arXiv:1002.3634 [math.CA].


\bibitem{GJP}
P. R. Gordoa, N. Joshi and A. Pickering, 
On a Generalized $2 + 1$ Dispersive Water Wave Hierarchy, 
{\it Publ. RIMS}, {\bf 37} (2001), 327--347.


\bibitem{GG}
A. Grassi and J. Gu, 
Argyres-Douglas theories, 
Painlev\'e II and quantum mechanics, preprint; 
arXiv:1803.02320 [hep-th].





\bibitem{HN}
L. Hollands and A. Neitzke, 
Spectral networks and Fenchel-Nielsen coordinates,
{\it Lett. Math. Phys.}, {\bf 106} (2016), 811--877;
arXiv:1312.2979 [math.GT].

\bibitem{Hone-Zullo}
A. N. W. Hone and F. Zullo, 
Hirota bilinear equations for Painlev\'e transcendents, 
preprint; arXiv:1706.02341 [math.CA].

\bibitem{ILT}
N. Iorgov, O. Lisovyy and J. Teschner, 
Isomonodromic tau-functions from Liouville conformal blocks, 
{\it Comm. Math. Phys.}, {\bf 336}, (2015), 671-694; 
arXiv:1401.6104 [hep-th].

\bibitem{I13}
K. Iwaki, 
Parametric Stokes Phenomenon for the Second Painlev\'e Equation, 
{\it Funkcial. Ekvac.}, {\bf 57}, 173--243.

\bibitem{I15}
K. Iwaki, 
On WKB Theoretic Transformations for Painlev\'e Transcendents 
on Degenerate Stokes Segments, 
{\it Publ. Res. Inst. Math. Sci.}, {\bf 51} (2015), 1--57;
arXiv:1312.1874 [math.CA].

\bibitem{IKoT1}
K. Iwaki, T. Koike and Y.-M. Takei,  
Voros Coefficients for the Hypergeometric Differential
Equations and Eynard-Orantin's Topological Recursion, 
-- Part I : For the Weber Equation --, 
preprint; arXiv:1805.10945.

\bibitem{IKoT2}
K. Iwaki, T. Koike and Y.-M. Takei,  
Voros Coefficients for the Hypergeometric
Differential Equations and Eynard-Orantin's Topological Recursion, 
-- Part II : For the Confluent Family of Hypergeometric Equations --, 
preprint; arXiv:1810.02946.

\bibitem{IM}
K. Iwaki and O. Marchal, 
Painlev\'e 2 equation with arbitrary monodromy parameter, 
topological recursion and determinantal formulas, 
{\it Ann. Henri Poincar\'e}, {\bf 18} (2017), 2581--2620;
arXiv:1411.0875. 


\bibitem{IMS}
K. Iwaki,  O. Marchal and A. Saenz, 
Painlev\'e equations, topological type property
and reconstruction by the topological recursion, 
{\it J. Geom. Phys.}, {\bf 124} (2018), 16--54; 
arXiv:1601.02517 [math-ph].

\bibitem{IN14} 
K. Iwaki and T. Nakanishi,
Exact WKB analysis and cluster algebras,
{\it J. Phys. A: Math. Theor}., {\bf 47} (2014), 474009;
arXiv:1401.7094 [math.CA].

\bibitem{IN15} 
K. Iwaki and T. Nakanishi,
Exact WKB analysis and cluster algebras II: 
Simple poles, orbifold points, and generalized cluster algebras, 
{\it Int. Math. Res. Not}., {\bf 2016} (2016), 4375--4417;
arXiv:1409.4641 [math.CA].

\bibitem{IS}
K. Iwaki and A. Saenz, 
Quantum curve and the first Painlev\'e equation, 
{\it SIGMA}, {\bf 12} (2016), 24 pages; 
arXiv:1507.06557. 

\bibitem{JM2}
M. Jimbo and T. Miwa, 
Monodromy perserving deformation of linear ordinary 
differential equations with rational coefficients. II, 
{\it Physica D}, {\bf 2} (1981), 407--448.

\bibitem{JMU}
M. Jimbo, T. Miwa and K. Ueno, 
Monodromy preserving deformations of linear ordinary differential 
equations with rational coefficients I, 
General theory and tau function, 
{\it Physica D}, {\bf 2} (1981), 306--352.


\bibitem{JNS}
M. Jimbo,  H. Nagoya and  H. Sakai, 
CFT approach to the $q$-Painlev\'e VI equation, 
{\it Journal of Integrable Systems}, {\bf 2} (2017), xyx009;
arXiv:1706.01940 [math-ph].


\bibitem{Kamimoto-Koike}
S. Kamimoto and T. Koike, 
On the Borel summability of $0$-parameter solutions of 
nonlinear ordinary differential equations, 
preprint of RIMS-1747 (2012).



\bibitem{Kap}
A. A. Kapaev, 
Asymptotics of solutions of the Painlev\'e equation of 
the first kind, 
{\it Diff. Eqns.}, {\bf 24} (1989), 1107--1115 
(translated from: Diff. Uravnenija {\bf 24} (1988), 
1684--1695 (Russian)).

\bibitem{Kap-Kit}
A. A. Kapaev and A. V. Kitaev, 
Connection formulae for the first Painlev\'e 
Transcendent in the complex domain, 
{\it Lett. Math. Phys.}, {\bf 27} (1993), 243--252.

\bibitem{KT-PI}
T. Kawai and Y. Takei, 
WKB analysis of Painlev\'e transcendents with a large parameter. I,
{\it Adv. Math.},  {\bf 118} (1996), 1--33.

\bibitem{KT-PIII}
T. Kawai and Y. Takei, 
WKB analysis of Painlev\'e transcendents with a large parameter. III, 
Local equivalence of 2-parameter Painleve transcendents,
{\it Adv. Math.}, {\bf 134} (1998), 178--218.

\bibitem{KT05}
\begin{CJK}{UTF8}{min}
河合隆裕, 竹井義次, 特異摂動の代数解析学, 岩波書店, 1998.
\end{CJK}

(English version: 
T. Kawai and Y. Takei, 
Algebraic Analysis of Singular Perturbation Theory, 
Translations of Mathematical Monographs, vol {\bf 227}, 
AMS, 2005.)

\bibitem{KNS}
H. Kawakami, A. Nakamura and H. Sakai, 
Degeneration scheme of 4-dimensional Painlev\'e-type equations,
preprint, arXiv:1209.3836 [math.CA].

\bibitem{Koike-Schafke} 
 T. Koike and R. Sch\"afke, 
 On the Borel summability of WKB solutions of 
 Schr\"odinger equations with rational potentials 
 and its application, in preparation; also Talk given 
 by T. Koike in the RIMS workshop 
 ``Exact WKB analysis --- Borel summability of WKB solutions''
 September, 2010.

\bibitem{Koike-Takei}
 T. Koike and Y. Takei, 
 On the Voros coefficient for the Whittaker equation with a large parameter
 -- Some progress around Sato's conjecture in exact WKB analysis,
 {\it Publ. Res. Inst. Math. Sci.}, 
 Kyoto Univ. {\bf 47} (2011), pp. 375--395.


\bibitem{LNR}
O. Lisovyy, H. Nagoya and J. Roussillon, 
Irregular conformal blocks and connection formulae for Painlev\'e V functions, 
{\it J. Math. Phys.}, {\bf 59} (2018), 091409;
arXiv:1806.08344 [math-ph]


\bibitem{LR}
O. Lisovyy and J. Roussillon, 
On the connection problem for Painlev\'e I, 
{\it J. Phys. A: Math. Theor.},  {\bf 50} (2017), 255202; 
arXiv:1612.08382 [nlin.SI].

\bibitem{MO19}
O. Marchal and N. Orantin, 
Isomonodromic deformations of a rational differential system 
and reconstruction with the topological recursion:
the ${sl}_2$ case, preprint; 
arXiv:1901.04344 [math-ph].

\bibitem{MN}
Y. Matsuhira and H. Nagoya, 
Combinatorial expressions for the tau functions of 
$q$-Painlev\'e V and III equations, preprint; 
arXiv:1811.03285 [math-ph].



\bibitem{Nagoya1}
H. Nagoya, 
Irregular conformal blocks, with an application 
to the fifth and fourth Painlev\'e equations, 
{\it J. Math. Phys.} {\bf 56} (2015), 123505;
arXiv:1505.02398. 

\bibitem{Nagoya2}
H. Nagoya,  
Remarks on irregular conformal blocks and Painlev\'e 
III and II tau functions, preprint; arXiv:1804.04782. 


\bibitem{NY03}
H. Nakajima and K. Yoshioka, 
Instanton counting on blowup. I. 
4-dimensional pure gauge theory, 
{\it Inv. math.}, {\bf 162} (2005), 313--355; 
arXiv:math/0306108.

\bibitem{NY04}
H. Nakajima and K. Yoshioka, 
Lectures on Instanton Counting, 
in {\it Algebraic Structures and Moduli Spaces}, 
CRM Proc. Lect. Notes {\bf 38}, AMS, 2004, 31--101;
arXiv:math/0311058.


\bibitem{Nekrasov}
N. A. Nekrasov, 
Seiberg-Witten Prepotential From Instanton Counting, 
{\it Adv. Theor. Math. Phys.}, {\bf 7} (2004),  831--864; 
arXiv:hep-th/0206161.

\bibitem{Nor}
P. Norbury, 
Quantum curves and topological recursion, in
{\it Proceedings of Symposia in Pure Mathematics, String-Math 2014}, 
{\bf 93} (2016), 41--65; 
arXiv:1502.04394 [math-ph].


\bibitem{Noumi-Yanada}
M. Noumi and Y. Yamada, 
Higher order Painlev\'e equations of type $A_\ell^{(1)}$, 
{\it Funkcial.Ekvac.}, {\bf 41} (1998), 483--503; 
arXiv:math/9808003 [math.QA].


\bibitem{Okamoto}
K. Okamoto, 
Polynomial Hamiltonians associated with Painlev\'e equations I,
{\it Proc. Japan Acad. Ser. A Math. Sci}.,
{\bf 56} (1980), 264--268.

\bibitem{Okamoto2}
K. Okamoto, 
Polynomial Hamiltonians associated with Painlev\'e equations. II. 
Differential equations satisfied by polynomial Hamiltonians, 
{\it Proc. Japan Acad. Ser. A Math. Sci.}, 
{\bf 56} (1980), 367--371.


\bibitem{Painleve}
P. Painlev\'e, 
Sur les \'equations diff\'erentielles du second ordre et 
d'ordre sup\'erieur dont l'int\'egrale g\'en\'erale est uniforme, 
{\it Acta Math}., {\bf 25} (1902), 1--85.

\bibitem{SvdP}
M. van der Put and M.-H. Saito, 
Moduli spaces for linear differential equations and the Painlev\'e equations, 
{\it Ann. Inst. Fourier} (Grenoble), {\bf 59} (2009), 2611--2667;
arXiv:0902.1702 [math.AG].

\bibitem{RGH}
A. Ramani, B. Grammaticos and J. Hietarinta, 
Discrete versions of the Painlev\'e equations, 
{\it Phys. Rev. Lett.}, {\bf 67} (1991), 1829--1832.


\bibitem{Sakai}
H. Sakai, 
Rational surfaces associated with affine root systems 
and geometry of the Painlev\'e equations, 
{\it Commun. Math. Phys.}, {\bf 220} (2001), 165--229.

\bibitem{SAKT} 
\begin{CJK}{UTF8}{min}
佐藤幹夫, 青木貴史, 河合隆裕, 竹井義次, 特異摂動の代数解析 (金子晃記), 数理研講究録, 
{\bf 750} (1991), 43--51.
\end{CJK}

(M. Sato, T. Aoki, T. Kawai and Y. Takei, 
Algebraic analysis of singular perturbations 
(in Japanese; written by A. Kaneko), 
RIMS K\^oky\^uroku, {\bf 750} (1991), 43-51.)


\bibitem{Sauzin}
D. Sauzin,
Introduction to 1-summability and resurgence, 
in {\it Divergent Series, Summability and Resurgence I: 
Monodromy and Resurgence}, 
Lecture notes in mathematics, {\bf 2153}, 2016;
arXiv:1405.0356.

\bibitem{Strebel}
K. Strebel, Quadratic differentials, Springer Verlag, Berlin, 1984.

\bibitem{Sutherland1}
T. Sutherland, 
The modular curve as the space of stability conditions of a CY3 algebra, 
preprint; arXiv:1111.4184 [math.AG].

\bibitem{Sutherland2}
T. Sutherland, 
Stability conditions for Seiberg-Witten quivers,
PhD thesis, University of Sheffield.


\bibitem{Takei}
Y. Takei, 
An explicit description of the connection formula for the first 
Painlev\'e equation, In
{\it Toward the Exact WKB Analysis of Differential Equations, 
Linear or Non-Linear}, Kyoto Univ. Press, 2000, pp. 271--296.

\bibitem{Takei-Sato-conj}
Y. Takei, 
Sato's conjecture for the Weber equation and
transformation theory for Schr{\"o}dinger equations
with a merging pair of turning points,
{\it RIMS K\^okyur\^oku Bessatsu}, {\bf B10} (2008), pp. 205--224.

\bibitem{Voros83} 
A. Voros,
The return of the quartic oscillator
-- The complex WKB method,
{\it Ann. Inst. Henri Poincar\'e}, {\bf 39} (1983),
211--338.

\bibitem{WW}
E. T. Whittaker and G. N. Watson, 
A Course of Modern Analysis, 
Cambridge University Press, 1902.

\bibitem{Yoshida}
S. Yoshida,
2-Parameter Family of Solutions for Painlev\'e Equations 
(I) $\sim$ (V) at an Irregular Singular Point, 
{\it Funkcial. Ekvac.}, {\bf 28} (1985), 233--248.



\end{thebibliography}
\end{document}